\documentclass{iopart}
\usepackage{iopams}

\usepackage{hyperref}
\usepackage{graphicx}
\usepackage{mathptmx}
\usepackage[ngerman,american]{babel}
\usepackage{cite}
\usepackage{url}
\usepackage{textcomp}
\usepackage{wasysym}
\usepackage{amsthm}
\usepackage{rotating}
\usepackage{tabularx}
\usepackage{numprint}

\newcommand{\pcbond}{p_c^{\mbox{\scriptsize bond}}}
\newcommand{\pcsite}{p_c^{\mbox{\scriptsize site}}}

\newcommand{\dc}{d_{\mathrm{c}}}

\newcommand{\calH}{\mathcal{H}}

\newtheorem{theorem}{Theorem}
\newtheorem{lemma}{Lemma}
\newtheorem{definition}{Definition}

\allowdisplaybreaks

\begin{document}

\title{Series Expansion of the Percolation Threshold on Hypercubic Lattices}

\author{Stephan Mertens$^{1,2}$ and Cristopher Moore$^2$}

\address{\selectlanguage{ngerman}{$^1$Institut\ f"ur\ Physik,
    Otto-von-Guericke Universit"at, PF~4120, 39016 Magdeburg,
    Germany}} 

\address{$^2$Santa Fe Institute,
1399 Hyde Park Rd,
Santa Fe, NM 87501,
USA}

\ead{mertens@ovgu.de,moore@santafe.edu}

\begin{abstract}
  We study proper lattice animals for bond- and site-percolation on
  the hypercubic lattice $\mathbb{Z}^d$ to derive asymptotic series of
  the percolation threshold $p_c$ in $1/d$, The first few terms of these
  series were computed in the 1970s, but the series have not been
  extended since then. We add two more terms to the series for
  $\pcsite$ and one more term to the series for $\pcbond$, using a
  combination of brute-force enumeration, combinatorial identities and
  an approach based on Pad\'e approximants, which requires much fewer
  resources than the classical method.  We discuss why it took 40
  years to compute these terms, and what it would take to compute the
  next ones. En passant, we present new perimeter polynomials for site
  and bond percolation and numerical values for the growth rate of
  bond animals.
\end{abstract}




\section{Introduction}
\label{sec:intro}

Forty years ago, two papers appeared in this journal, 
each containing the first few terms of a remarkable series for the percolation thresold on the hypercube $\mathbb{Z}^d$. The first paper, by Gaunt, Sykes, and Ruskin~\cite{gaunt:sykes:ruskin:76}, presented a series expansion for the threshold for site percolation on the $d$-dimensional cubic lattice $\mathbb{Z}^d$, 
\begin{equation}
  \label{eq:pc-site-series-old}
  \pcsite(d) = \sigma^{-1} + \frac{3}{2}\sigma^{-2} 
  + \frac{15}{4}\sigma^{-3} + \frac{83}{4}\sigma^{-4} + O(\sigma^{-5})
\end{equation}
where $\sigma = 2d-1$.
The second paper, by Gaunt and Ruskin~\cite{gaunt:ruskin:78}, gave the corresponding series for 
bond percolation,
\begin{equation}
  \label{eq:pc-bond-series-old}
  \pcbond(d) = \sigma^{-1} + \frac{5}{2}\sigma^{-3} +
  \frac{15}{2}\sigma^{-4} + 57\sigma^{-5} + O(\sigma^{-6}) \, .
\end{equation}
It is known that the series expansion for $\pcbond(d)$ has rational coefficients to all orders~\cite{vanderhofstad:slade:05}, and the terms up through $\sigma^{-3}$ have been established rigorously~\cite{vanderhofstad:slade:06}. A crucial tool in these proofs is the lace expansion for the two-point connectivity function for percolation, which can also be used to prove that mean-field behavior takes over in sufficiently high dimensions (e.g.~\cite{hara:slade:95,fitzner:vanderhofstad:17} and references therein).
 
Despite decades of active research in percolation by physicists and
mathematicians, these series have not yet been extended. In this
contribution we extend both series by computing additional terms, 
namely the coefficients of $\sigma^{-6}$ for $\pcbond$, and the
coefficients of $\sigma^{-5}$ and $\sigma^{-6}$ for $\pcsite$.
To do all this we use enumerations of lattice animals that require 
a large amount of computation, new and recent analytical
results and finally a new method to derive the series from the
available data.
 
The paper is organized as follows. In
Section~\ref{sec:lattice-animals} we introduce lattice animals and
perimeter polynomials, focusing on bond animals. In
Section~\ref{sec:series-expansion-bond} we review the methods
of~\cite{gaunt:ruskin:78} that take us from perimeter polynomials to
the series expansion for $\pcbond$, and use our recent enumeration
results to derive the next term in the series. In
Section~\ref{sec:enumerations} we discuss the computational resources
needed to carry out these enumerations, and the fact that a purely
brute-force approach would take over a year given current
resources. In Section~\ref{sec:bond-percolation} we prove a set of
analytic results which let us avoid the most costly enumerations,
bringing this calculation within reach. We discuss site animals and
apply the same techniques to site percolation in
Section~\ref{sec:site-percolation}, using the methods of~\cite{gaunt:sykes:ruskin:76} to obtain the next term in the
series for $\pcsite$. In Section~\ref{sec:pade} we report on a new
method based on Pad\'e approximants to derive the series for $p_c$. 
We prove that the convergence of this method is equivalent to that of the 
classical approach of~\cite{gaunt:sykes:ruskin:76,gaunt:ruskin:78}. 
However, the Pad\'e method requires much less data, allowing us to add yet another term to the series
for $\pcsite$ using existing data. Finally, we conclude in
Section~\ref{sec:conclusion}, and discuss the challenge of pushing
these series even further.

\section{From Lattice Animals to Perimeter Polynomials}
\label{sec:lattice-animals}

The starting point of series expansions like
\eqref{eq:pc-site-series-old} and \eqref{eq:pc-bond-series-old} is the
counting of lattice animals.  We start with the animals relevant to bond percolation, 
and reserve the (somewhat simpler) discussion of site animals for Section~\ref{sec:site-percolation}.

A \emph{bond animal} is a connected set
of edges of the hypercubic lattice. Two animals are considered
distinct if they differ by a rotation or reflection, but not by
translation. We define the \emph{size} of a bond animal as the number
$e$ of edges in it. For instance, in the two-dimensional lattice there
are two bond animals of size $1$ and six of size $2$, namely,
\,\raisebox{-3pt}{\includegraphics[scale=0.1]{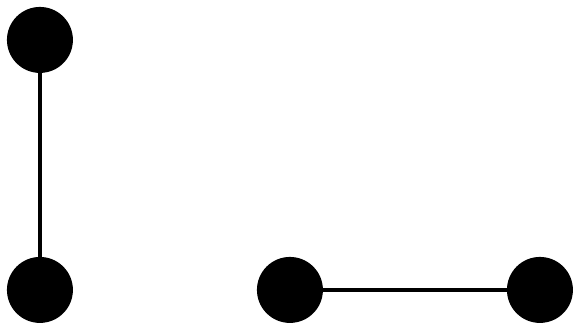}}\, and
\,\raisebox{-7pt}{\includegraphics[scale=0.1]{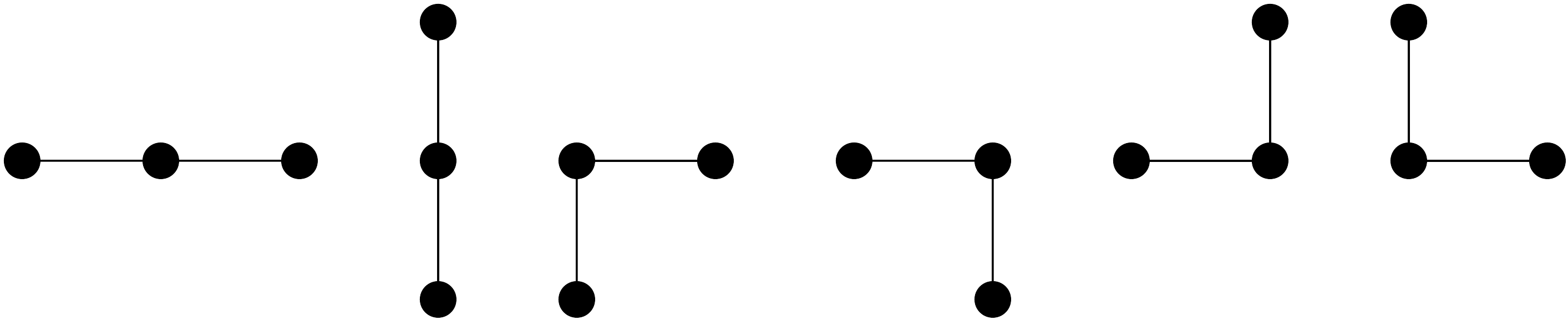}}\, .

Following Lunnon~\cite{lunnon:75}, we say an animal is \emph{proper in $k$ dimensions} if its edges span a $k$-dimensional subspace. For instance, the animals \,\raisebox{-7pt}{\includegraphics[scale=0.1]{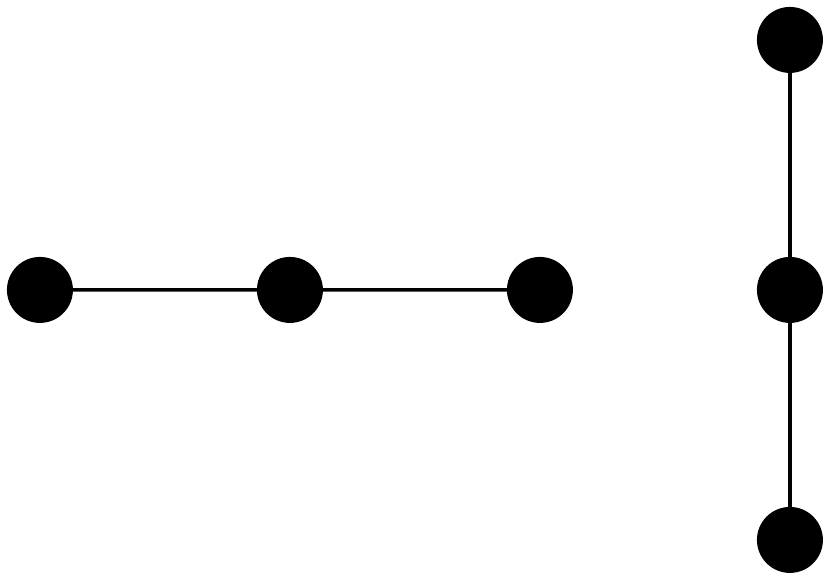}}\, are proper in $1$ dimension, and when projected into that dimension they are identical. We denote the number of animals of size $e$ proper in $k$ dimensions as $G_e^{(k)}$. Since such an animal can be embedded in the $d$-dimensional lattice in ${d \choose k}$ different ways, the total number of lattice animals of size $e$ in $d$ dimensions is
\begin{equation}
\label{eq:d-choose-k}
A_d(e) = \sum_{k=1}^e {d \choose k} G_e^{(k)} \, . 
\end{equation}
For instance, we have $G_2^{(1)} = 1$ and $G_2^{(2)} = 4$.

\begin{figure}
  \centering
  \includegraphics[width=0.75\linewidth]{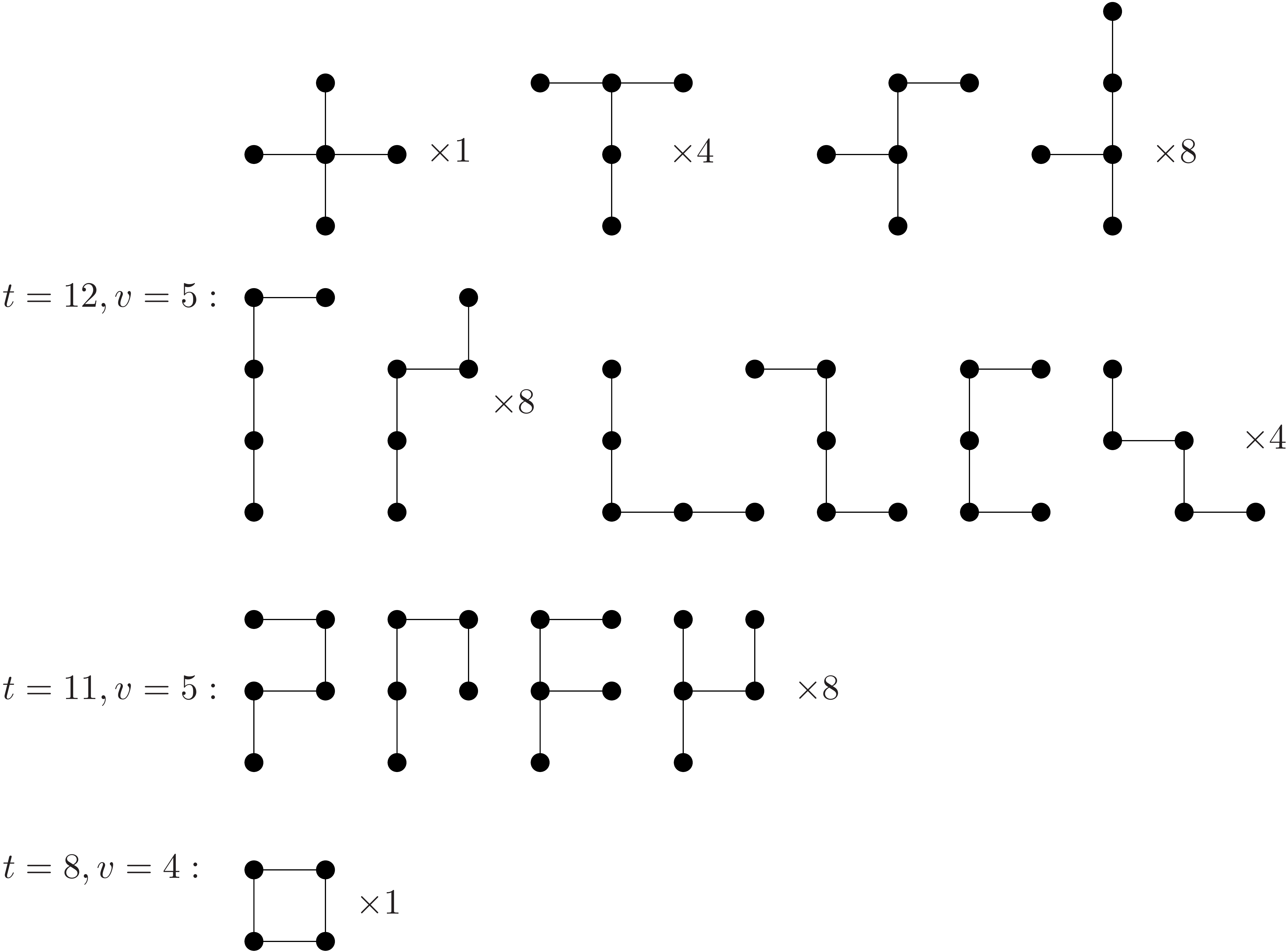}
  \caption{The $G_4^{(2)} = 86$ bond animals of size $e=4$ proper in
    $2$ dimensions, classified by their perimeter $t$ and the number
    of vertices $v$. There are $15$ different shapes up to rotation and reflection, and the notation $\times 8$ etc.\ gives the number of animals corresponding to each one.}
  \label{fig:proper-4}
\end{figure}

In order to compute the total probability that an edge belongs to a connected cluster of a given size, we also need to classify bond animals according to their perimeter, i.e., the number of edges that are incident to but not part of the cluster, and their number of vertices. Figure~\ref{fig:proper-4} shows all $86$ animals of size $e=4$ that are proper in $2$ dimensions, classified by their perimeter $t$ and number of vertices $v$. We denote the number of bond animals proper in dimension $k$ of size $e$, perimeter $t$ and number of vertices $v$ by $G_{e,t,v}^{(k)}$. From Figure~\ref{fig:proper-4} we see that
\[
  G^{(2)}_{4,12,5} = 53 \qquad G^{(2)}_{4,11,5} = 32 \qquad
  G^{(2)}_{4,8,4} = 1\,.
\]
Note that we define $t$ as the number of perimeter edges that live in the subspace spanned by the animal.  

Once we know the numbers $G_{e,t,v}^{(k)}$ for a given $e$, we can compute the perimeter polynomial
\begin{equation}
  \label{eq:def-perimeter-poly-bond}
  D_e(q) = \sum_{k=1}^e \sum_{t,v} {d \choose k} G^{(k)}_{e,t,v}\, q^{t + 2(d-k)v} \, .
\end{equation}
This arises by noting that when a bond animal proper in $k$ dimensions
is embedded in $d\geq k$ dimensions, 
it has $2(d-k)v$ additional perimeter edges pointing ``up'' and ``down'' in the other $d-k$ dimensions. 
Summing over all embeddings of such an animal in $d$ dimensions, and summing over all $k$, $t$, and $v$, 
gives~\eqref{eq:def-perimeter-poly-bond}. Comparing to~\eqref{eq:d-choose-k}, we see that $D_e(1) = A_d(e)$.

The perimeter polynomials allow us to express quantities such as the expected size $S$ of the cluster to which a random occupied edge belongs, i.e., conditioned on the event that that edge is occupied.  Following~\cite{gaunt:ruskin:78} we have
\begin{equation}
\label{eq:expected-size-bond}
S = \frac{1}{dp} \sum_e e^2 p^e D_e(1-p) \, . 
\end{equation}
To see this, focus on a particular edge in the lattice, say between $u=(0,0,\ldots,0)$ and $v=(1,0,\ldots,0)$. The probability that $(u,v)$ is occupied and belongs to a particular translation of a particular animal of size $e$ and total perimeter $t+2(d-k)v$ is $p^e (1-p)^{t+2(d-k)v}$. Summing over all translations and averaging over all rotations, $(u,v)$ can appear in each animal of size $e$ in $e/d$ different ways, since a fraction $d$ of the edges in all rotations lie along this axis.  Thus each animal counted by $G^{(k)}_{e,t,v}$ contributes $p^e (1-p)^{t + 2(d-k)v} e/d$ to the probability that $(u,v)$ is occupied and part of a cluster of size $e$, and $p^e (1-p)^t e^2/d$ to the expected size of the cluster. Since $S$ is the expected size conditioned on the event that $(u,v)$ is occupied, we divide by the probability $p$ of this event. Finally, summing over all $e$ gives~\eqref{eq:expected-size-bond}.


\section{From Perimeter Polynomials to the Series Expansion for the Threshold}
\label{sec:series-expansion-bond}

We now follow Gaunt and Ruskin~\cite{gaunt:ruskin:78} in using the perimeter polynomials, and their behavior for large $d$, to compute a high-dimension (or low-density) expansion for the critical density $\pcbond$. First we rewrite~\eqref{eq:expected-size-bond} slightly and expand the expected cluster size $S$ in powers of $p$, 
\begin{equation}
\label{eq:br-bond}
dpS = \sum_e e^2 p^e D_e(1-p) 
\triangleq \sum_{r=1}^\infty b_r p^r \, .
\end{equation}
(In~\cite{gaunt:ruskin:78} the authors write $b'_r$ instead of $b_r$.) Now suppose that $b_r$ grows exponentially in the limit of large $r$, 
\[
b_r \sim \mu^r \, ,
\]
where $\mu$ depends only on $d$.  In that case, the right-hand side of~\eqref{eq:br-bond}, and the expected cluster size $S$, diverge precisely when $p \ge \pcbond$ where
\[
\pcbond = 1/\mu \, . 
\]
Thus our goal is to compute
\begin{equation}
\label{eq:mu-limit}
\ln \mu = \lim_{r \to \infty} \frac{1}{r} \ln b_r \, . 
\end{equation}
and in particular its behavior for large $d$. In the limit $d \to \infty$, we expect $\mu$ to approach its mean-field value $\sigma = 2d-1$, the branching ratio of the lattice: that is, the number of children each edge would have if the lattice were a tree.

As we discuss below, through a mixture of brute-force enumeration and analytic results we obtain explicit formulas for $D_e(q)$ for $e \le 11$, which we exhibit in~\ref{sec:bond-polynomials}. Since $b_r$ depends only on $D_e$ for $e \le r$, these formulas allow us to compute $b_r$ for $r \le 11$ directly. We can extend this to $r=12$ by noticing that $b_r$ depends on $D_r$ only through  $D_r(1) = A_d(r)$, i.e., on the total number of animals in each dimension regardless of their perimeter, since the factor $(1-p)^t$ only contributes to higher-order terms.  Moreover, following~\cite{luther:mertens:11a}, we can deduce $D_r(1)$ from $D_e(q)$ for $e < r$ from the fact that the total probability of belonging to any animal is $p$:
\begin{equation}
\label{eq:luther-mertens-bond}
\sum_{e=1}^\infty e p^e D_e(1-p) = dp \, . 
\end{equation}
For any $r > 1$, the coefficient of $p^r$ in the left-hand side is zero. Thus we have
\begin{equation}
\label{eq:identity-bond}
D_r(1) = A_d(r) = -\frac{1}{r} \left[ \sum_{e=1}^{r-1} e p^e D_e(1-p) \right]_r \, ,
\end{equation}
where $[f(p)]_r$ denotes the coefficient of $p_r$ in the power series of $f(p)$. In~\ref{sec:bond-polynomials} we use this to derive $D_{12}(1)$.

\begin{figure}
  \centering
  \includegraphics[width=0.8\linewidth]{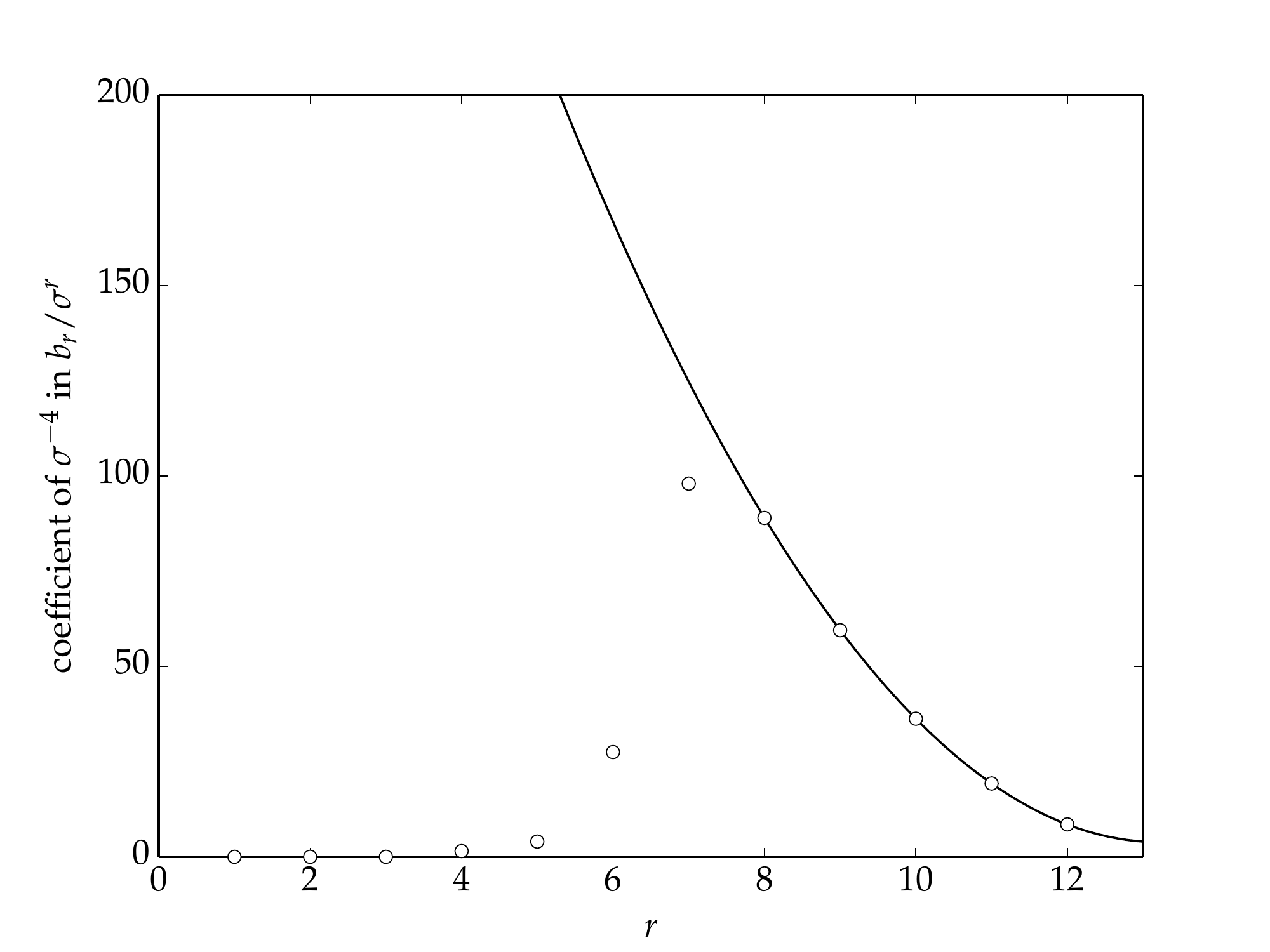}
  \caption{The stabilization of the coefficients in~\eqref{eq:br-polys-bond}. Here we show the coefficient of $\sigma^{-4}$, which coincides with a quadratic curve for $r \ge 8$.}
  \label{fig:b_r}
\end{figure}

Since $D_e$ is a degree-$e$ polynomial in $d$, it follows that $b_r$ is a degree-$r$ polynomial in $d$, or equivalently a degree-$r$ polynomial in $\sigma=2d-1$. We give these polynomials explicitly in~\ref{sec:br}, and for $r > 1$ the leading term $\sigma^r$ indeed matches the mean-field limit. If we expand $b_r / \sigma^r$ in powers of $\sigma^{-1}$, we find that the correction terms stabilize as $r$ increases. Specifically, as in~\cite{gaunt:ruskin:78}, we find that the coefficient of $\sigma^{-j}$ coincides with a polynomial in $r$ of degree $\lfloor j/2 \rfloor$ whenever $r \ge 2j$. We show this for the case $j=4$ in Figure~\ref{fig:b_r}. 

We extend the corresponding series in~\cite{gaunt:ruskin:78} by one term. Namely, we obtained the coefficient of $\sigma^{-5}$ by fitting a quadratic function of $r$ to just three points, the minimum number necessary, for $r=10,11,12$. This gives
\begin{eqnarray}
\frac{b_r}{\sigma^r} = 1 &+ \sigma^{-1} & \quad (r \ge 2) \nonumber \\
&+ \left( \frac{17}{2} - \frac{5 r}{2} \right) \sigma^{-2} & \quad (r \ge 4) \nonumber \\
&+ (57 - 10r) \sigma^{-3} & \quad (r \ge 6) \nonumber \\
&+ \left( 550-\frac{661 r}{8}+\frac{25 r^2}{8} \right) \sigma^{-4} & \quad (r \ge 8) \nonumber \\
&+ \left( \frac{\numprint{67933}}{12}-\frac{\numprint{15503} r}{24}+\frac{175 r^2}{8} \right) \sigma^{-5} & \quad (r \ge 10) \nonumber \\
&+ O(\sigma^{-6}) \, .
\label{eq:br-polys-bond}
\end{eqnarray}

Our confidence in~\eqref{eq:br-polys-bond} is bolstered by the fact that, if we assume $b_r$ has this form for sufficiently large $r$ and take its logarithm, the terms that are quadratic and higher-order in $r$ cancel, giving
\begin{align}
\ln b_r &= r \ln \sigma 
+ \sigma^{-1} 
+ \left( 8 - \frac{5r}{2} \right) \sigma^{-2}
+ \left( \frac{293}{6} - \frac{15 r}{2} \right) \sigma^{-3} 
\nonumber \\
&+ \left( \frac{3721}{8} - \frac{431 r}{8} \right) \sigma^{-4} 
+ \left( \frac{\numprint{71213}}{15} - \frac{2315 r}{6} \right) \sigma^{-5}
+ O(\sigma^{-6}) \, ,
\label{eq:log-br-bond}
\end{align}
where as in~\eqref{eq:br-polys-bond} the coefficient of $\sigma^{-j}$ holds for $r \ge 2j$.  
In particular, if the last coefficient of $\sigma^{-5}$ in~\eqref{eq:br-polys-bond} had a different linear or quadratic term in $r$, or if it contained any higher-order terms, then the coefficient of $\sigma^{-5}$ in~\eqref{eq:log-br-bond} would have higher-order terms in $r$. 

Proceeding from~\eqref{eq:log-br-bond} we now take the limit $r \to \infty$ defined in~\eqref{eq:mu-limit}, giving
\begin{equation}
\label{eq:mu-bond}
\ln \mu 
= \ln \sigma 
- \frac{5r}{2} \sigma^{-2}
- \frac{15 r}{2} \sigma^{-3} 
- \frac{431 r}{8} \sigma^{-4} 
- \frac{2315 r}{6} \sigma^{-5}
- O(\sigma^{-6}) \, . 
\end{equation}
Finally, exponentiating this gives
\begin{equation}
\label{eq:pc-bond-series-new}
\pcbond = \frac{1}{\mu} 
= \sigma^{-1}
+ \frac{5}{2} \sigma^{-3}
+ \frac{15}{2} \sigma^{-4}
+ 57 \sigma^{-5} 
+ \frac{4855}{12} \sigma^{-6} 
+ O(\sigma^{-7}) \, ,
\end{equation}
extending the series~\eqref{eq:pc-bond-series-old} from~\cite{gaunt:ruskin:78} by one more term. 

In the succeeding sections we will show how a combination of brute-force enumeration and new analytic results allowed us to obtain $D_1(q),\ldots,D_{11}(q)$.  Indeed, without these analytic results we would have required months of additional computation to obtain the new term in~\eqref{eq:pc-bond-series-new}.  



\section{Computer Enumerations of Bond Animals}
\label{sec:enumerations}

Computerized enumerations of lattice animals have a long tradition in
statistical mechanics. The first algorithm was published by Martin in
1974~\cite{martin:74}. The classical algorithm  for counting lattice
animals is due to Redelmeier~\cite{redelmeier:81}. 
Originally developed for the square lattice,
Redelmeier's algorithm was later shown to work on arbitrary
lattices and in higher dimensions~\cite{mertens:90} and to be
efficiently parallelizable~\cite{mertens:lautenbacher:92}.
For two dimensional lattices there is a much faster counting method based
on transfer matrices~\cite{jensen:01}, but for $d\geq 3$ Redelmeier's
algorithm is still the most efficient known way to count lattice animals. 

For counting in high dimensions one faces the problem of
storing a piece of the lattice large enough to hold all possible animals of a given size. 
The most naive approach of using a $d$-dimensional hypercube of side
length $2e+1$ for animals of size $e$ requires terabytes of memory 
for the dimensions and sizes we study here. We can reduce the memory requirements to megabytes 
by using a $d$-dimensional $\ell_1$ ball of radius $e$ instead, 
i.e., a generalized octahedron, since its volume is only $1/d!$ times that of the hypercube. 
Using this idea, extensive enumerations of site animals on high-dimensional lattices were performed in~\cite{luther:mertens:11a}. 

More generally, Redelmeier's algorithm takes the adjacency matrix of any graph as input and 
counts the animals on that graph. This lets us use the same approach for bond animals, 
since we can compute the adjacency matrix of the ``line graph'' or ``covering graph,'' 
whose vertices are the edges of the original ball and where two edges are adjacent if they share a vertex. 
The code to compute these graphs and the implementation of Redelmeier's algorithm can be downloaded from
our project website~\cite{animals:site}. For a detailed description of
the counting algorithm we refer to~\cite{luther:mertens:11a}. Here we
focus on the time complexity of the algorithm and how this limits
the size of the animals that we can count.

While its memory requirements are modest, Redelmeier's algorithm counts lattice animals of a given size by actually constructing all of them. Hence its running time scales essentially as $A_d(e)$, the total number of lattice animals of size $e$ in dimension $d$. These cluster numbers $A_d(e)$ grow asymptotically as
\begin{equation}
  \label{eq:scaling}
  A_d(e) \sim C \lambda_d^e\, e^{-\theta_d}
  \left(1+\frac{b}{e^{\Delta}} + \mbox{corrections} \right) \, . 
\end{equation}
Here the exponents $\theta_d$ and $\Delta$ are universal constants,
i.e., their value depends on the dimension $d$ but not on the
underlying lattice, while $C$, $b$, and the growth rate $\lambda_d$ are nonuniversal, 
lattice-dependent quantities~\cite{adler:etal:88}.  This universality
facilitates the computation of $\theta_d$ for some values of $d$ using
field-theoretic arguments. In particular, $\theta_2 = 1$ and $\theta_3=3/2$
\cite{parisi:sourlas:81,imbrie:03}, $\theta_4=11/6$ \cite{dhar:83} and
$\theta_d=5/2$ (the mean-field value) for $d\geq\dc=8$, the
critical dimension for animal growth \cite{lubensky:isaacson:79}.  

\begin{table}
  \centering
  \begin{tabular}{rrrrrr}
  $e$ & $d=2$ & $d=3$ & $d=4$ \\[1ex]
   13 & \numprint{79810756} & \numprint{208438845633} & \numprint{26980497086268} \\
14 & \numprint{386458826}  & \numprint{1979867655945} &
                                                        \numprint{384428067086544} \\
15 & \numprint{1880580352}  & \numprint{18948498050586} & \numprint{5527398761722192}\\
16 & \numprint{9190830700}  & \numprint{182549617674339} & \\
17 & \numprint{45088727820}  & \numprint{1768943859449895} & \\
18 & \numprint{221945045488}  & \numprint{17230208981859485} & \\
19 & \numprint{1095798917674}  & & \\
20 & \numprint{5424898610958}  & & \\
21 & \numprint{26922433371778}  & & \\
22 & \numprint{133906343014110}  & & \\
23 & \numprint{667370905196930}  & & \\
24 & \numprint{3332257266746004}  & &  
  \end{tabular}
  \caption{Number of bond lattice animals $A_d(e)$. We also know $A_5(13) =
  \numprint{900703198101845}$ and
  $A_6(13)=\numprint{14054816418877200}$. Values of $A_d(e)$
for $e\leq 12$ and general $d$ can be computed from
\eqref{eq:d-choose-k} and the perimeter polynomials in \ref{sec:bond-polynomials}.}
  \label{tab:Ad}
\end{table}

Values of $A_d(e)$ for $e \leq 12$ and general dimension $d$ can be
computed from \eqref{eq:d-choose-k} and the perimeter polynomials in
\ref{sec:bond-polynomials}. For some values of $d$ we have enumerated
animals larger than $e=12$, see Table~\ref{tab:Ad}. 
The enumeration data for $A_d(e)$ can be used to estimate both $\lambda_d$
and $\theta_d$.  For that we compute $\lambda_d(e)$ and $\theta_d(e)$ as the
solutions of the system
\begin{equation}
  \label{eq:lambda-theta-system}
  \ln A_d(e-k) = \ln C + (e-k)\ln\lambda_d(e) - \theta_d(e) \ln(e-k)
\end{equation}
for $k=0,1,2$.  We need three equations to eliminate the constant
$\ln C$. The growth rate $\lambda_d$ and exponent $\theta_d$ are
obtained by extrapolating the numbers $\lambda_d(e)$ and
$\theta_d(e)$ to $e\to\infty$. From \eqref{eq:scaling} we expect
that
\begin{equation}
  \label{eq:lambda-finite-size}
  \ln\lambda_d(e) \sim \ln\lambda_d + \frac{b}{e^{\Delta+1}} 
\end{equation}
for large values of $e$. We used the data points $\lambda_d(e)$ to fit the parameters $\ln\lambda_d$, $b$ and $\Delta$ in \eqref{eq:lambda-finite-size}. A plot of $\log\lambda_d(e)$ versus $e^{-(\Delta+1)}$ (Figure~\ref{fig:log-lambda}) then shows that the data points in fact scale like \eqref{eq:lambda-finite-size}. The resulting estimates for $\ln\lambda_d$ are listed in Table~\ref{tab:lambda}. The growth rate $\lambda_d$ increases linearly with $d$, see Figure~\ref{fig:lambda_d}.

\begin{figure}
  \centering
  \includegraphics[width=0.8\linewidth]{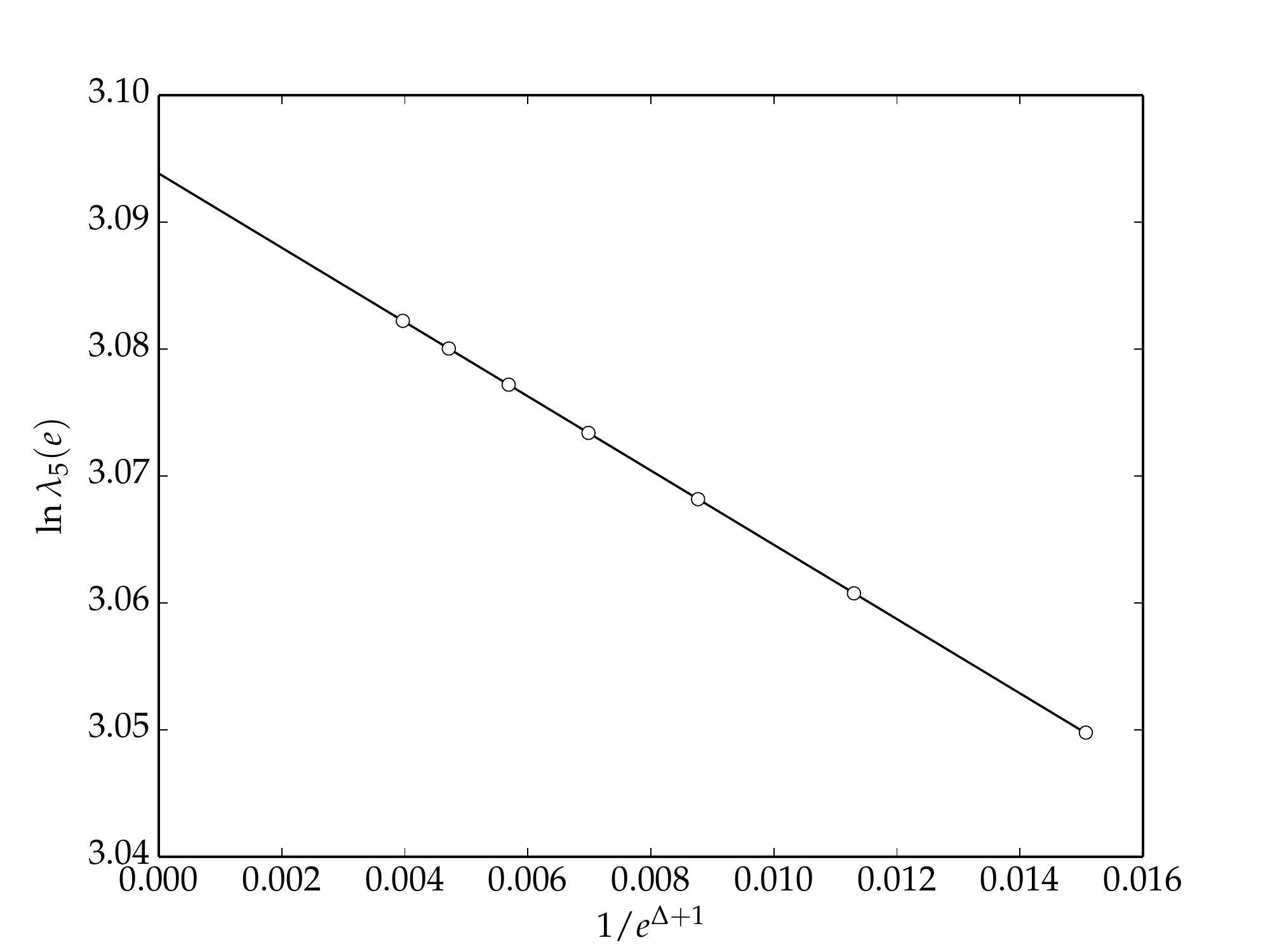}
  \caption{Bond animal growth rate $\lambda_d$ for $d=5$.}
  \label{fig:log-lambda}
\end{figure}

\begin{figure}
  \centering
  \includegraphics[width=0.8\linewidth]{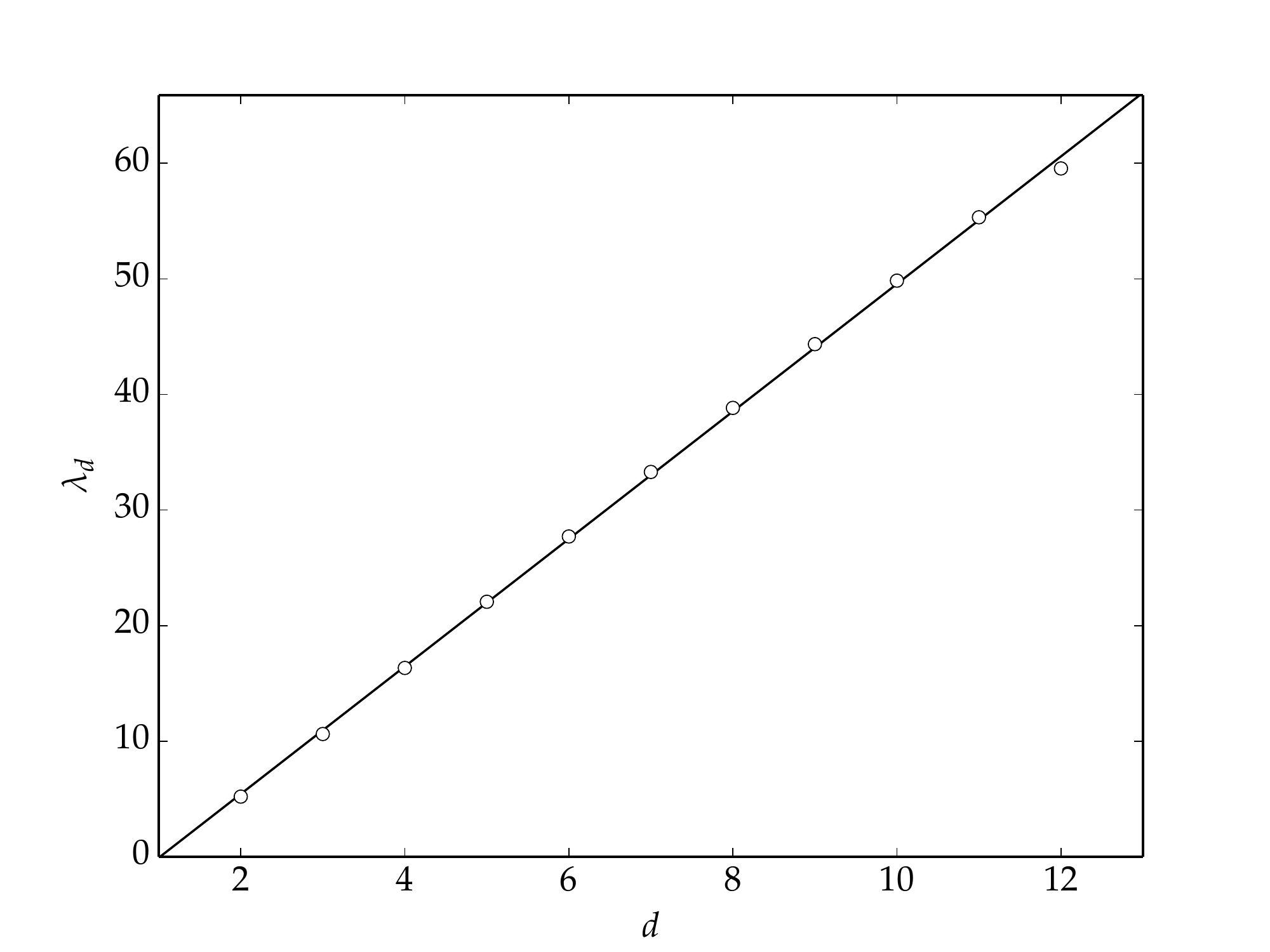}
  \caption{Growth rate $\lambda_d$ of the number of bond animals in
    $d$ dimensions.}
  \label{fig:lambda_d}
\end{figure}

The same approach can be used to compute the exponent $\theta_d$. Here
we expect
\begin{equation}
  \label{eq:theta-finite-size}
  \theta_d(e) \sim \theta_d + \frac{b}{e^{\Delta}}\,.
\end{equation}
Figure~\ref{fig:theta} shows that $\theta_d(e)$ in fact scales like
\eqref{eq:theta-finite-size}. The resulting estimates for $\theta_d$
(Table~\ref{tab:lambda}) deviate from the Monte Carlo results and the 
exact values by no more than 5\%, a deviation most likely 
induced by the extrapolation $e\to\infty$.

\begin{figure}
  \centering
  \includegraphics[width=0.8\linewidth]{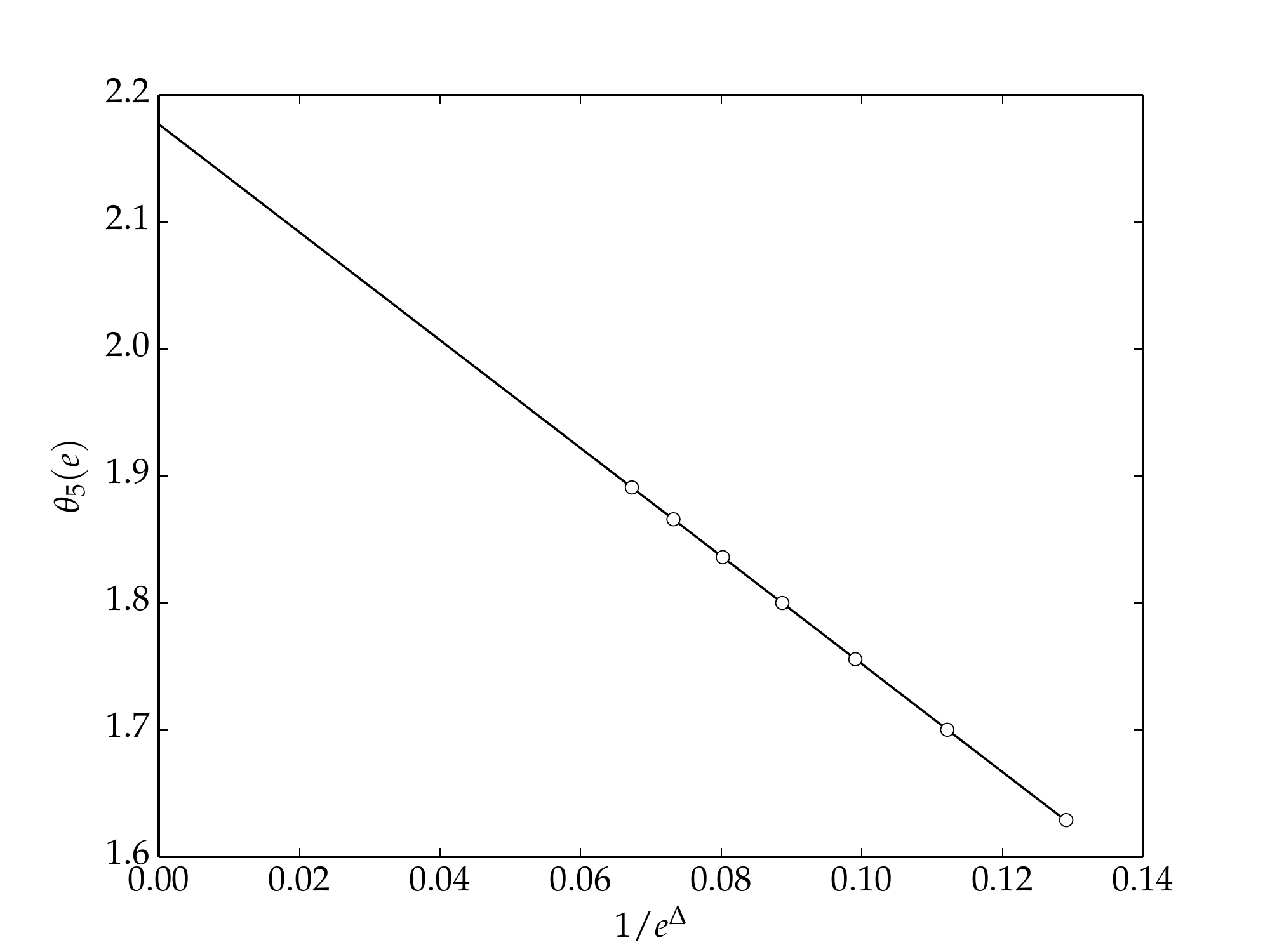}
  \caption{Exponent $\theta_d$ for $d=5$.}
  \label{fig:theta}
\end{figure}

\begin{table}
  \centering
  \begin{tabular}{c|l|cc}
    & \rule[-1.5ex]{0pt}{3ex}$\log \lambda_d$ &
    \multicolumn{2}{c}{$\Theta_d$} \\
    $d$ & enum. &  enum. & exact, MC \\\hline
      2 & \rule{0pt}{2.5ex}1.650 & 0.99 & 1 \\
      3 & 2.362 & 1.52 & $3/2$ \\
      4 & 2.793 & 1.90 & $11/6$ \\
      5 & 3.093 & 2.18 & 2.080(7) \\
      6 & 3.322 & 2.34 & \phantom{1}2.261(12)\\
      7 & 3.505 & 2.42 & 2.40(2)\phantom{1} \\
      8 & 3.659 & 2.47 & $5/2$ \\
      9 & 3.792 & 2.49 & $5/2$ \\
    10 & 3.909 & 2.50 & $5/2$ \\
    11 & 4.013 & 2.51 & $5/2$ \\
    12 & 4.107 & 2.52 & $5/2$ 
  \end{tabular}
  \caption{Growth rates $\lambda_d$ and exponents $\theta_d$ obtained
    from extrapolating the enumeration data. The column marked ``exact,
    MC'' contains exact values from field-theoretic arguments or values from large scale
    Monte Carlo simulations \cite{hsu:nadler:grassberger:05,hsu:nadler:grassberger:05a}.}
  \label{tab:lambda}
\end{table}

\begin{table}
  \centering
  \begin{tabular}{lccc}
  CPU & frequency & nodes$\times$cpus$\times$cores & memory/core \\[0.5ex]
  E5-1620 & 3.60 GHz & $1\times 2\times 4$ & 4.0 GByte \\
  E5-2630 & 2.30 GHz & $5\times 4\times 6$  & 5.3 GByte \\
  E5-2630v2 & 2.60 GHz &  $5\times 4\times 6$ & 5.3 GByte \\
  E5-2640v4 & 2.40 Ghz & $3\times 4 \times 10$ & 6.4 GByte
  \end{tabular}
  \caption{Computing machinery used for the enumerations in this
    paper. All CPUs are Intel\textsuperscript{\textregistered} Xeon\textsuperscript{\textregistered}.}
  \label{tab:leonardo}
\end{table}

Now, for practical enumerations we have to face the fact that the running time grows exponentially with the size of the animals, and that the associated growth rate $\lambda_d$ is an increasing function of $d$. Hence the hardest enumeration tasks are those with large size $e$ in high dimension $d$. In particular, to compute the next term of the series~\eqref{eq:pc-bond-series-old}, we have to enumerate animals with $e=11$ in $d=8,9,10,11$ while keeping track of their perimeter and number of vertices, i.e., compute $G_{e,t,v}^{(d)}$ for all $t$ and $v$. On our Linux cluster with 368 cores (Table~\ref{tab:leonardo}), doing this for $e=11$ in $d=8$ took 12 days and 6 hours wall clock time. 
The running times for the next few values of $d$ can be estimated by extrapolation as
\begin{eqnarray*}
  (d=9,e=11) &\rightarrow& \phantom{1}47\, \text{days,}\\
  (d=10,e=11) &\rightarrow& 181\, \text{days,}\\
  (d=11,e=11) &\rightarrow& 624\, \text{days.} 
\end{eqnarray*}
Thus $d=9$ would take about $7$ weeks, while $d=10$ would take months and $d=11$ would take over a year. Luckily we can spare ourselves all three of these tasks by computing the corresponding perimeter polynomials analytically---that is, by providing explicit formulas for $G_{e,t,v}^{(e)}$, $G_{e,t,v}^{(e-1)}$, and $G_{e,t,v}^{(e-2)}$. We do this in the next section.

\section{Analytic Results for Bond Animals}
\label{sec:bond-percolation}

\subsection{Counting proper and almost-proper bond animals}

Computing $G_e$ is computationally expensive, with effort that grows exponentially as a function of $e$ with a base that grows linearly in $d$. We can save time by deriving some analytic results. In this section we prove explicit formulas for
$G_{e}^{(e)}$, $G_{e}^{(e-1)}$, and $G_{e}^{(e-2)}$, which appeared without proofs or derivations in~\cite{gaunt:ruskin:78}, and for the coefficients of ${d \choose e}$, ${d \choose {e-1}}$ and ${d \choose {e-2}}$ in $D_e(q)$. Thanks to these formulas we can compute the next term of $\pcbond$ using enumerations of bond animals just for $e \leq 11$ and $d \leq 8$, which as discussed in the previous section greatly reduces our total computation time. 


To count animals, we will relate them to labeled trees of various kinds. An \emph{edge-labeled tree} is a tree with $e$ edges where each one is given a unique label $\{1,\ldots,e\}$, and a \emph{vertex-labeled tree} is a tree with $n$ vertices where each one is given a unique label in $\{1,\ldots,n\}$.  A \emph{rooted tree} is one where a particular vertex is distinguished as the root.

There is a one-to-one correspondence between rooted edge-labeled trees with $e \ge 2$ edges and vertex-labeled trees with $n=e+1$ vertices. We direct the edges away from the root, copy each edge label to the vertex it points to, and give the root vertex the label $n$.  Since any of the $n$ vertices can be treated as the root, and since by Cayley's formula~\cite{cayley:1889} the number of vertex-labeled trees is $n^{n-2}$, the number of edge-labeled trees with $e \ge 2$ edges and $n=e+1$ vertices is~\cite{cameron:95,barequet:barequet:rote:10,luther:mertens:17}
\begin{equation}
\label{eq:edge-labeled-trees}
\mbox{\# edge-labeled trees} = n^{n-3} = (e+1)^{e-2} \, . 
\end{equation}
Note that this gives the nonsense answer $1/2$ when $e=1$: since the two endpoints of the graph with a single edge are indistinguishable, making either one the root leads to the same vertex-labeled graph.  For similar reasons, we will find that some of our formulas will only work when $e$ is sufficiently large. On the other hand, some conveniently give the right answer for all $e \ge 1$, in which case we will state them without qualification.


\begin{theorem}
  \label{thm:proper-e}
  The number of bond animals of size $e$ that are proper in $e$ dimensions is
  \begin{equation}
    \label{eq:proper-e}
     G_e^{(e)} = 2^e\,(e+1)^{e-2} \, .
  \end{equation}
\end{theorem}

\begin{proof}
In a directed tree, each edge can have two orientations. Hence
the right hand side of \eqref{eq:proper-e} is the number of
directed edge labeled trees with $e$ bonds. To prove Theorem~\ref{thm:proper-e} we
need to show that there is a one-to-one correspondence between
these trees and bond animals of size $e$ that are
proper in $e$ dimensions.  

Any bond animal of size $e$ that is proper in $e$ dimensions is a
tree: there can be no loops since each edge has to point along its own
dimensional axis. Thus it corresponds to an edge-labeled tree, where
each edge is uniquely labeled with the axis it points along. For each
axis there are two possibilities with respect to orientation: it can
point either ``up'' or ``down'' along that coordinate axis, and these possibilities correspond to
the direction of the edge.  Conversely, we can read every directed,
edge labeled tree as a blueprint for a proper animal: take any vertex
of the tree as the initial site of the animal and and add the tree
neighbors of that vertex to the animal as indicated by the label and
the direction of the corresponding edge. Proceed with the neighbors of
the neighbors etc. The result is an animal that
is proper in $e$ dimensions since we have spanned each dimension
exactly once. Thus the mapping between directed edge-labeled trees and
bond animals proper in $e$ dimensions is one-to-one, which proves \eqref{eq:proper-e}.

\end{proof}

Next we prove two lemmas which are helpful in counting edge-labeled trees that contain certain labeled subgraphs.

\begin{lemma}
 \label{lem:forests}
  The number of ordered forests of $k \ge 1$ rooted trees with a
  total of $n$ vertices, $n-k$ edges and distinct edge labels $1,\ldots,n-k$ is
  $k n^{n-k-1}$.
\end{lemma}

\begin{proof}
This is proved in~\cite[Lemma 4]{barequet:barequet:rote:10}. We present a slightly modified proof. By merging the roots of the trees into a single vertex $v$, we obtain an edge-labeled tree with $n-k+1$ vertices. As described above, this corresponds to a vertex-labeled tree where $v$ is labeled $n-k+1$. This map is not one-to-one; if $v$ has degree $\ell$, then there are $k^\ell$ ordered forests of $k$ trees that would have mapped to this tree. 

A generalization of Cayley's formula that follows from Pr\"ufer codes~\cite{moon} states that the number of vertex-labeled trees with $n$ vertices where, for each $1 \le i \le n$, the vertex labeled $i$ has degree $d_i$ is the multinomial ${ n-2 \choose d_1-1, \cdots, d_n-1 }$. Thus the total number of trees with $n-k+1$ vertices where the vertex labeled $n-k+1$ has degree $\ell$ is
\begin{align}
&\sum_{\substack{d_1,\ldots,d_{n-k}: \\ \sum_{i=1}^{n-k} (d_i-1) = n-k-\ell}}  
  { n-k-1 \choose d_1-1, \cdots, d_{n-k}-1, \ell-1 } 
  \nonumber \\
&= { n-k-1 \choose \ell-1 } 
  \sum_{\substack{e_1,\ldots,e_{n-k}: \\ \sum_{i=1}^n e_i = n-k-\ell}} 
  { n-k-\ell \choose e_1, \cdots, e_{n-k} } 
  \nonumber \\
&= { n-k-1 \choose \ell-1 } (n-k)^{n-k-\ell} \, . 
\label{eq:degree-ell}
\end{align}
Multiplying by the number of forests $k^\ell$ that map to each such tree and summing over $\ell$ gives
\begin{align*}
&\sum_{\ell=1}^{n-k} { n-k-1 \choose \ell-1 } (n-k)^{n-k-\ell} k^\ell \\
&= k \sum_{\ell'=0}^{n-k-1} { n-k-1 \choose \ell' } (n-k)^{n-k-1-\ell'} k^{\ell'} \\
&= k n^{n-k-1} \, . 
\end{align*}
Note that this formula correctly gives $1$ when $k=n$, i.e., when the forest consists of $k$ vertices and no edges.
%
\end{proof}

\begin{definition}
\label{def:animal} 
Let $\calH = \{ H_1, \ldots, H_m \}$ be a collection of subgraphs where $H_i$ consists of $k_i$ vertices and $e_i$ edges for each $1 \le i \le m$.  Define an \emph{animal of size $e$ containing $\calH$} as a decorated graph $G$ with $e$ edges and the following properties:
\begin{itemize}
\item $G$ contains one copy of each $H_i$, and these copies are vertex-disjoint.
\item For each $i$, the vertices in the copy of $H_i$ are labeled $(i,j)$ where $1 \le j \le k_i$ to identify which $H_i$ they belong to, and to which vertex of $H_i$ they correspond.
\item If we contract each $H_i$ to form a single vertex, the resulting graph $G'$ (which has $n' = n+m-\sum_{i=1}^m k_i$ vertices and $e' = e-\sum_{i=1}^m e_i$ edges) is a tree.
\item Finally, the edges of $G'$ are given distinct labels $1,\ldots,e'$ and each one is directed.
\end{itemize}
\end{definition}
\noindent
Note that the $H_i$ themselves are not necessarily trees, so we have $n' = e'+1$ but not necessarily $n=e+1$.

\begin{lemma}
\label{lem:animals}
Let $\calH = \{ H_1, \ldots, H_m \}$ be defined as above.  Then the number of animals $G$ of size $e$ containing $\calH$ is 
\begin{equation}
\label{eq:animals}
2^{e-\sum_{i=1}^m e_i} 
\left( \,\prod_{i=1}^m k_i \right) 
\frac{(n+m-1-\sum_{i=1}^m k_i)!}{(n-\sum_{i=1}^m k_i)!} 
\,n^{n+m-2-\sum_{i=1}^m k_i} 
\, .
\end{equation}
\end{lemma}

\begin{proof}
In the proof of Lemma~\ref{lem:forests}, we merged the roots of the $k$ rooted trees to form a single vertex, and considered vertex-labeled trees where that vertex $\ell$ had a given degree. This is just a more elaborate version of the same idea. Indeed, that Lemma corresponds to the special case of this one where $m=1$.  

Each $H_i$ becomes a vertex $v_i$ with degree $\ell_i$.  Let $n'$ be defined as in Definition~\ref{def:animal} and write $t = n'-2-\sum_{i=1}^m (\ell_i-1)$.  Again invoking the generalization of Cayley's formula, the number of vertex-labeled trees with $n'$ vertices where for each $1 \le i \le m$ the vertex labeled $i$ has degree $\ell_i$ is
\begin{align*}
& \sum_{\substack{d_{m+1},\ldots,d_{n'}: \\ \sum_{i=m+1}^{n'} (d_i-1) = t}}
  {n'-2 \choose \ell_1-1, \ldots, \ell_m-1, d_{m+1}-1, \ldots, d_{n'}-1}  \\
&= {n'-2 \choose \ell_1-1, \ldots, \ell_m-1, t} 
  \sum_{\substack{e_{m+1},\ldots,e_{n'}: \\ \sum_{i=m+1}^{n'} e_i = t}} 
  {t \choose e_{m+1}, \ldots, e_{n'}}  \\
&= {n'-2 \choose \ell_1-1, \ldots, \ell_m-1, t} \left( n - \sum_{i=1}^m k_i \right)^{\!t} \, . 
\end{align*}
For each such tree, there are $\prod_{i=1}^m {k_i}^{\ell_i}$ ways to assign the edges of each $v_i$ to the $k_i$ vertices of $H_i$. Summing over the $\ell_i$ then gives
\begin{align*}
& \sum_{\substack{\ell_1,\ldots,\ell_m,t: \\ \sum_{i=1}^m (\ell_i-1) = n'-2-t}} 
  {n'-2 \choose \ell_1-1, \ldots, \ell_m-1, t} \left( n - \sum_{i=1}^m k_i \right)^{\!t} \prod_{i=1}^m {k_i}^{\ell_i} \\
&= \left( \,\prod_{i=1}^m k_i \right) \;\;\;\times\!\!\!
  \sum_{\substack{e_1,\ldots,e_m,t: \\ \sum_{i=1}^m e_i = n'-2-t}} 
  {n'-2 \choose e_1, \ldots, e_m, t} \left( n - \sum_{i=1}^m k_i \right)^{\!t} \prod_{i=1}^m {k_i}^{e_i} \\
&= \left( \,\prod_{i=1}^m k_i \right) n^{n'-2} \, .
\end{align*}

So far we have counted the number of ways to include $\calH$ in $G'$ where the vertices of $G'$ are labeled $1,\ldots,n'$ and the vertex $v_i$ corresponding to $H_i$ is labeled $i$ for each $1 \le i \le m$. To convert these to edge labels for $G'$, we first multiply by $n' (n'-1) (n'-2) \cdots (n'-m+1) = n'! / (n'-m)!$ so that the vertices of $G'$, including the $v_i$, can have any distinct labels.  We then declare the vertex $w$ labeled $n'$ the root, orient the edges of $G'$ away from it, and copy the label of edge vertex onto its incoming edge.  This makes $G'$ an edge-labeled graph rooted at $w$, so we divide by the $n'$ choices of $w$.  Finally, we multiply by the $2^{e'}$ possible directions on the edges of $G'$. This gives
\begin{align*}
&2^{e'} \left( \,\prod_{i=1}^m k_i \right) \frac{(n'-1)!}{(n'-m)!} \,n^{n'-2}  \\
&= 2^{e-\sum_{i=1}^m e_i}
\left( \,\prod_{i=1}^m k_i \right) 
\frac{(n+m-1-\sum_{i=1}^m k_i)!}{(n-\sum_{i=1}^m k_i)!} 
\,n^{n+m-2-\sum_{i=1}^m k_i} 
\end{align*}
and completes the proof.
\end{proof}

We now obtain formulas for $G_e^{(e-1)}$ and $G_e^{(e-2)}$.  We can do this using a kind of duality, where animals that are proper in $k < e$ dimensions contain certain small collections of animals which are proper in $k$ dimensions.

\begin{theorem}
  \label{thm:proper-e-1}
  The number of bond animals of size $e \ge 2$ that are proper in $e-1$ dimensions is
  \begin{equation}
    \label{eq:proper-e-1}
     G_e^{(e-1)} = 2^{e-2} (e-1) (2e-1) (e+1)^{e-3} \, .
  \end{equation}
\end{theorem}

\begin{proof}
Since one can not form a loop with just two edges along the same axis, a bond animal of size $e$ that is proper in $e-1$ dimensions is still a directed, edge-labeled tree, but with only $e-1$ distinct edge labels so that one label $x$ appears on two edges.

We will treat the edges labeled $x$ as a collection of subgraphs $\calH$ and apply Lemma~\ref{lem:animals}.  There are two cases.  First, if these edges are vertex-disjoint, then $\calH$ is a pair of subgraphs each consisting of two vertices connected by a directed edge.  Lemma~\ref{lem:animals} with $m=2$, $k_1=k_2=2$, and $e_1=e_2=1$ gives
\[
2^{e-2} \times 4 (n-3) n^{n-4} 
= 2^e (e-2) (e+1)^{e-3}   \, .
\]
This is the number of ways to give the $e-2$ edges outside $\calH$ distinct labels. However, there are $e-1$ choices of the duplicate label $x$, so we need multiply this by $e-1$. Finally, Lemma~\ref{lem:animals} assumes that the two subgraphs $H_1$ and $H_2$ are distinguishable, but since they are identical we divide by $2$.  Thus the number of directed edge-labeled trees with one duplicate label, where the edges with that label do not share a vertex, is
\begin{equation}
\label{eq:e-1-separate}
2^{e-1} (e-1)(e-2) (e+1)^{e-3} \, .
\end{equation}
For $e=3$, for instance, this counts the $4$ animals \,\raisebox{-3pt}{\includegraphics[scale=0.12]{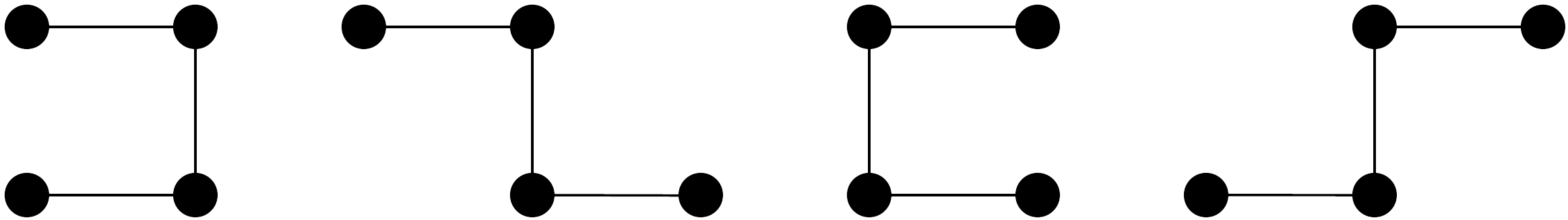}}\, which have two vertex-disjoint horizontal edges, plus another $4$ with two vertical edges.

In the case where the two edges labeled $x$ share a vertex, the orientations $\bullet\stackrel{x}{\longrightarrow}\bullet\stackrel{x}{\longleftarrow}\bullet$ and $\bullet\stackrel{x}{\longleftarrow}\bullet\stackrel{x}{\longrightarrow}\bullet$ are forbidden. This is because these two bonds would overlap in the lattice: the two outer vertices would have the same displacement from the center vertex. Thus we can take $\calH$ to be a single graph $H_1 = \bullet\stackrel{x}{\longrightarrow}\bullet\stackrel{x}{\longrightarrow}\bullet$.  Lemma~\ref{lem:animals} with $m=1$, $k_1=3$, and $e_1=2$, or equivalently Lemma~\ref{lem:forests} with $k=3$, then gives 
\[
2^{e-2} \times 3 n^{n-4} = 2^{e-2} \times 3 (e+1)^{e-3} \, . 
\]
We again multiply by the $e-1$ choices of $x$, giving
\begin{equation}
\label{eq:e-1-together}
2^{e-2} \times 3 (e-1) (e+1)^{e-3}
\end{equation}
for the number of directed edge-labeled trees with a duplicate pair of edges forming a directed path of length $2$.  For $e=3$, for instance, this counts the $6$ animals with two joined horizontal edges \,\raisebox{-3pt}{\includegraphics[scale=0.12]{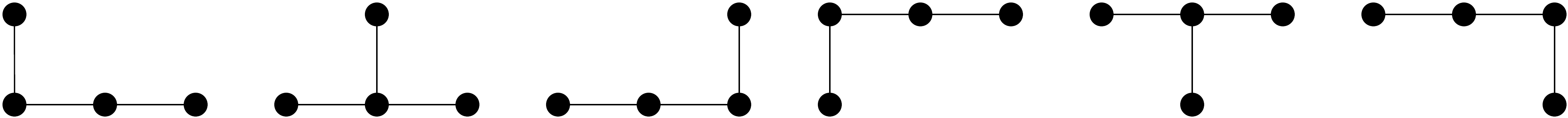}}\, and another $6$ with two joined vertical edges.  
Adding~\eqref{eq:e-1-separate} and~\eqref{eq:e-1-together} gives~\eqref{eq:proper-e-1}.
\end{proof}

\begin{theorem}
  \label{thm:proper-e-2}
  The number of bond animals of size $e$ that are proper in $e-2$
  dimensions is
  \begin{eqnarray}
    \label{eq:proper-e-2}
    \fl
     G_e^{(e-2)} &=& 2^{e-3} (e-2)(e-3) e^{e-5} \nonumber\\\fl
     & &+ \frac{1}{3} 2^{e-5} (e-2) \big(12e^4-20e^3-33e^2-46e+195\big) (e+1)^{e-5} \, .
  \end{eqnarray}
\end{theorem}

\begin{proof}
 We first note that there are two types of animals proper in $e-2$ dimensions: those where two distinct edge labels are duplicated, and those where one label appears three times.  In the first case, we multiply the results of Lemma~\ref{lem:animals} by the ${e-2 \choose 2}$ choices of the duplicate labels; in the second, we multiply by the $e-2$ choices of the triplicate one. 

We show the contribution from various subgraph collections $\calH$ in Table~\ref{tab:e-2}, keeping track of how many images they have under symmetry transformations. When two or more of the $H_i$ are identical, we also divide by the number of permutations in $S_m$ that preserve $\calH$.

The only real difference from the calculation in Theorem~\ref{thm:proper-e-1} is that now, for the first time, we have animals that contain a loop of size $4$ (the top row of Table~\ref{tab:e-2}).  In this case, we have $n=e$ rather than $n=e+1$, so this contribution is
\begin{equation}
\label{eq:square-loop}
{e-2 \choose 2} 2^{e-4} \times 4 n^{n-5} 
= 2^{e-3} (e-2)(e-3) e^{e-5} \, . 
\end{equation}
This is the first term in~\eqref{eq:proper-e-2}.  Adding the other contributions shown in Table~\ref{tab:e-2} and simplifying gives the second term.
\end{proof}

\begin{table}
\arraycolsep=2pt
\def\arraystretch{1.4}
$
\begin{array}{|c|c|c|} \hline
\calH & \text{parameters for Lemma~\ref{lem:animals}} & \text{contribution to $G_e^{(e-2)}$ for $e \ge 2$} \\ \hline 
\includegraphics[scale=0.12]{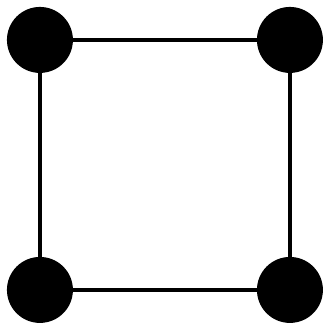} 
& m=1, k_1 = 4, e_1 = 4
& 2^{e-3} (e-2)(e-3) e^{e-5} 
\\ \hline
\begin{array}{c}
\raisebox{-8pt}{\includegraphics[scale=0.12]{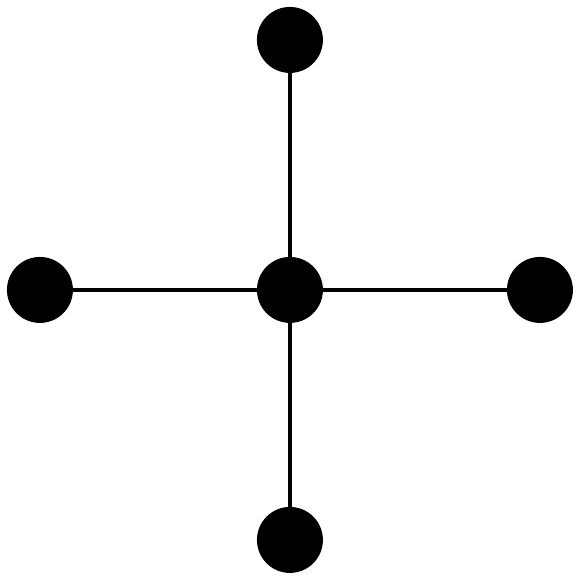}} \\
\raisebox{-8pt}{\includegraphics[scale=0.12]{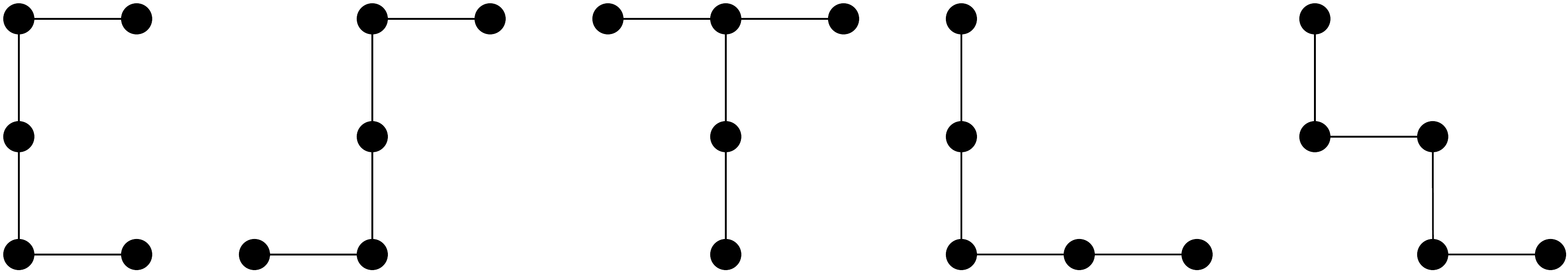}}\, \times 4 \\
\raisebox{-8pt}{\includegraphics[scale=0.12]{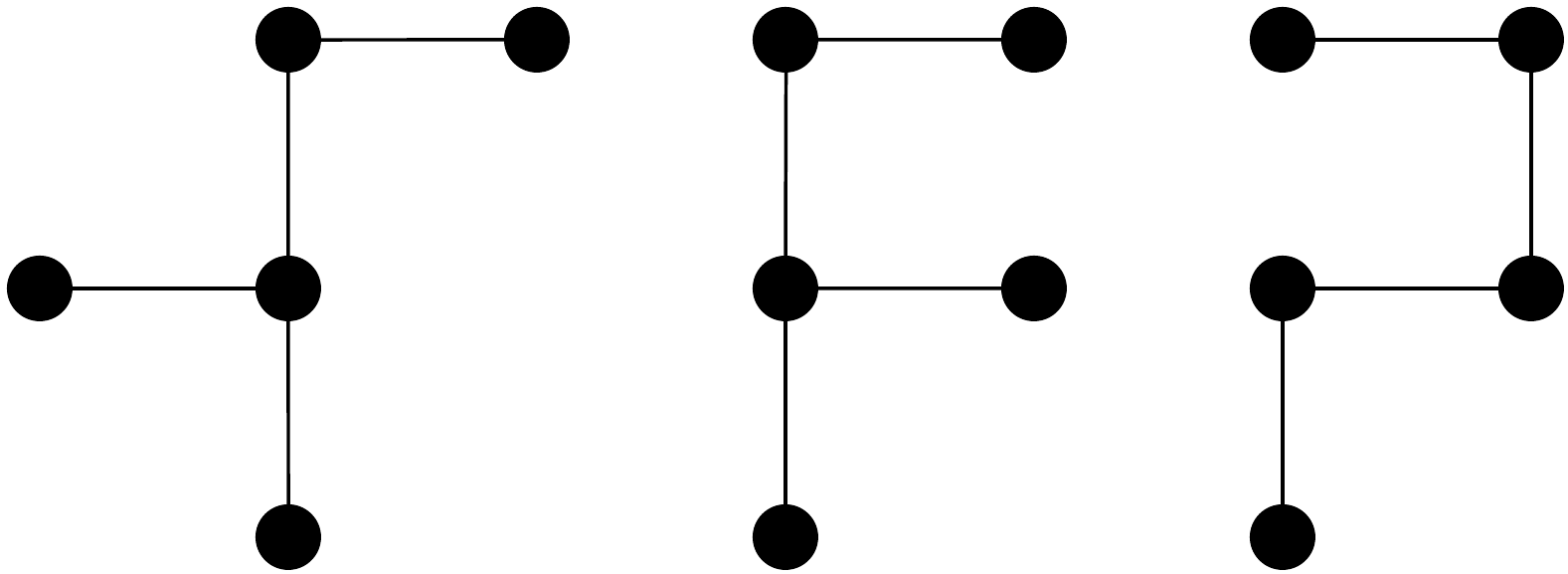}}\, \times 8 
\end{array}
& 
m=1, k_1 = 5, e_1 = 4 
& \begin{array}{c}
45 \times 5 \times \\ 
2^{e-5} (e-2)(e-3) (e+1)^{e-5} 
\end{array}
\\ \hline
\begin{array}{c} 
\raisebox{-12pt}{\includegraphics[scale=0.12]{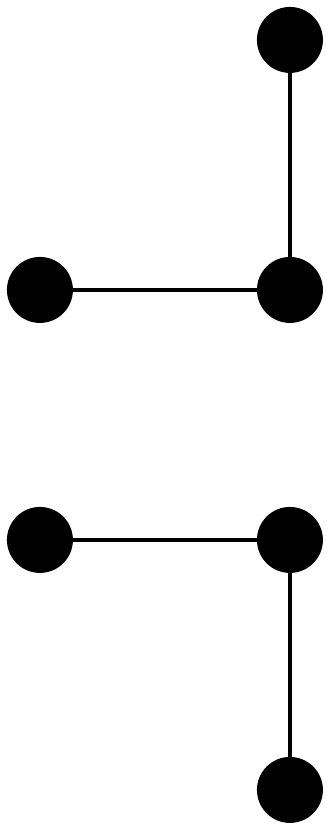}}\, \times 4, \;
\raisebox{-12pt}{\includegraphics[scale=0.12]{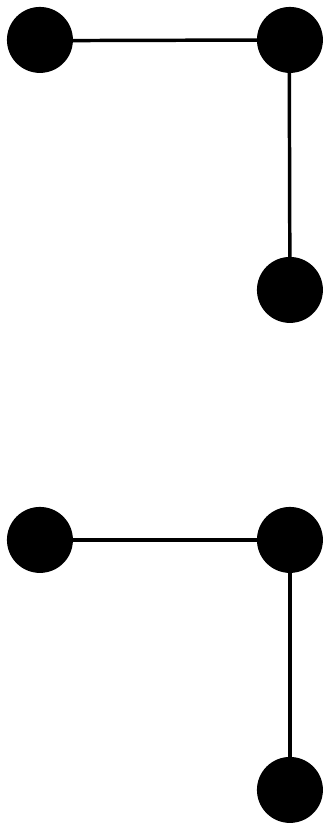}}\, \times 4 \times \frac{1}{2}, \;
\raisebox{-12pt}{\includegraphics[scale=0.12]{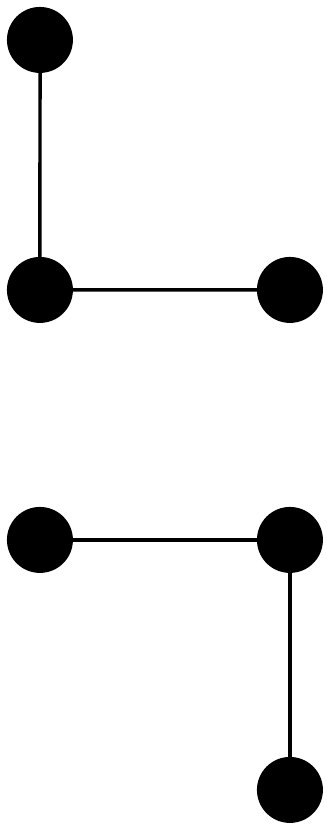}}\, \times 2, \;
\raisebox{-12pt}{\includegraphics[scale=0.12]{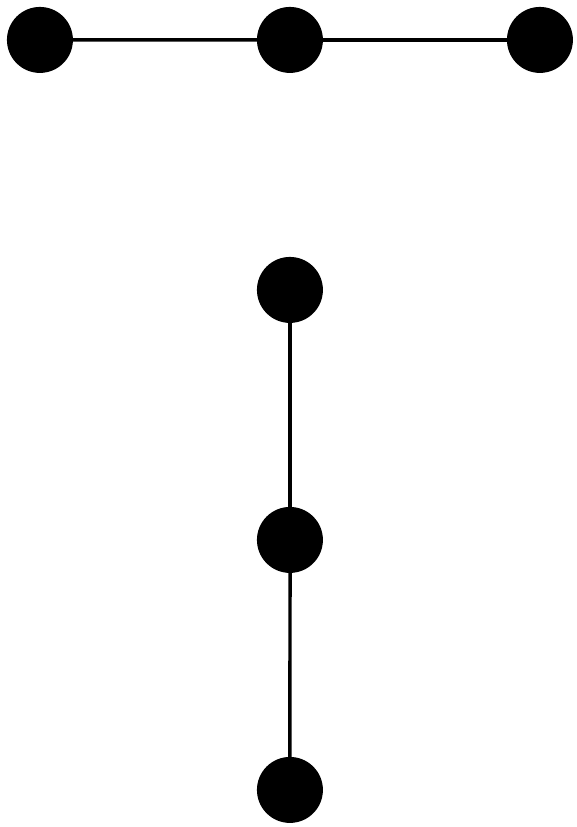}} 
\end{array}
& 
\begin{array}{c} 
m=2 \\ k_1=k_2=3 \\ e_1=e_2=2
\end{array}
& 
\begin{array}{c}
9 \times 9 \times (e-4) \times \\
2^{e-5}  (e-2)(e-3) (e+1)^{e-5} 
\end{array}
\\ \hline
\begin{array}{c} 
\raisebox{-8pt}{\includegraphics[scale=0.12]{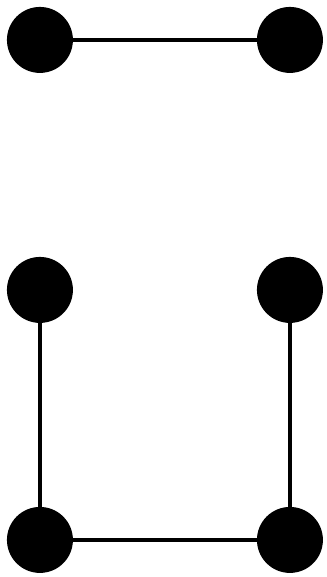}}\, \times 4, \; 
\raisebox{-8pt}{\includegraphics[scale=0.12]{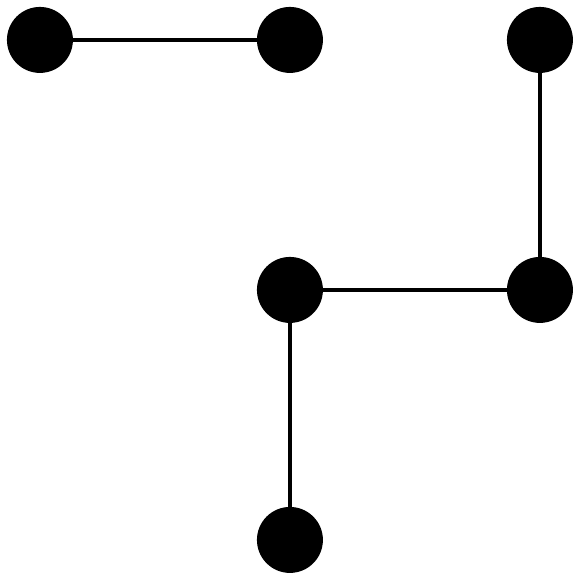}}\, \times 4, \;
\raisebox{-8pt}{\includegraphics[scale=0.12]{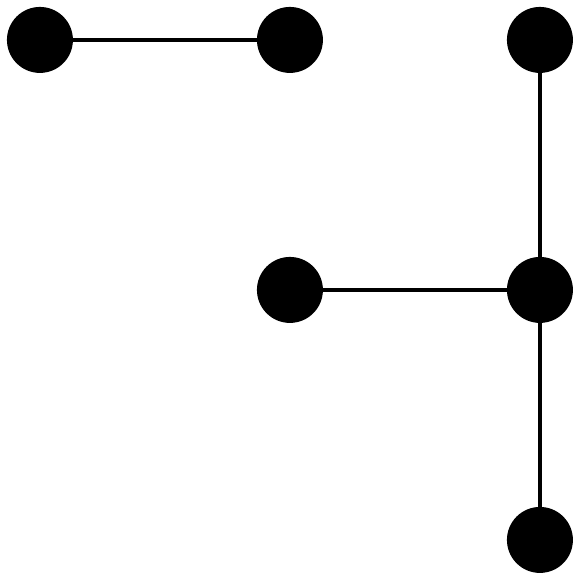}}\, \times 4, \;
\raisebox{-8pt}{\includegraphics[scale=0.12]{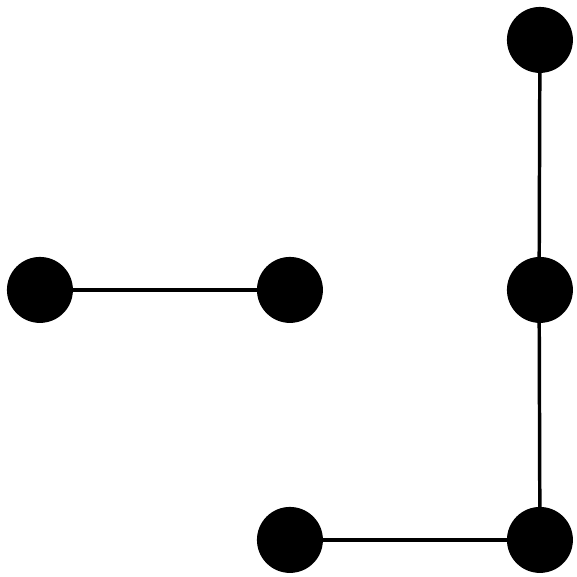}}\, \times 8
\end{array}
& 
\begin{array}{c}
m=2 \\ k_1=2, k_2=4 \\ e_1=1, e_2=3
\end{array}
& 
\begin{array}{c}
20 \times 8 \times (e-4) \times \\
2^{e-5} (e-2)(e-3) (e+1)^{e-5} 
\end{array}
\\ \hline
\begin{array}{c} 
\raisebox{-8pt}{\includegraphics[scale=0.12]{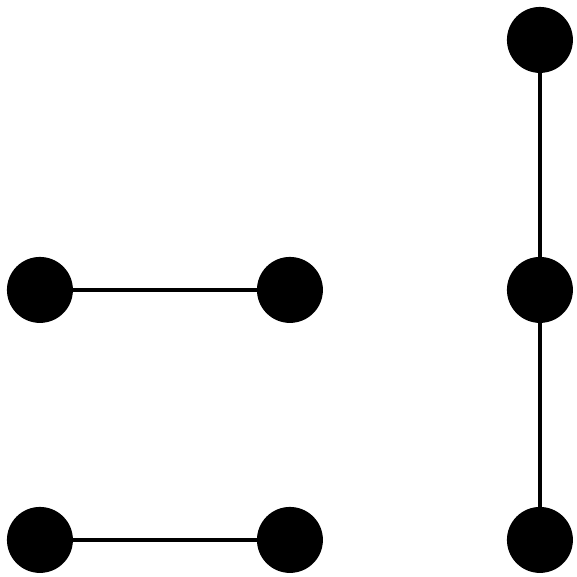}}\, \times 2 \times \frac{1}{2}, \; 
\raisebox{-8pt}{\includegraphics[scale=0.12]{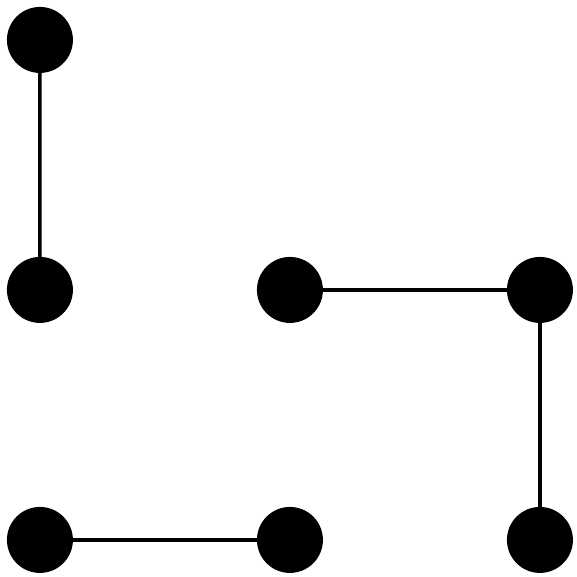}}\, \times 4
\end{array}
&
\begin{array}{c}
m=3 \\ k_1=k_2=2, k_3=3 \\ e_1=e_2=1, e_3=2
\end{array}
& 
\begin{array}{c}
5 \times 12 \times (e-4)(e-5) \times \\
2^{e-5} (e-2)(e-3) (e+1)^{e-5}
\end{array}
\\ \hline
\raisebox{-3pt}{\includegraphics[scale=0.12]{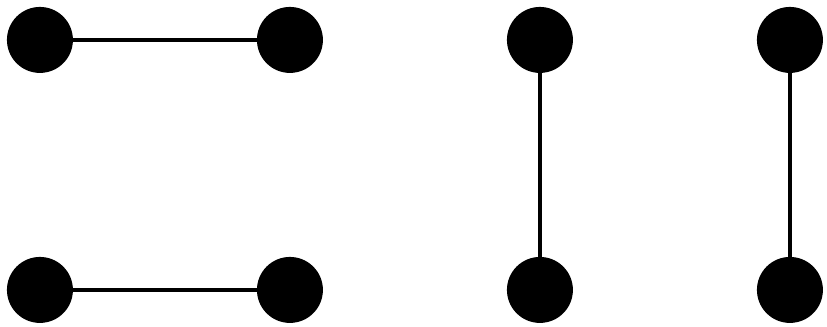}}\, \times \frac{1}{4} 
&
\begin{array}{c}
m=4 \\ k_1=k_2=k_3=k_4=2 \\ e_1=e_2=e_3=e_4=1
\end{array}
& 
\begin{array}{c}
\frac{1}{4} \times 16 \times (e-4)(e-5)(e-6) \times \\
2^{e-5} (e-2)(e-3) (e+1)^{e-5} 
\end{array}
\\ \hline \hline
\includegraphics[scale=0.12]{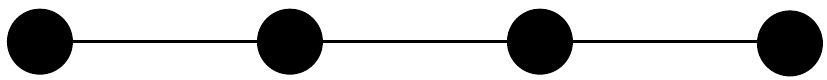} 
& m=1, k_1=4, e_1 = 3
& 
\begin{array}{c} 
4 \times \\ 
2^{e-3} (e-2) (e+1)^{e-4} 
\end{array}
\\ \hline
\includegraphics[scale=0.12]{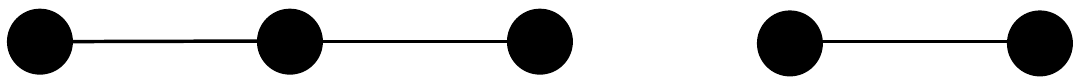} 
& 
\begin{array}{c}
m=2 \\ k_1=3, k_2=2 \\ e_1=2, e_2=1
\end{array}
& 
\begin{array}{c} 
6 (e-3) \times \\
2^{e-3} (e-2) (e+1)^{e-4} 
\end{array}
\\ \hline
\raisebox{1pt}{\includegraphics[scale=0.12]{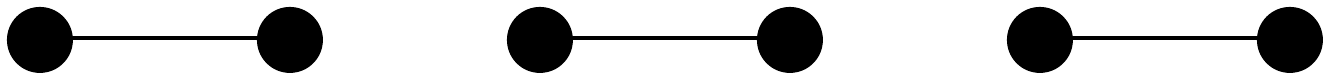}}\, \times \frac{1}{6}
& 
\begin{array}{c} 
m=3 \\ k_1=k_2=k_3=2 \\ e_1=e_2=e_3=1 
\end{array}
& 
\begin{array}{c}
\frac{1}{6} \times 8 \times (e-3)(e-4) \times \\
2^{e-3} (e-2) (e+1)^{e-4} 
\end{array}
\\ \hline
\end{array}
$
\caption{Contributions to $G_e^{(e-2)}$ from various subgraph collections $\calH$. Those above the double line come from labeled animals where two edge labels are duplicated; those below have one edge label which appears three times.}
\label{tab:e-2}
\end{table}


\subsection{Coefficients of perimeter polynomials}


In~\ref{sec:bond-polynomials} we give explicit formulas for $D_e(q)$ for all $d$ for $e \le 11$. As discussed above, we reduce our computation time by computing some of these terms analytically. As in~\eqref{eq:d-choose-k}, each $D_e$ has terms proportional to the ${d \choose k}$ ways that an animal proper in $k$ dimensions can appear. The following theorems compute the coefficient of ${d \choose e}$, ${d \choose e-1}$, and ${d \choose e-2}$.

\begin{theorem}
  \label{thm:coefficient-e}
  The coefficient of ${d \choose e}$ in the perimeter polynomial $D_e$ is
  \begin{equation}
    \label{eq:coefficient-e}
2^e\,(e+1)^{e-2} \,q^{2(e+1)d - 2e} \, .
  \end{equation}
\end{theorem}

\begin{proof}
The coefficient of ${d \choose e}$ is determined by the bond animals of size $e$ that are proper in $e$ dimensions. In Theorem~\ref{thm:proper-e} we showed there are $G_e^{(e)} = 2^e\,(e+1)^{e-2}$ of these. We thus need to show that all these animals have the same perimeter
\begin{equation}
  \label{eq:perimeter-e}
  t_0 = 2(e+1)d - 2e \, .
\end{equation}
when embedded in a $d$ dimensional lattice with $d \geq e$. 

We will prove this by induction on $e$. As noted in the proof of Theorem~\ref{thm:proper-e}, an animal of size $e$ that is proper in dimension $e$ is a tree, and the labels on the edges are all distinct. For the base case, a bond animal of size $e=0$, i.e., consisting of a single vertex and no edges, has perimeter $2d$. Now suppose we increment $e$, connecting some vertex $u$ in the animal to a new vertex $v$. If $v$ were adjacent in the lattice to any vertex $w$ in the animal other than $u$, adding the edge $(v,w)$ would create a loop, but this would imply that some pair of labels are repeated. Thus the induction step replaces the perimeter edge $(u,v)$, and creates $2d-1$ new perimeter edges incident to $v$.  By induction, the perimeter is
\[
t_0 = 2d + (2d-2) e \, , 
\]
which equals~\eqref{eq:perimeter-e}.
%
\end{proof}

\begin{theorem}
  \label{thm:coefficient-e-1}
The coefficient of ${d \choose e-1}$ in the perimeter polynomial $D_e$ is
  \begin{equation}
    \label{eq:coefficient-e-1-structure}
    q^{2(e+1)d - 2e}\left[\alpha_0 + \alpha_1q^{-1}\right] \, .
  \end{equation}
  with
  \begin{equation}
    \label{eq:coefficient-e-1-a0}
    \alpha_0 = 2^{e-2} (e-1) (2e^2-3e+7) (e+1)^{e-4} 
  \end{equation}
  and
  \begin{equation}
    \label{eq:coefficient-e-1-a1}
    \alpha_1 = 2^e (e-1)(e-2) (e+1)^{e-4}  \, .
  \end{equation}
\end{theorem}

\begin{proof}
The coefficient of ${d \choose e-1}$ is determined by the bond animals of size $e$ that are proper in $e-1$ dimensions.  These animals are still trees. However, while $\alpha_0$ of them have perimeter $t_0$ given by~\eqref{eq:perimeter-e}, in $\alpha_1$ of them a pair of vertices share a perimeter edge, reducing the perimeter by $1$.  This gives rise to the form of~\eqref{eq:coefficient-e-1-structure}.  Since $\alpha_0+\alpha_1 = G_e^{(e-1)}$ and we already know $G_e^{(e-1)}$ from Theorem~\ref{thm:proper-e-1}, it suffices to compute $\alpha_1$. 

The animals where two vertices share a perimeter edge, i.e., where both endpoints of a perimeter edge are elements of the animal's vertex set, are those that contain one of the $4$ rotations of the animal \,\raisebox{-3pt}{\includegraphics[scale=0.12]{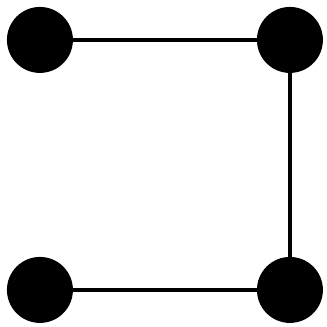}}\,. If this animal is embedded in $d$ dimensions, it has perimeter $9 + 8(d-2)$.  Since the other $e-3$ edge labels are distinct, the inductive argument of Theorem~\ref{thm:coefficient-e} gives a perimeter
\[
t = 9 + 8(d-2) + (2d-2)(e-3) = 2d-1 + (2d-2) e = t_0 - 1 \, .
\]
Applying Lemma~\ref{lem:animals} with $m=1$, $k_1=4$, $e_1=3$ and multiplying by the ${e-1 \choose 2}$ choices of labels corresponding to these two dimensions gives
\[
\alpha_1 
= 4 \times 2^{e-3} \times 4 n^{n-5} {e-1 \choose 2} 
= 2^e (e-1) (e-2) (e+1)^{e-4} \, .
\]
This proves~\eqref{eq:coefficient-e-1-a1}, and subtracting $\alpha_1$ from the expression~\eqref{eq:proper-e-1} for $G_e^{(e-1)}$ gives $\alpha_0$ in~\eqref{eq:coefficient-e-1-a0}.

Note that these animals form a subset of those counted by~\eqref{eq:e-1-separate}, where the edges with the duplicate label do not share a vertex.  Those containing the animal \,\raisebox{-3pt}{\includegraphics[scale=0.12]{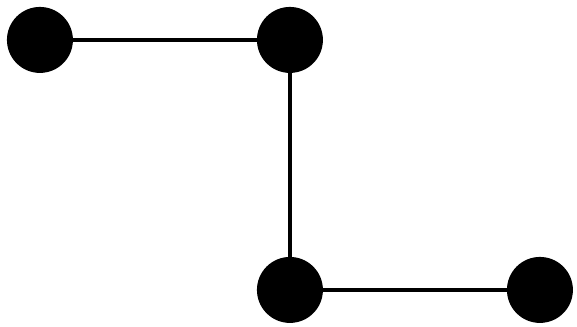}}\, have perimeter $t_0$, as do those counted by~\eqref{eq:e-1-together} which contain \,\raisebox{2pt}{\includegraphics[scale=0.12]{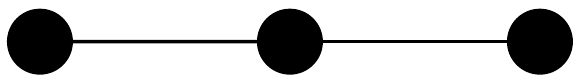}}\,.
\end{proof}

\begin{theorem}
\label{thm:coefficient-e-2}
  For $d \ge 2$ and $e \ge 2$, the coefficient of ${d \choose e-2}$ in the perimeter polynomial $D_e$ is
  \begin{equation}
    \label{eq:coefficient-e-2-structure}
    q^{2(e+1) d - 2e} \left( \beta_0 + \beta_1 q^{-1} + \beta_2 q^{-2} + \beta_{2d} q^{-2d} \right) 
  \end{equation}
  with
  \begin{eqnarray}
    \label{eq:coefficient-e-2-a0}
    \beta_0 &= \frac{1}{3} \,2^{e-5} (e-2) \nonumber \\
    & \quad \times \left( 12 e^5 - 56 e^4 + 115 e^3 - 115 e^2 + 185 e - 237 \right) (e+1)^{e-6}  \, ,
  \end{eqnarray}
  \begin{equation}
    \label{eq:coefficient-e-2-a1}
    \beta_1 = 2^{e-2} (e-2) (e-3) \left( 2 e^3 - 3 e^2 -6 e + 44\right) (e+1)^{e-6} \, ,
  \end{equation}
  \begin{equation}
    \label{eq:coefficient-e-2-a2}
    \beta_2 = 2^{e-3} (e-2)(e-3)(e-4) (4e+25) \,(e+1)^{e-6} \, ,
  \end{equation}
  and
  \begin{equation}
    \label{eq:coefficient-e-2-a3}
    \beta_{2d} = 2^{e-3} (e-2)(e-3)\,e^{e-5} \, .
  \end{equation}
\end{theorem}

\begin{proof}
The coefficient of ${d \choose e-2}$ is determined by the bond animals of size $e$ that are proper in $e-2$ dimensions, and $\beta_j$ is the number of these where the perimeter is $t_0-j$.  We start with animals that contain a loop \,\raisebox{-3pt}{\includegraphics[scale=0.12]{square-loop}}\,.  This loop on its own has perimeter $8+8(d-2) = 8(d-1)$.  Again using induction for the $e-4$ additional edges, an animal containing it has perimeter
\[
8(d-1) + (2d-2) (e-4) = (2d-2) e = t_0 - 2d \, ,
\]
and we already computed the number $\beta_{2d}$ of these animals in~\eqref{eq:square-loop}.

The bond animals that contribute to $\beta_0$, $\beta_1$ and $\beta_2$ are trees.  As in Theorem~\ref{thm:proper-e-2}, these have two duplicate edge labels or one triplicate one. We start with $\beta_2$, where two perimeter edges are shared, i.e., both their endpoints are in the animal's vertex set.  None of the subgraph collections $\calH$ shown in Table~\ref{tab:e-2} have this property on their own, but for several of them their subgraphs $H_i$ can be connected to produce shared perimeter edges. 

For instance, in the top row of Table~\ref{tab:beta-2} we see how two elbow-shaped animals can be connected along a third dimension, creating two shared perimeter edges parallel to the connecting edge. Multiplying by the total number of images under symmetry and by the ${e-2 \choose 3}$ choices of these three axes among the $e-2$ distinct edge labels, these combined animals contribute
\[
36 \times 6 \times 2^{e-5} {e-2 \choose 3} n^{n-7}
= 36 \times 2^{e-5} (e-2)(e-3)(e-4) (e+1)^{e-6}
\]
to $\beta_2$.  Adding the contributions from Table~\ref{tab:beta-2} and simplifying gives~$\beta_2$ in~\eqref{eq:coefficient-e-2-a2}.

\begin{table}
\arraycolsep=2pt
\def\arraystretch{1.4}
$
\begin{array}{|c|c|c|} \hline
\calH & \text{parameters for Lemma~\ref{lem:animals}} & \text{contribution to $\beta_2$ for $e \ge 2$} \\ \hline 
\raisebox{-7pt}{\includegraphics[scale=0.12]{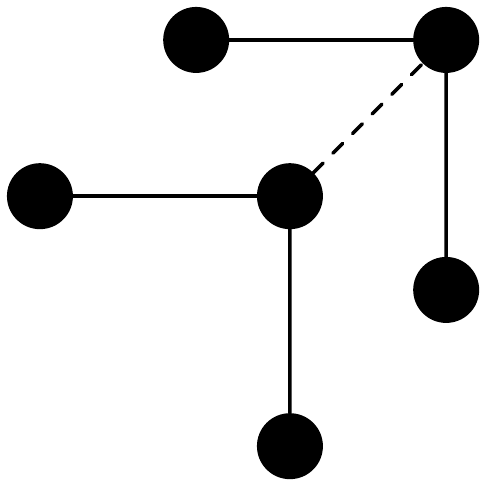}}\, \times 12 , \; 
\raisebox{-7pt}{\includegraphics[scale=0.12]{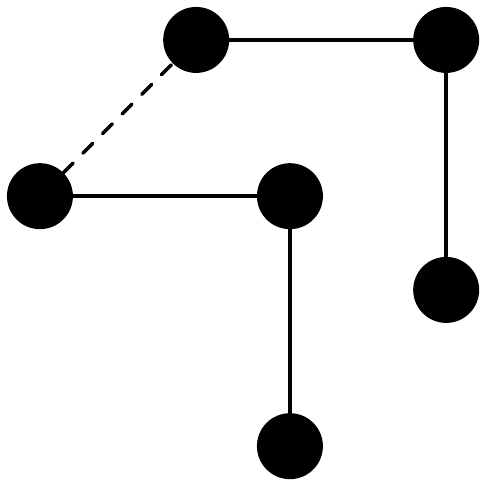}}\, \times 24 , \;  
& 
m=1, k_1 = 6, e_1 = 5
& 
\begin{array}{c}
36 \times 6 \times \\
2^{e-5} {e-2 \choose 3} (e+1)^{e-6}
\end{array}
\\ \hline
\raisebox{-4pt}{\includegraphics[scale=0.12]{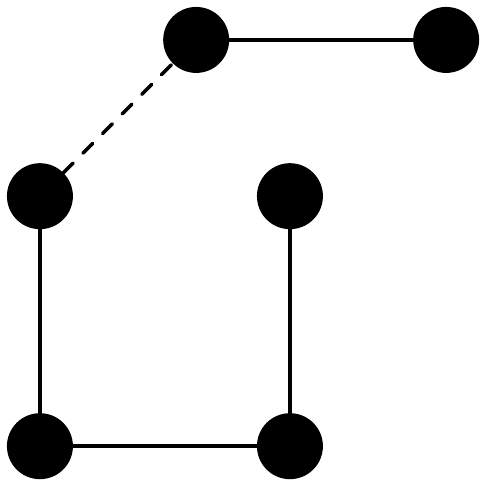}}\, , \; 
\raisebox{-4pt}{\includegraphics[scale=0.12]{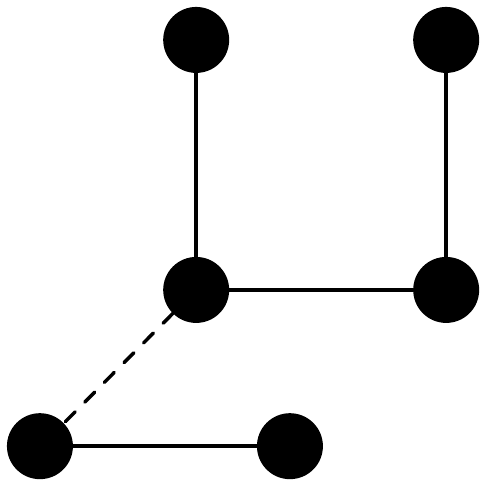}}\, \times 48 
& 
m=1, k_1 = 6, e_1 = 5
& 
\begin{array}{c}
96 \times 6 \times \\
2^{e-5} {e-2 \choose 3} (e+1)^{e-6}
\end{array}
\\ \hline
\raisebox{-4pt}{\includegraphics[scale=0.12]{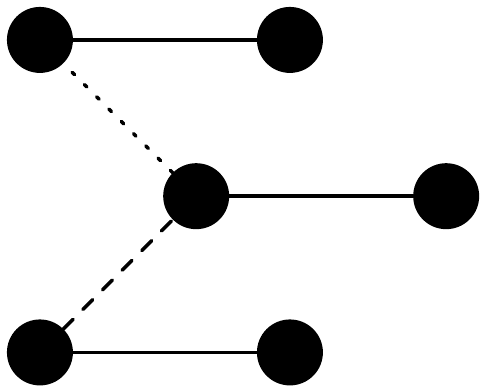}}\, , \; 
\raisebox{-4pt}{\includegraphics[scale=0.12]{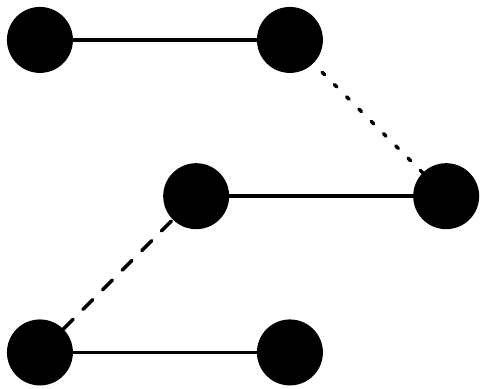}}\, \times 24 
& 
m=1, k_1 = 6, e_1 = 5
& 
\begin{array}{c}
48 \times 6 \times \\
2^{e-5} {e-2 \choose 3} (e+1)^{e-6}
\end{array}
\\ \hline \hline
\raisebox{-6pt}{\includegraphics[scale=0.12]{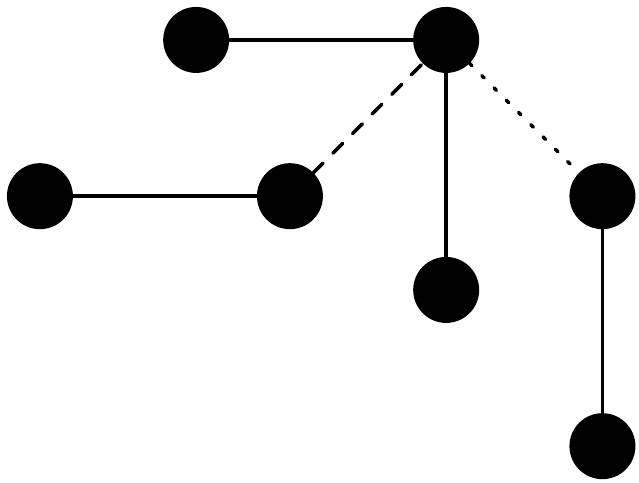}}\, \times 192 , \; 
\raisebox{-6pt}{\includegraphics[scale=0.12]{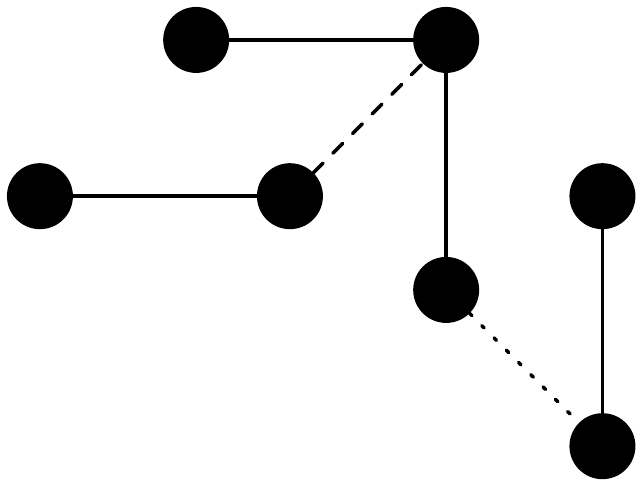}}\, \times 384 , \; 
\raisebox{-6pt}{\includegraphics[scale=0.12]{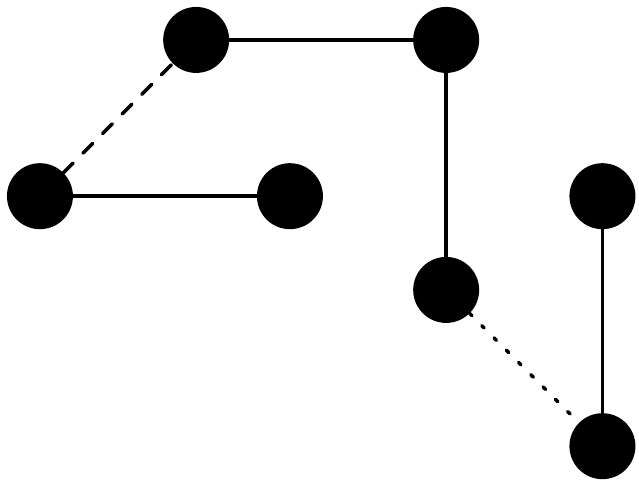}}\, \times 192 
& 
m=1, k_1 = 7, e_1 = 6 
& 
\begin{array}{c}
768 \times 7 \times \\ 
2^{e-6} {e-2 \choose 4} (e+1)^{e-7} 
\end{array}
\\ \hline
\raisebox{-4pt}{\includegraphics[scale=0.12]{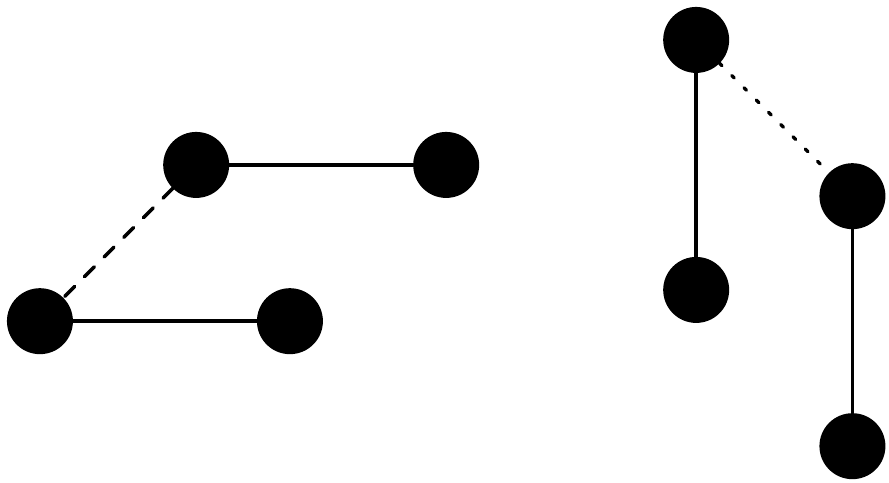}}\, \times 48
&
\begin{array}{c}
m=2 \\ k_1=k_2=4 \\ e_1=e_2=4
\end{array}
&
\begin{array}{c}
48 \times 16 \times (e-6) \times \\
2^{e-6} {e-2 \choose 4}  (e+1)^{e-7}
\end{array}
\\ \hline
\end{array}
$
\caption{Contributions to $\beta_2$. These arise from connecting subgraphs shown in Table~\ref{tab:e-2} in a way that creates two perimeter edges whose endpoints are both in the animal. Solid horizontal and vertical edges are those in subgraphs with repeated labels from Table~\ref{tab:e-2}. Dotted and dashed diagonal lines represent distinct edge labels pointing in a third and fourth dimension. The first two rows involve three of the $e-2$ distinct edge labels, and the last two involve four of them.}
\label{tab:beta-2}
\end{table}

\begin{table}
\arraycolsep=2pt
\def\arraystretch{1.4}
$
\begin{array}{|c|c|c|} \hline
\calH & \text{Lemma~\ref{lem:animals}} & \text{contribution to $\beta_1$ for $e \ge 2$} \\ \hline 
\raisebox{-7pt}{\includegraphics[scale=0.12]{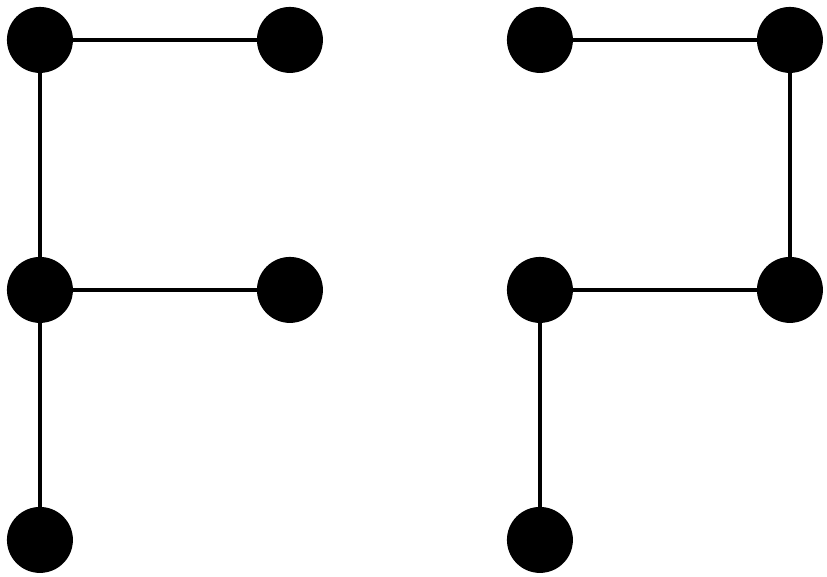}}\, \times 8
& 
m=1, k_1 = 5, e_1 = 2
& 
16 \times 5 \times 
2^{e-4} {e-2 \choose 2} (e+1)^{e-5}
\\ \hline 
\begin{array}{c}
\raisebox{-10pt}{\includegraphics[scale=0.12]{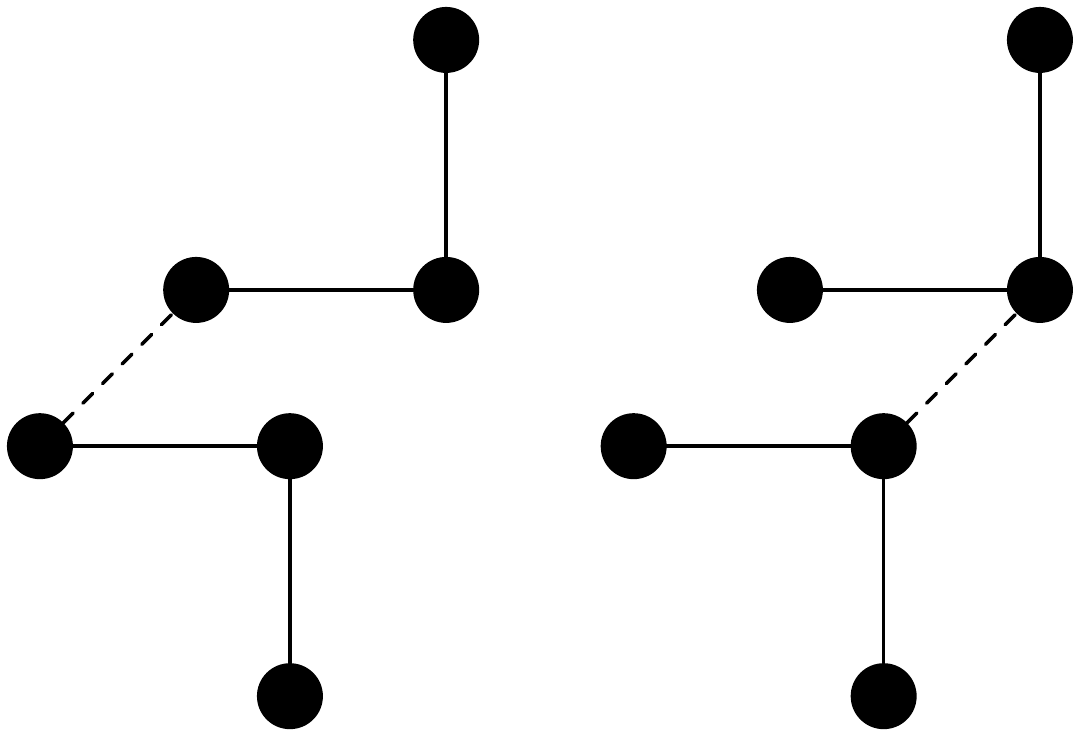}}\, \times 24, \; 
\raisebox{-6pt}{\includegraphics[scale=0.12]{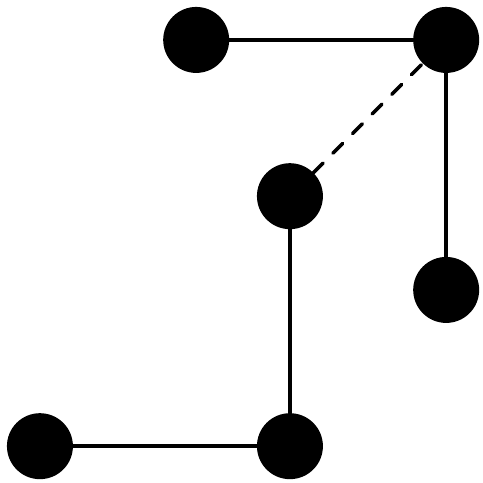}}\, \times 48 
\\
\raisebox{-6pt}{\includegraphics[scale=0.12]{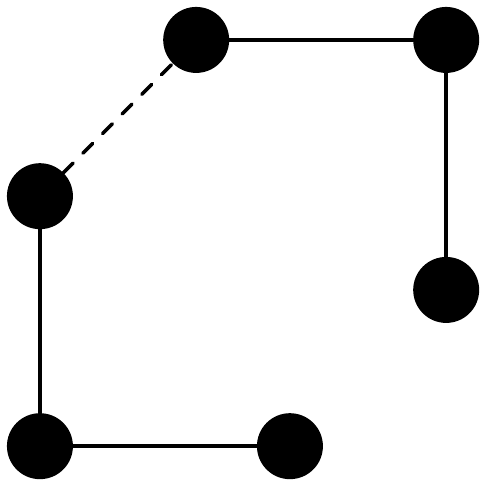}}\, \times 24, \; 
\raisebox{-10pt}{\includegraphics[scale=0.12]{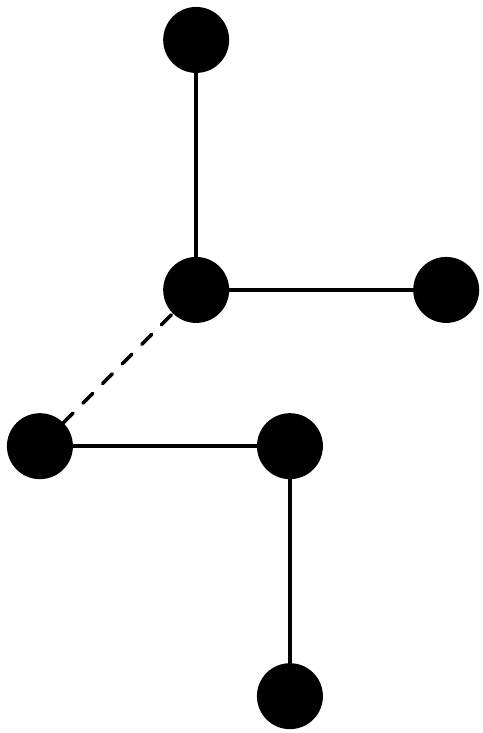}}\, \times 48 
\\
\raisebox{-7pt}{\includegraphics[scale=0.12]{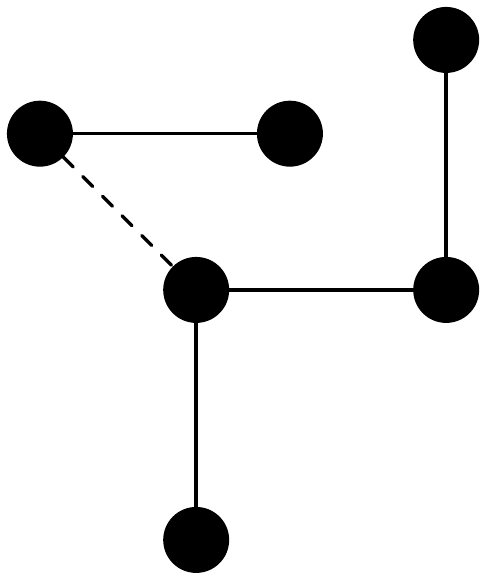}}\, \times 48 , \; 
\raisebox{-7pt}{\includegraphics[scale=0.12]{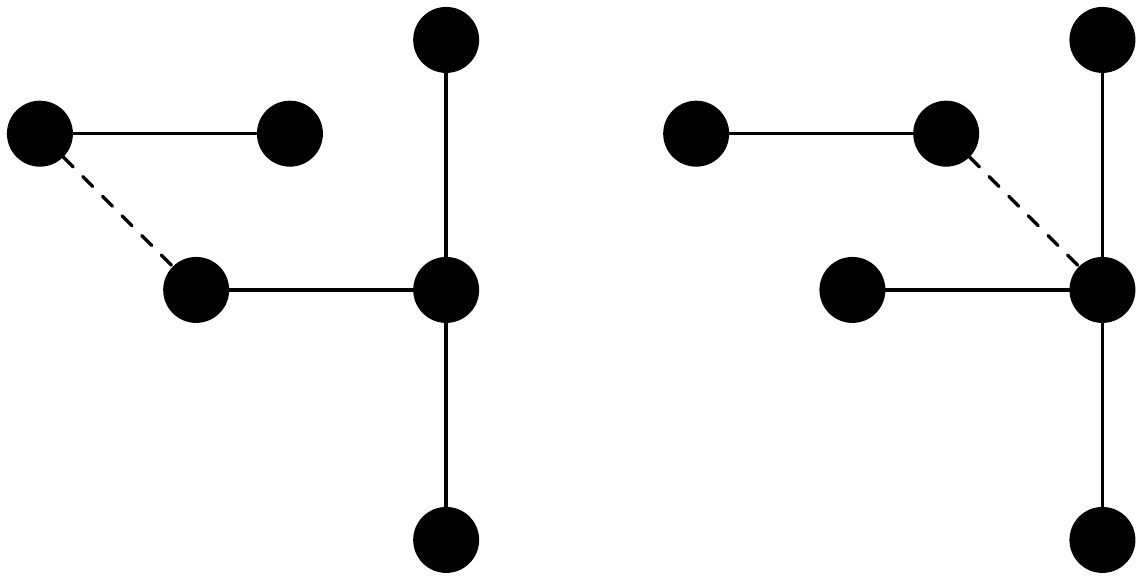}}\, \times 24 , \; 
\raisebox{-7pt}{\includegraphics[scale=0.12]{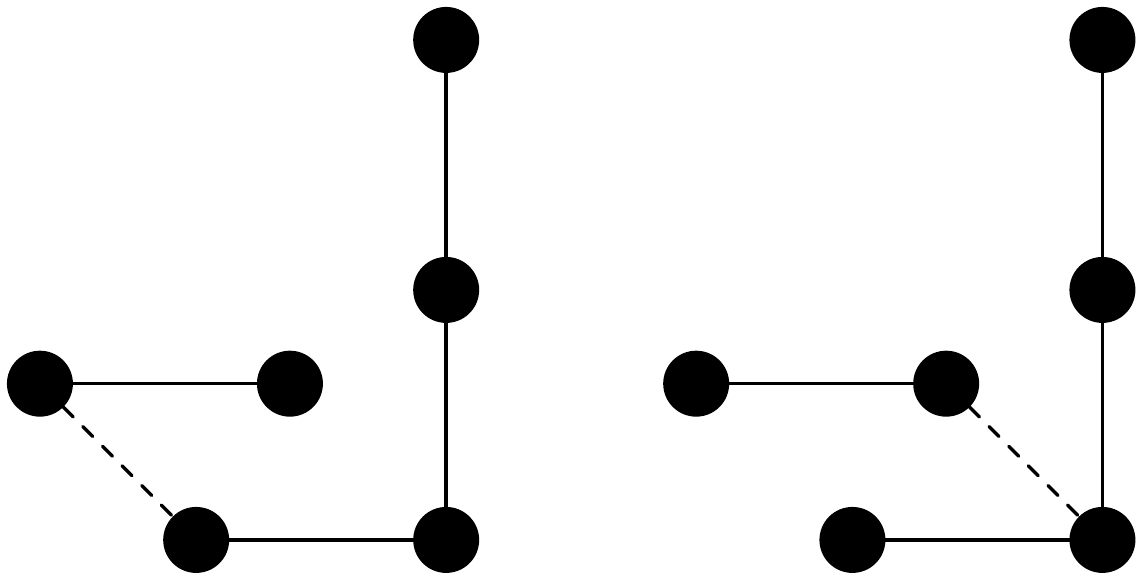}}\, \times 48  
\end{array}
& 
m=1, k_1 = 6, e_1 = 5
& 
360 \times 6 \times 
2^{e-5} {e-2 \choose 3} (e+1)^{e-6}
\\ \hline
4\; \raisebox{-7pt}{\includegraphics[scale=0.12]{1-3-pair1}} \,
- 48 \left( 
\, \raisebox{-5pt}{\includegraphics[scale=0.12]{c-and-domino-2}} 
\,+\,
\raisebox{-5pt}{\includegraphics[scale=0.12]{c-and-domino-1}} \,
\right)
&
\begin{array}{c}
m=2 \\ k_1=2, k_2=4 \\ e_1=1, e_2=3
\end{array}
&
\begin{array}{c}
32 (e-4) 2^{e-4} {e-2 \choose 2} (e+1)^{e-5} \\
-\, 576 \times 2^{e-5} {e-2 \choose 3} (e+1)^{e-6} \\
= 6 \times 2^e (e-2) {e-2 \choose 3} (e+1)^{e-6}
\end{array}
\\ \hline
\raisebox{-7pt}{\includegraphics[scale=0.12]{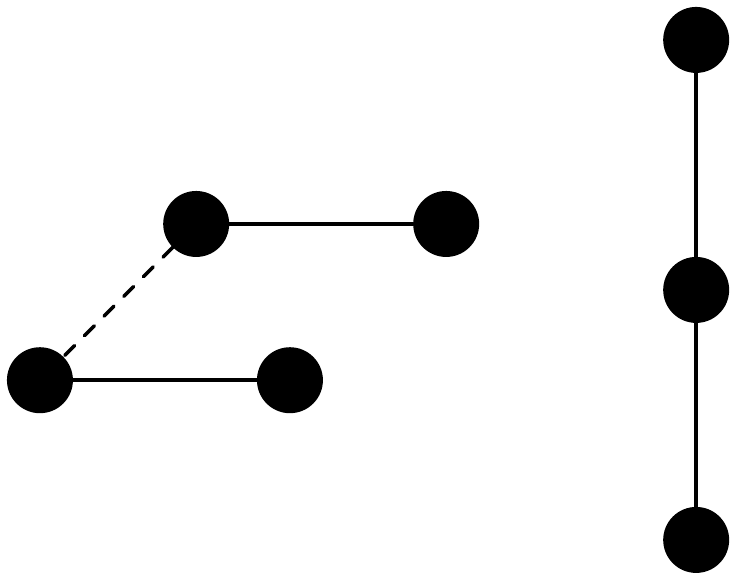}}\, \times 12 
& 
\begin{array}{c}
m=2 \\ k_1 = 4, k_2 = 3 \\ e_1 = 3, e_2 = 2
\end{array}
& 
12 \times 12 (e-5) 
2^{e-5} {e-2 \choose 3} (e+1)^{e-6}
\\ \hline
\begin{array}{c}
48 \left(
\;\raisebox{-9pt}{\includegraphics[scale=0.12]{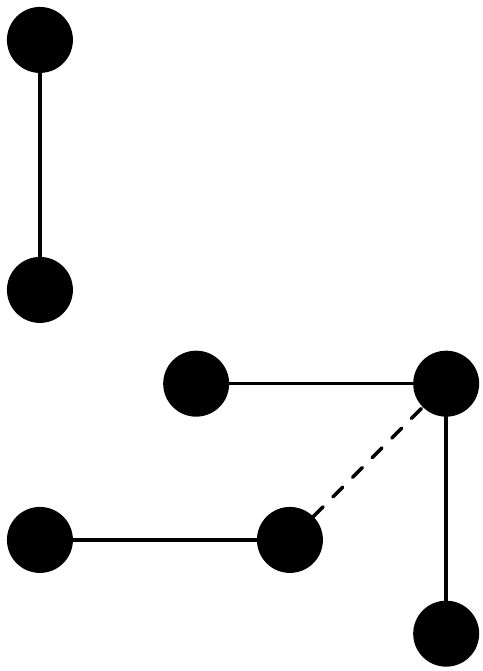}} 
\,+\,
\raisebox{-9pt}{\includegraphics[scale=0.12]{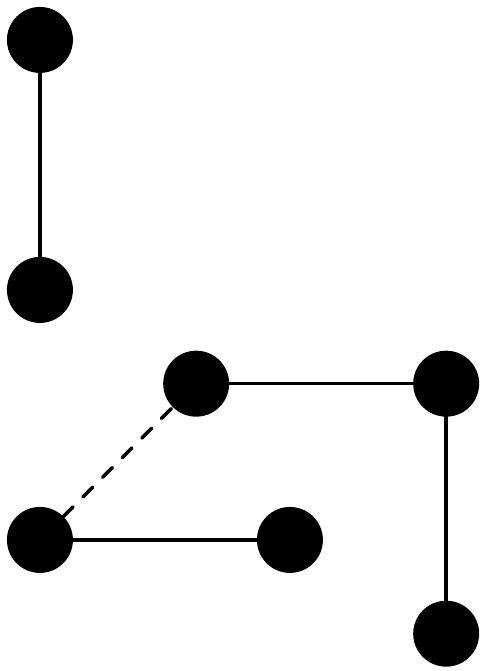}} \;
\right)
\\
-\, 2 \left(
192\; \raisebox{-6pt}{\includegraphics[scale=0.12]{elbow-dominos-1}}
\,+\,
384\; \raisebox{-6pt}{\includegraphics[scale=0.12]{elbow-dominos-2}}
\,+\,
192\; \raisebox{-6pt}{\includegraphics[scale=0.12]{elbow-dominos-3}} \;
\right)
\end{array}
&
\begin{array}{c}
m=2 \\ k_1=2, k_2=5 \\ e_1=1, e_2=4
\end{array}
&
\begin{array}{c}
96 \times 10 (e-5) 
2^{e-5} {e-2 \choose 3} (e+1)^{e-6} \\ 
-\, 2 \times 768 \times 7 \times 
2^{e-6} {e-2 \choose 4} (e+1)^{e-7} \\
= 24 \times 2^e (5e-2) {e-2 \choose 4} (e+1)^{e-7}
\end{array}
\\ \hline
12 \times \frac{1}{2} \times\, \raisebox{-3pt}{\includegraphics[scale=0.12]{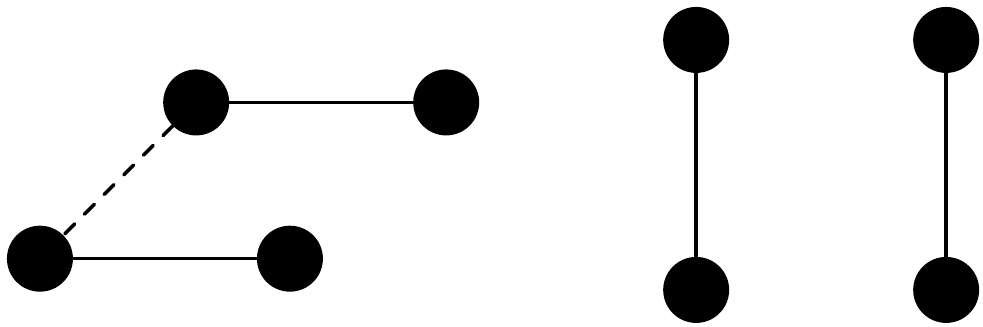}} 
\;-\,
2 \left( 48\; \raisebox{-6pt}{\includegraphics[scale=0.12]{two-cs}}\; \right)
&
\begin{array}{c}
m=3 \\ k_1=4, k_2=k_3=2 \\ e_1=3, e_2=e_3=1
\end{array}
&
\begin{array}{c}
6 \times 16 (e-5)(e-6) 2^{e-5} {e-2 \choose 3} (e+1)^{e-6} \\
-\, 2 \times 48 \times 16 (e-6) 
2^{e-6} {e-2 \choose 4}  (e+1)^{e-7} \\ 
= 60 \times 2^e (e-1) {e-2 \choose 5} (e+1)^{e-7}
\end{array}
\\ \hline
\raisebox{-1pt}{\includegraphics[scale=0.12]{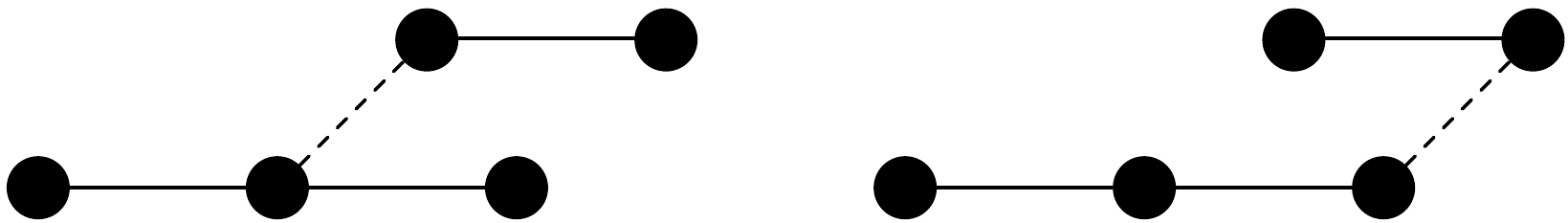}}\, \times 8
&
m=1, k_1=5, e_1=4
&
16 \times 5 \times 2^{e-4} {e-2 \choose 2} (e+1)^{e-5} 
\\ \hline
4 \; \raisebox{-6pt}{\includegraphics[scale=0.12]{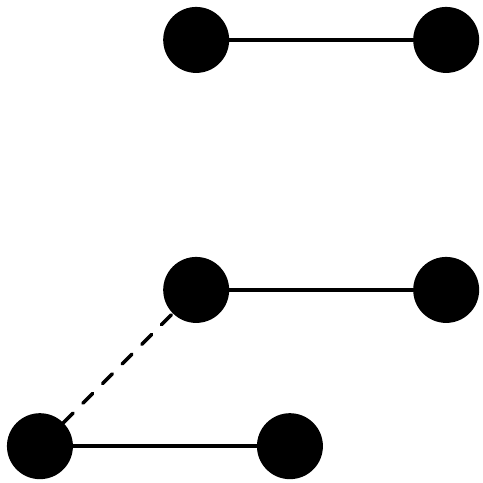}} \; 
- 2 \times 24 \left( 
\raisebox{-4pt}{\includegraphics[scale=0.12]{three-dominos-1}}
\,+\,
\raisebox{-4pt}{\includegraphics[scale=0.12]{three-dominos-2}}
\right)
& 
\begin{array}{c}
m=2 \\ k_1=2, k_2=4 \\ e_1=1, e_2=3
\end{array}
& 
\begin{array}{c}
4 \times 8 (e-4) 2^{e-4} {e-2 \choose 2} (e+1)^{e-5} \\ 
-\, 2 \times 48 \times 6 \times 
2^{e-5} {e-2 \choose 3} (e+1)^{e-6} \\
= 6 \times 2^e (e-2) {e-2 \choose 3} (e+1)^{e-6}
\end{array}
\\ \hline
\end{array}
$
\caption{Contributions to $\beta_1$. These arise from connecting subgraphs shown in Table~\ref{tab:e-2} in a way that creates exactly one perimeter edge whose endpoints are both in the animal. Solid horizontal and vertical edges are those in subgraphs with repeated labels from Table~\ref{tab:e-2}. Dotted and dashed diagonal lines represent distinct edge labels pointing in a third and fourth dimension. In several cases, we have to subtract the subgraphs counted in Table~\ref{tab:beta-2} to avoid two shared perimeter edges; then the second column gives the parameters for the leading term. In the 5th, 6th, and 8th rows, there is a factor of 2 in front of the subtracted subgraphs because there are two ways to construct each one by connecting the subgraphs in the leading term.}
\label{tab:beta-1}
\end{table}

For $\beta_1$, the number of cases is larger. Several of the subgraph collections $\calH$ in Table~\ref{tab:e-2} already have a shared perimeter edge, and others can be connected in ways that produce one. Some of the resulting subgraph collections can be further connected to produce two shared edges, so we have to subtract some of the subgraphs we already counted in Table~\ref{tab:beta-2} toward $\beta_2$. In some cases we have to subtract these subgraphs twice, since there are two distinct ways to construct them. Adding the contributions from Table~\ref{tab:beta-1} and simplifying gives $\beta_1$ in~\eqref{eq:coefficient-e-2-a1}.

Finally, since $\beta_0 + \beta_1 + \beta_2 + \beta_{2d} = G_e^{(e-2)}$, we obtain $\beta_0$ by subtracting $\beta_1, \beta_2$, and $\beta_{2d}$ from the expression~\eqref{eq:proper-e-2}. Simplifying gives~\eqref{eq:coefficient-e-2-a0}.
\end{proof}

\section{Site Percolation}
\label{sec:site-percolation}

In this section we apply the same strategy to site percolation, following Gaunt, Sykes, and Ruskin~\cite{gaunt:sykes:ruskin:76} who derived the series~\eqref{eq:pc-site-series-old} for $\pcsite$. The situation here is somewhat simpler than for bond percolation, and we are able to compute the next term using existing enumerations~\cite{luther:mertens:11a} and recent analytic results from~\cite{luther:mertens:17}.

A site animal of size $v$ is a connected set of $v$ vertices, or if you prefer the subgraph induced by this set, i.e., these vertices and all the bonds between them. The perimeter $t$ of an animal is the number of vertices adjacent to it; so, for instance, the animal \,\raisebox{-1.5pt}{\includegraphics[scale=0.1]{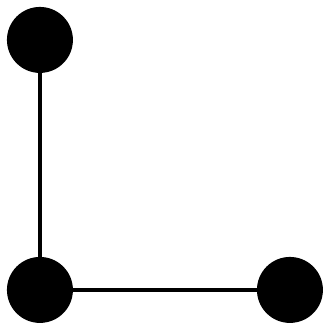}}\, has size $v=3$ and perimeter $t=7$. We again say that an animal is proper in $k$ dimensions if it spans a $k$-dimensional subspace when one of its vertices is at the origin. Note that a site animal of size $v$ is proper in at most $v-1$ dimensions.

Re-using the notation of Section~\ref{sec:lattice-animals} in this setting, let $A_d(v)$ denote the number of animals of size $v$ in $d$ dimensions, and let $G_v^{(k)}$ denote the number of animals proper in $k$ dimenions. Replacing $e$ with $v$ in~\eqref{eq:d-choose-k} gives~\cite{lunnon:75}
\begin{equation}
\label{eq:d-choose-k-site}
A_d(v) = \sum_{k=1}^{v-1} {d \choose k} G_v^{(k)} \, . 
\end{equation}
We write $G_{v,t}^{(k)}$ for the number of animals of size $v$ and perimeter $t$, counting only the adjacent vertices in their $k$-dimensional subspace. When embedded in $d$ dimensions, these have $2(d-k)v$ additional perimeter vertices ``above'' and ``below'' them, giving the perimeter polynomials
\begin{equation}
  \label{eq:def-perimeter-poly-site}
  D_v(q) = \sum_{k=1}^{v-1} \sum_{t} {d \choose k} G^{(k)}_{v,t}\, q^{t + 2(d-k)v} \, .
\end{equation}
This is simpler than the corresponding definition~\eqref{eq:def-perimeter-poly-bond} for bond animals, since we no longer need to keep track separately of the number of vertices and edges in the animal. As before we have $D_v(1) = A_d(v)$. The expected size of the cluster to which a given vertex belongs, conditioned on the event that that vertex is occupied, is 
\begin{equation}
\label{eq:expected-size-site}
S = \frac{1}{p} \sum_v v^2 p^v D_v(1-p) \, ,
\end{equation}
since an animal of size $v$ has $v$ possible translations, each of which contributes $p^v D_v(1-p)$ to the probability that a given vertex belongs to an animal of size $v$. 

Following~\cite{gaunt:sykes:ruskin:76}, we expand $S$ in powers of $p$, 
\begin{equation}
\label{eq:br-site}
S = \sum_v v^2 p^{v-1} D_v(1-p) 
\triangleq \sum_{r=0}^\infty b_r p^r \, .
\end{equation}
The coefficients $b_r$ only depend on $D_v$ for $v \le r+1$. Moreover, analogous to the discussion above for bond percolation, $b_r$ depends on $D_{r+1}$ only through $A_d(r+1) = D_{r+1}(1)$. We can derive $D_{r+1}(1)$ from $D_v(q)$ for $v \le r$ by using the fact that the total probability of belonging to any animal is $p$, 
\begin{equation}
\label{eq:luther-mertens-site}
\sum_{v=1}^\infty v p^v D_v(1-p) = p \, .
\end{equation}
This yields the identity~\cite{luther:mertens:11a}
\begin{equation}
\label{eq:identity-site}
D_{r+1}(1) = A_d(r+1) = -\frac{1}{r+1} \left[ \sum_{v=1}^{r} v p^v D_v(1-p) \right]_{r+1} \, .
\end{equation}

In~\ref{sec:site-polynomials} we give $D_v(q)$ for all $v \le 12$, using enumerations carried out in~\cite{luther:mertens:11a} and analytical results in~\cite{luther:mertens:17} analogous to Theorems~\ref{thm:coefficient-e} and~\ref{thm:coefficient-e-1}. We then use~\eqref{eq:identity-site} to derive $D_{13}(1)$. This lets us compute the coefficients $b_r$ in the power series for $S$ for all $r \le 12$, and we give these explicitly in~\ref{sec:site-br}. 

As in Section~\ref{sec:series-expansion-bond} we now ask whether the correction of $b_r$ to the mean-field behavior, which corresponds to the leading term $\sigma^r$, stabilizes as $r$ increases. We find that $b_r$ is consistent with the following, where the terms up to $r \ge 6$ appeared in~\cite{gaunt:sykes:ruskin:76}:
\begin{eqnarray}
\frac{b_r}{\sigma^r} = 1 &+ \left( 4 - \frac{3r}{2} \right) \sigma^{-1} & \quad (r \ge 2) \nonumber \\
&+ \left( \frac{37}{2} - \frac{69 r}{8} + \frac{9 r^2}{8} \right) \sigma^{-2} & \quad (r \ge 4) \nonumber \\
&+ \left( \frac{651}{4} - \frac{109 r}{2} + \frac{135 r^2}{16} - \frac{9 r^3}{16} \right) \sigma^{-3} & \quad (r \ge 6) \nonumber \\
&+ \left( 
\frac{\numprint{13375}}{8}
-\frac{\numprint{89035} r}{192}
+\frac{8241 r^2}{128}
-\frac{333 r^3}{64}
+\frac{27 r^4}{128}
\right) \sigma^{-4} & \quad (r \ge 8) \nonumber \\
&+ O(\sigma^{-5}) \, .
\label{eq:br-polys-site}
\end{eqnarray}
The coefficient of $\sigma^{-j}$ appears to be a polynomial of degree $j$;  we found the coefficient of $\sigma^{-4}$ by fitting a 4th-order polynomial to five points, the minimum necessary.  

Assuming that~\eqref{eq:br-polys-site} holds, as for bond percolation, in each coefficient of the logarithm the terms that are quadratic or higher order in $r$ cancel out, giving
\begin{align}
\ln b_r &= r \ln \sigma 
+ \left( 4 - \frac{3r}{2} \right) \sigma^{-1} 
+ \left( \frac{21}{2} - \frac{21r}{8} \right) \sigma^{-2}
+ \left( \frac{1321}{12} - \frac{65 r}{4} \right) \sigma^{-3} 
\nonumber \\
&+ \left( \frac{4327}{4} - \frac{\numprint{20359} r}{192} \right) \sigma^{-4} 
+ O(\sigma^{-5}) \, . 
\label{eq:log-br-site}
\end{align}
where the coefficient of $\sigma^{-j}$ holds for $j \ge 2r$. This again lets us take the limit $r \to \infty$ in~\eqref{eq:mu-limit}, giving
\begin{align}
\label{eq:mu-site}
\ln \mu 
&= \ln \sigma 
- \frac{3}{2} \sigma^{-1} 
- \frac{21}{8} \sigma^{-2}
- \frac{65}{4} \sigma^{-3} 
- \frac{\numprint{20359}}{192} \sigma^{-4} 
- O(\sigma^{-5}) \, .
\end{align}
Finally, setting $\pcsite = 1/\mu$ and exponentiating gives
\begin{align}
  \label{eq:pc-site-series-new}
  \pcsite(d) 
  &= \sigma^{-1} + \frac{3}{2}\sigma^{-2}
  + \frac{15}{4}\sigma^{-3} + \frac{83}{4}\sigma^{-4} 
  + \frac{6577}{48} \sigma^{-5} + O(\sigma^{-6}) \, ,
\end{align}
extending~\eqref{eq:pc-site-series-old} by one term.

\section{Series Expansion from Pad\'e Approximants}
\label{sec:pade}

Having obtained series expansions for the mean cluster size $S$, another approach to estimating the threshold is to use Pad\'e approximants: that is, to approximate $S$ as a rational function of $p$, find the smallest real root of its denominator where this approximation diverges, and expand this root in powers of $\sigma^{-1}$.  This yields tantalizing results.

For site percolation, we use the first $\ell$ terms of the series
expansion~\eqref{eq:br-site} of $S$ for $1 \le \ell \le 12$. For each
$\ell$, we compute the $(\ell-1,1)$ Pad\'e approximant of $S$, i.e.,
where the numerator is $(\ell-1)$st-order in $p$ and the denominator
is linear in $p$. This means that we compute numbers
$a_0,\ldots,a_{\ell-1}$ and $\beta$ such that
\begin{displaymath}
  \frac{a_0 + a_1p + a_2p^2+\cdots+a_{\ell-1}p^{\ell-1}}{1-\beta p} =
  \sum_{r=0}^\ell b_r p^r\,.
\end{displaymath}
Expanding $(1-\beta p)^{-1}$ using the well-known formula for the geometric series, we get $\beta b_{\ell-1} = b_\ell$. Hence the root of the denominator can be written explicitly as $b_{\ell-1} / b_\ell$, and we expand this in powers of $\sigma^{-1}$. As Table~\ref{tab:site-pade} shows, our new term in~\eqref{eq:pc-site-series-new} shows up at $\ell=9$. Moreover, at $\ell=11$ the next term beyond it appears, suggesting
\begin{align}
\pcsite &= \sigma^{-1} + \frac{3}{2}\sigma^{-2} + \frac{15}{4}\sigma^{-3} + \frac{83}{4}\sigma^{-4} 
+ \frac{6577}{48} \sigma^{-5} 
\nonumber \\ 
&+ \frac{\numprint{119077}}{96} \sigma^{-6} + O(\sigma^{-7}) \, . 
\label{eq:pc-site-series-pade}
\end{align} 
More generally, we conjecture based on Table~\ref{tab:site-pade} that the coefficient of $\sigma^{-j}$ stabilizes when $\ell = 2j-1$. If that is the case, then we could obtain the coefficient of $\sigma^{-7}$ with just one more value of $b_r$, namely $b_{13}$ for site animals.

For bond percolation, we use the first $\ell$ terms of the series expansion~\eqref{eq:br-bond} of $dpS$. As shown in Table~\ref{tab:bond-pade} the coefficient of $\sigma^{-j}$ again appears to stabilize when $\ell=2j-1$, reproducing the new term in~\eqref{eq:pc-bond-series-new} at $\ell=13$, and suggesting that we can obtain the term beyond it if we can compute $b_{13}$ for bond animals.

This stabilization is no accident. In fact, we will now show that it
is equivalent to the same type of stabilization observed
in~\eqref{eq:br-polys-bond} and~\eqref{eq:log-br-bond} for bond
percolation and in~\eqref{eq:br-polys-site} and~\eqref{eq:log-br-site}
for site percolation: namely, that the coefficient of $\sigma^{-j}$ in
$\ln b_r$ becomes linear in $r$ for $r \ge 2j$. 

\begin{theorem}
  \label{thm:pade-equiv} 
The following two conditions are equivalent:
\begin{enumerate}
\item There are coefficients $c_k, m_k$ for $k = 1, 2, 3\ldots$ such that, for all $j \ge 1$, if $r \ge 2j$ then 
\begin{equation}
\label{eq:gsr-convergence} 
\ln b_r = r \ln \sigma + \sum_{k=1}^j (c_k - m_k r) \sigma^{-k} + O(\sigma^{-(j+1)}) \, . 
\end{equation}
\item For all $j \ge 1$, the estimate of $p_c$ 
from the $\ell$th Pad\'e approximant where $\ell \ge 2j-1$ is
\begin{equation}
\frac{b_{\ell-1}}{b_\ell} 
= \frac{1}{\mu} \big( 1+O(\sigma^{-j}) \big)
= \frac{1}{\mu} + O(\sigma^{-(j+1)}) \, ,
\label{eq:pade-convergence}
\end{equation}
where 
\[
\ln \mu = \lim_{r \to \infty} \frac{1}{r} b_r = \ln \sigma - \sum_{k=1}^\infty m_k \sigma^{-k} \, . 
\]
\end{enumerate}
\end{theorem}

\begin{proof}
In one direction, assume that~\eqref{eq:gsr-convergence} holds. Then if $\ell \ge 2j-1$, 
\begin{align}
\ln \frac{b_{\ell-1}}{b_\ell} 
= \ln b_{\ell-1} - \ln b_\ell
= -\ln \sigma + \sum_{k=1}^{j-1} m_k \sigma^{-k} + O(\sigma^{-j})
= -\ln \mu + O(\sigma^{-j}) \, . 
\label{eq:pade-log}
\end{align}
Recalling that $\mu = O(\sigma)$ and exponentiating gives~\eqref{eq:pade-convergence}. 

Conversely, if~\eqref{eq:pade-convergence} holds then taking the logarithm gives~\eqref{eq:pade-log}. We then have a telescoping sum
\begin{align*}
\ln b_r 
&= \ln b_0 + \sum_{\ell=1}^r ( \ln b_\ell - \ln b_{\ell-1} ) \\
&= r \ln \sigma + \ln b_0 + \sum_{k=1}^r ( \ln b_\ell - \ln b_{\ell-1} - \ln \sigma ) \, . 
\end{align*}
Now consider the coefficient of $\sigma^{-k}$ in this expression. Recall that $b_\ell$ is a polynomial in $\sigma$ with leading term $\sigma^\ell$, so we can write
\[
\ln b_\ell = \ell \ln \sigma + \sum_{k=1}^\infty c_{\ell,k} \sigma^{-k} \, .
\]
Then 
\[
\ln b_\ell - \ln b_{\ell-1} - \ln \sigma = \sum_{k=1}^\infty (c_{\ell,k} - c_{\ell-1,k}) \sigma^{-k} \, , 
\]
and~\eqref{eq:pade-log} implies that, if $\ell > 2j$, then for all $k \le j$ we have
\[
c_{\ell,k} - c_{\ell-1,k} = -m_k \, . 
\]
Therefore, if we define 
\[
c_k = c_{2j,k} - 2j m_k \, ,
\]
and $r \ge 2j$, then for all $k \le j$ we have
\[
c_{r,k} = c_{2j,k} - (r-2j) m_k = c_k - m_k r \, ,
\]
proving~\eqref{eq:gsr-convergence}.
\end{proof}


\noindent
Note that in addition to showing these two types of stabilization are equivalent, Theorem~\ref{thm:pade-equiv} also proves that the classical approach and the $(\ell-1,1)$ Pad\'e method produce the same series expansions.

For both site and bond percolation, we also computed $(\ell-2,2)$ Pad\'e approximants, i.e., where the numerator has degree $\ell-2$ and the denominator is quadratic in $p$.  In this case the estimate of the threshold is the smallest positive root of the equation 
\[
(b_{\ell-1}^2 - b_\ell b_{\ell-2}) p^2 
+ (b_\ell b_{\ell-3} - b_{\ell-1} b_{\ell-2} ) p
+ (b_{\ell-2}^2 - b_{\ell-1} b_{\ell-3})
= 0 \, . 
\]
It appears that for $j \ge 3$ the coefficient of $\sigma^{-j}$ in $p$ now stabilizes when $\ell=2j$. Thus increasing the degree of the denominator does not appear to help. We note that the $(\ell-1,1)$ Pad\'e approximants for $\ell=1, 2, 3$ are known to be upper bounds on $S$, and therefore that the root of their denominator is a lower bound on the threshold~ \cite{torquato:jiao:13a}.

In any case, these results suggest that the techniques of~\cite{gaunt:sykes:ruskin:76,gaunt:ruskin:78} are not the optimal way to extract the coefficients of the series for $\pcsite$ and $\pcbond$ from enumerations of lattice animals. To put this differently, if we assume that the coefficients of $\log b_r$ are linear in $r$, as they must be for the limit $\log \mu$ to exist, then we don't need to compute the coefficients of $b_r$ to all orders. Proving that the stabilization of Theorem~\ref{thm:pade-equiv} holds strikes us as a very interesting question.

\begin{table}
\centering
$
\arraycolsep=2pt
\def\arraystretch{1.5}
\begin{array}{cl}
\ell & \pcsite \\ \hline
1 & \sigma^{-1} + O(\sigma^{-2}) \\
3 & \sigma^{-1} + \frac{3}{2}\sigma^{-2} + O(\sigma^{-3}) \\
5 & \sigma^{-1} + \frac{3}{2}\sigma^{-2} + \frac{15}{4}\sigma^{-3} + O(\sigma^{-4}) \\
7 & \sigma^{-1} + \frac{3}{2}\sigma^{-2} + \frac{15}{4}\sigma^{-3} + \frac{83}{4}\sigma^{-4} + O(\sigma^{-5}) \\
9 & \sigma^{-1} + \frac{3}{2}\sigma^{-2} + \frac{15}{4}\sigma^{-3} + \frac{83}{4}\sigma^{-4} + \frac{6577}{48} \sigma^{-5} + O(\sigma^{-6}) \\
11 & \sigma^{-1} + \frac{3}{2}\sigma^{-2} + \frac{15}{4}\sigma^{-3} + \frac{83}{4}\sigma^{-4} + \frac{6577}{48} \sigma^{-5} 
+ \frac{\numprint{119077}}{96} \sigma^{-6} + O(\sigma^{-7})
\end{array}
$
\caption{Estimates of $\pcsite$ using Pad\'e approximants. For each $\ell$, construct the $\ell$th-order series expansion~\eqref{eq:br-site} of $S$, compute its Pad\'e approximant with a numerator which is $(\ell-1)$st-order in $p$ and a denominator which is linear in $p$, and expand the root of its denominator in powers of $\sigma^{-1}$. We conjecture that the coefficient of $\sigma^{-j}$ stabilizes when $\ell=2j-1$.}
\label{tab:site-pade}
\end{table}

\begin{table}
\centering
$
\arraycolsep=2pt
\def\arraystretch{1.5}
\begin{array}{cl}
\ell & \pcbond \\ \hline
3 & \sigma^{-1} \\
5 & \sigma^{-1} + \frac{5}{2} \sigma^{-3} + O(\sigma^{-4}) \\
7 & \sigma^{-1} + \frac{5}{2} \sigma^{-3} + \frac{15}{2} \sigma^{-4} + O(\sigma^{-5}) \\
9 & \sigma^{-1} + \frac{5}{2} \sigma^{-3} + \frac{15}{2} \sigma^{-4} + 57 \sigma^{-5} + O(\sigma^{-6}) \\
11 & \sigma^{-1} + \frac{5}{2} \sigma^{-3} + \frac{15}{2} \sigma^{-4} + 57 \sigma^{-5} + \frac{4855}{12} \sigma^{-6} + O(\sigma^{-7})
\end{array}
$
\caption{Estimates of $\pcbond$ using Pad\'e approximants, using the $\ell$th-order series expansion~\eqref{eq:br-bond} of $dpS$. As for site percolation, the coefficient of $\sigma^{-j}$ appears to stabilize when $\ell=2j+1$, reproducing the new term given in~\eqref{eq:pc-bond-series-new} but not yet providing an additional one.}
\label{tab:bond-pade}
\end{table}

\section{Conclusions and Outlook}
\label{sec:conclusion}

\begin{figure}
  \centering
  \includegraphics[width=0.8\linewidth]{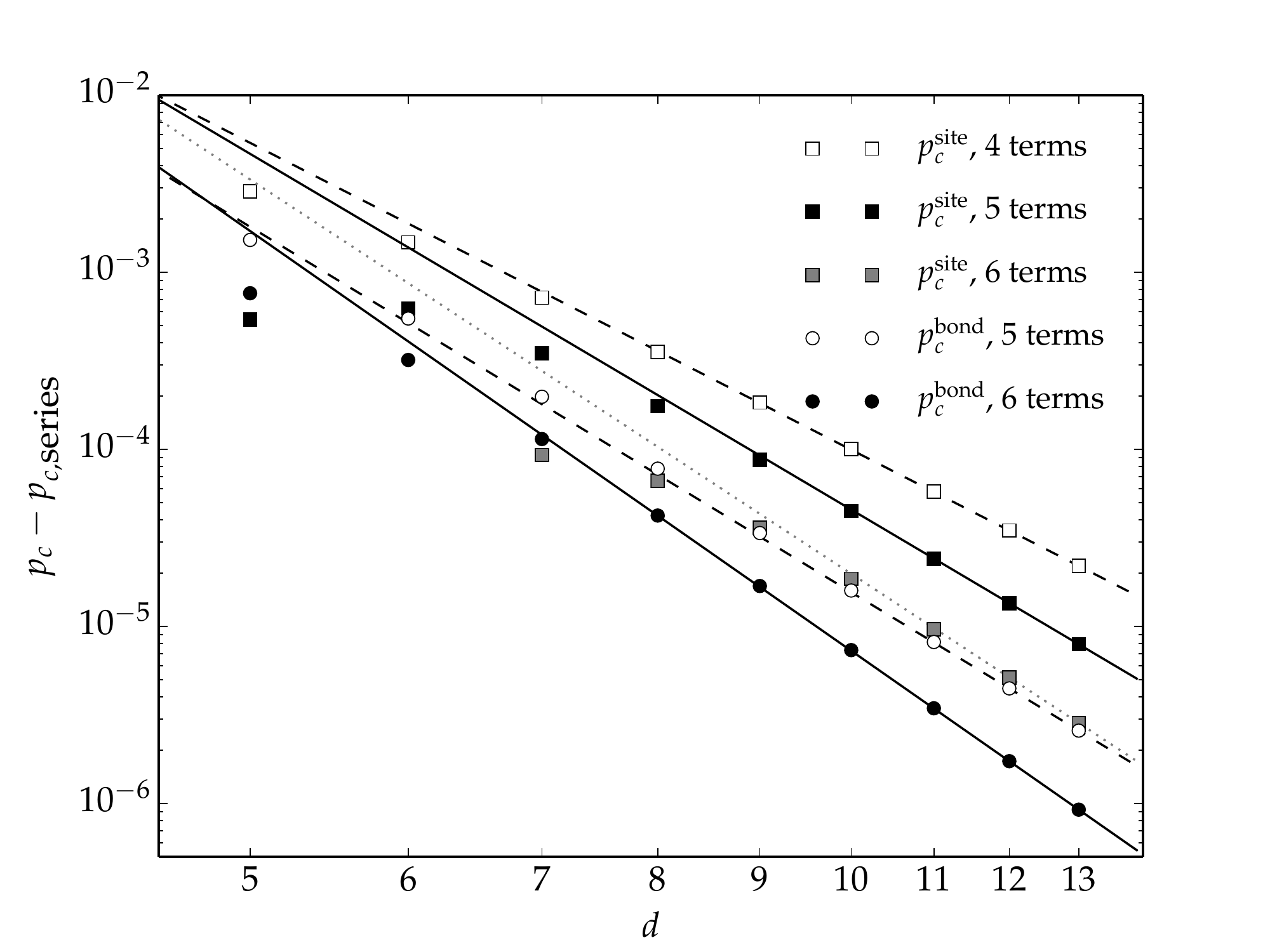}
  \caption{Difference between $p_c$ and the old and new series expansions for site
    and bond percolation.}
  \label{fig:pc-series}
\end{figure}

With a combination of brute force enumeration and new analytical results, we were able to extend the series for $\pcbond$ to the next order, and $\pcsite$ to two more orders. Let us compare the old series and the new, extended series to very precise numerical values for $p_c$ from recent simulations~\cite{mertens:moore:18b}.  

As seen from Figure~\ref{fig:pc-series}, the agreement with $p_c$ is
better for the new series, and the rate of convergence for
$d\to\infty$ increases with each new term. For bond
percolation, the error decreases as $\sigma^{-6.9}$ for the old series
and $\sigma^{-7.9}$ for the new series, and for site percolation the
corresponding rates are $\sigma^{-5.8}$ (old), $\sigma^{-6.7}$ (one
additional term) and $\sigma^{-7.4}$ (two additional terms).
These observations support the claim that these series are
asymptotic,
which is also suggested by the rapid increase in their
coefficients. 

What would it take to get the next term in the series for $\pcbond$, i.e., the coefficient of $\sigma^{-7}$?  Using the techniques of~\cite{gaunt:sykes:ruskin:76,gaunt:ruskin:78} appears to be expensive. Assuming that (continuing the pattern) the coefficient of $\sigma^{-6}$ in~\eqref{eq:br-polys-bond} is cubic in $r$, and that it only holds for $r \ge 12$, to find its coefficients we would need $b_r$ for $r=12, 13, 14, 15$, or equivalently $D_e(q)$ for $e=12, 13, 14$ and $A_d(e)=D_e(1)$ for $e=15$. 
Theorems~\ref{thm:coefficient-e}, \ref{thm:coefficient-e-1} and~\ref{thm:coefficient-e-2} save us from enumerating bond lattice animals of size $e=14$ in $d=12,13,14$, but the remaining enumeration tasks are still prohibitive, as shown in Table~\ref{tab:estimated_times}.

\begin{table}
  \centering
  \begin{tabular}{rrc@{\hspace{20mm}}rr}
    \multicolumn{1}{r}{$(d,e)$} &\multicolumn{1}{c}{wall clock time} & &
    \multicolumn{1}{c}{$(d,e)$} &\multicolumn{1}{c}{wall clock
                                  time}\\[1ex]
    $(5,14)$ & \numprint[\,days]{148} & & $(7,12)$ &
                                                   \numprint[\,days]{57} \\
    $(6,14)$ & \numprint[\,years]{9} & & $(8,12)$& \numprint[\,days]{387}\\
    $(7,14)$ & \numprint[\,years]{116} & & $(9,12)$ &
                                                    \numprint[\,years]{4,5}\\
    $(8,14)$ & \numprint[\,years]{1052} & & & \\
    $(9,14)$ & \numprint[\,years]{5183} & & & \\
    $(10,14)$ & \numprint[\,years]{28526} & & &\\
    $(11,14)$ & \numprint[\,years]{133800} & & &
  \end{tabular}
  \caption{Estimated wall clock times for the enumeration of bond animals to compute the next term of the series~\eqref{eq:pc-bond-series-new} for $\pcbond$, using the classical method of~\cite{gaunt:ruskin:78} (left) and the Pad\'e approximant approach (right). Enumeration data for $( e\leq 4,14)$ and $(e \leq 6,12)$ is already available, and data for $(e \geq 12, 14)$ and $(e \geq 10, 12)$ is provided by Theorems~\ref{thm:coefficient-e}, \ref{thm:coefficient-e-1}, and~\ref{thm:coefficient-e-2}.}
  \label{tab:estimated_times}
\end{table}

However, the situation is less daunting for the Pad\'e approach. For $\pcsite$ we were already able compute the next term from the existing data. For $\pcbond$ we just need to compute $b_{13}$, a task that is within reach of future (or even current) computer machinery (see Table~\ref{tab:estimated_times}).  The same is true for the next term of $\pcsite$: using the Pad\'e approach, we need $b_{13}$ for site animals, which is within reach if one is willing to spend months of computing time on a small cluster.

To avoid the hardest enumeration tasks, one can try to extend the
results of Section~\ref{sec:bond-percolation}, i.e., to analytically
compute the coefficients of ${d \choose e-k}$ in $D_e$ for
$k=3,4,5,\ldots$.  The number of subgraphs that one needs to identify
and analyze grows rapidly with $k$, making it hard to ensure that none
is omitted while avoiding double counting.  The actual computation is
not complicated, but tedious and error-prone.  This suggests
delegating it to a computer. 

This has in fact been done for site
animals, where the formulas for $G_v^{(v-1-k)}$ were derived for
$k \leq 5$ with the help of a computer~\cite{barequet:shalah:17}.
However, that work does not classify animals according to their
perimeter, so it unfortunately does not help us compute the perimeter
polynomials.  Formulas that include the perimeter were derived for
$k=0$ and $k=1$ in~\cite{luther:mertens:17}, but again this gets too
cumbersome to do manually for larger values of $k$. 
As far as we know, the computer-aided approach has not been used to derive formulas for the coefficients of the perimeter polynomials for site or bond percolation beyond~\cite{luther:mertens:17} and Theorems~\ref{thm:coefficient-e}, \ref{thm:coefficient-e-1}, and~\ref{thm:coefficient-e-2}.

The bottom line is that the computation of the next terms
in~\eqref{eq:pc-bond-series-new} and~\eqref{eq:pc-site-series-pade}
along the classical road is far beyond our current computational power
and analytical tools. With the Pad\'e approach, both tasks are within
reach, either by spending a considerable amount of computer time or by
extending the analytical results derived here and in
\cite{luther:mertens:17}. Going beyond that, however, requires a
significant new idea.

\appendix

\section{Perimeter Polynomials $\mathbf{D_e(q)}$ for Bond Animals}
\label{sec:bond-polynomials}

The perimeter polynomials $D_1,\ldots,D_8$ appeared in~\cite{gaunt:ruskin:78}.  Here we also give $D_9$, $D_{10}$ and $D_{11}$, as well as $D_{12}(1)$. We obtained these with our new computations, Theorems~\ref{thm:coefficient-e}, \ref{thm:coefficient-e-1}, and~\ref{thm:coefficient-e-2}, 
and the identity~\eqref{eq:identity-bond}.
{\small
\begin{align*}
  D_1(q) &= q^{4d-2} {d \choose 1} \\ 
  D_2(q) &= q^{6d-4} \left[{d \choose 1} + 4 {d\choose 2}\right]\\
  D_3(q) &= q^{8d-6} \left[{d \choose 1} + (16+4q^{-1}) {d\choose
          2}+32{d \choose 3}\right]\\
  D_4(q) &= q^{10d-8} \left[{d \choose 1} + (53+32q^{-1}+q^{-2d}) {d\choose
          2}+(324+96q^{-1}){d \choose 3}+400{d\choose 4}\right]\\
  D_5(q) &= q^{12d-10} \left[{d \choose 1} + (172+160q^{-1}+30q^{-2}+8q^{-2d}) {d\choose
          2}\right.\\
  &+\left.(\numprint{2448} + \numprint{1512} q^{-1} + 180q^{-2} + 24 q^{-2d}){d \choose
  3}+(\numprint{8064} + \numprint{2304} q^{-1}){d\choose 4}+\numprint{6912} {d\choose 5}\right]\\
  D_6(q) &= q^{14d-12} \left[{d \choose 1} + (568+672q^{-1}+332q^{-2}+40q^{-2d}+14q^{-2d-1}) {d\choose
          2}\right.\\
  &+(\numprint{17041}+\numprint{15600}q^{-1}+\numprint{4704}q^{-2}+400q^{-3}+376q^{-2d}+84q^{-2d-1}){d \choose
  3}\\
  &+(\numprint{112824} + \numprint{63744}q^{-1} + \numprint{9408}q^{-2} +
  576 q^{-2d}){d\choose 4}\\
  &+\left.(\numprint{239120} + \numprint{62720} q^{-1}){d\choose 5} +
  \numprint{153664} {d\choose 6}\right]\\
D_7(q) &= q^{16d-14}  \left[{d \choose 1} + (\numprint{1906}+\numprint{2712}q^{-1}+\numprint{2030}q^{-2}+336q^{-3}+168q^{-2d}+156q^{-2d-1}+2q^{-4d}) {d\choose
          2}\right.\\
  &+(\numprint{116004}+\numprint{137736}q^{-1}+\numprint{67812}q^{-2}+\numprint{15096}q^{-3}+384q^{-5}+\numprint{3864}q^{-2d}\\
  &+\numprint{2208}q^{-2d-1}+264q^{-2d-2}+12q^{-4d}){d \choose
  3}\\
  &+(\numprint{1382400}+\numprint{1141248}q^{-1}+\numprint{350400}q{-2} +\numprint{40256}q^{-3}
  +\numprint{15840}q^{-2d}+\numprint{4416}q^{-2d-1}){d\choose 4}\\
  &+(\numprint{5445120} + \numprint{2769920} q^{-1} + \numprint{407040} q^{-2} + \numprint{15680} q^{-2d}){d\choose 5} \\
  &+\left.(\numprint{8257536} + \numprint{1966080}
  q^{-1}){d\choose 6} + \numprint{4194304} {d\choose 7}\right] \\
D_8(q) &= q^{18d-16} \left[{d \choose 1} +
  (\numprint{6471}+\numprint{10880}q^{-1}+\numprint{9972}q^{-2}+\numprint{4064}q^{-3}+192q^{-4}+677q^{-2d}\right.\\
  &+958q^{-2d-1}+228q^{-2d-2}+22q^{-4d}) {d\choose
          2}\\
  &+(\numprint{787965}+\numprint{1140576}q^{-1}+\numprint{755532}q^{-2}+\numprint{287280}q^{-3}+\numprint{28704}q^{-4}+\numprint{9216}q^{-5}+\numprint{33996}q^{-2d}\\
  &+\numprint{31908}q^{-2d-1}+\numprint{10080}q^{-2d-2}+408q^{-2d-4}+312q^{-4d}+72q^{-4d-1}){d
    \choose 3}\\
  &+(\numprint{15998985}+\numprint{17116800}q^{-1}+\numprint{7855008}q{-2} +\numprint{1932864}q^{-3}
  +\numprint{114816}q^{-4}\\
  &+\numprint{24576}q^{-5}+\numprint{282216}q^{-2d}+\numprint{164928}q^{-2d-1}+\numprint{26880}q^{-2d-2} + 624q^{-4d}){d\choose 4}\\
  &+(\numprint{104454120}+\numprint{77177280}q^{-1}+\numprint{22232640}q^{-2}+\numprint{2609280}q^{-3}+\numprint{688640}q^{-2d}+\numprint{192000}q^{-2d-1}){d\choose 5} \\
  &+(\numprint{280717488} +\numprint{128770560}q^{-1} +\numprint{17729280}q^{-2} +\numprint{491520}q^{-2d}){d\choose 6}\\
  &+\left.(\numprint{326265408} +\numprint{70543872}q^{-1}){d\choose 7} +\numprint{136048896} {d\choose 8}\right]\\
D_9(q) &= q^{20d - 18} \left[ {d \choose 1}
+(\numprint{22200}+\numprint{43220}q^{-1}+\numprint{46004}q^{-2}+\numprint{27392}q^{-3}+\numprint{6062}q^{-4}+\numprint{2708}q^{-2d}\right.\\
&+\numprint{4724}q^{-2d-1}+\numprint{2776}q^{-2d-2}+164q^{-2d-3}+134q^{-4d}+60q^{-4d-1}) {d
  \choose 2} \\
&+(\numprint{5380600}+\numprint{9167304}q^{-1}+\numprint{7470900}q^{-2}+\numprint{3904416}q^{-3}+\numprint{952659}q^{-4}+\numprint{167760}q^{-5}+\numprint{33024}q^{-6}\\
&+\numprint{280608}q^{-2d}+\numprint{355860}q^{-2d-1}+\numprint{193440}q^{-2d-2}+\numprint{24720}q^{-2d-3}+\numprint{9792}q^{-2d-4}+\numprint{4452}q^{-4d}\\
&+\numprint{2712}q^{-4d-1}+212q^{-4d-3}+8q^{-6d}) {d \choose 3}\\
&+(\numprint{180558848}+\numprint{235351008}q^{-1}+\numprint{140954400}q^{-2}+\numprint{52264576}q^{-3}+\numprint{8518224}q^{-4}+\numprint{1256064}q^{-5}+\numprint{132096}q^{-6}\\
&+\numprint{4215664}q^{-2d}+\numprint{3700320}q^{-2d-1}+\numprint{1299648}q^{-2d-2}+\numprint{98880}q^{-2d-3}+\numprint{26112}q^{-2d-4}\\
&+\numprint{23040}q^{-4d}+\numprint{7232}q^{-4d-1}) {d \choose 4}\\
&+(\numprint{1839569920}+\numprint{1758624960}q^{-1}+\numprint{736709440}q^{-2}+\numprint{171765760}q^{-3}+\numprint{13580560}q^{-4}+\numprint{1228800}q^{-5}\\
&+\numprint{19098480}q^{-2d}+\numprint{10501440}q^{-2d-1}+\numprint{1756800}q^{-2d-2}+\numprint{26880}q^{-4d}) {d \choose 5}\\
&+(\numprint{7801139200}+\numprint{5187225600}q^{-1}+\numprint{1371264000}q^{-2}+\numprint{151859200}q^{-3}+\numprint{32037120}q^{-2d}+\numprint{8398080}q^{-2d-1}) {d \choose 6}\\
&+(\numprint{15572480000} +\numprint{6478080000}q^{-1} +\numprint{819840000}q^{-2} +\numprint{17635968}q^{-2d}) {d \choose 7}\\
&+\left.(\numprint{14540800000} +\numprint{2867200000}q^{-1}) {d \choose 8}
+\numprint{5120000000} {d \choose 9}
\right]\\
D_{10}(q) &= q^{22d - 20} \left[ {d \choose 1}
+(\numprint{76884}+\numprint{169784}q^{-1}+\numprint{207444}q^{-2}+\numprint{148728}q^{-3}+\numprint{63852}q^{-4}+\numprint{5696}q^{-5}+\numprint{10724}q^{-2d}\right.\\
  &+\numprint{21844}q^{-2d-1}+\numprint{18816}q^{-2d-2}+\numprint{5308}q^{-2d-3}+656q^{-4d}+728q^{-4d-1}+62q^{-4d-2}+6q^{-6d})
  {d \choose 2}\\
 &+(\numprint{37034319}+\numprint{72525600}q^{-1}+\numprint{69548916}q^{-2}+\numprint{44680224}q^{-3}+\numprint{17341872}q^{-4}+\numprint{3987504}q^{-5}+\numprint{919680}q^{-6}\\
 &+\numprint{104880}q^{-7}+\numprint{2248620}q^{-2d}+\numprint{3520212}q^{-2d-1}+\numprint{2643504}q^{-2d-2}+\numprint{827214}q^{-2d-3}+\numprint{177060}q^{-2d-4}+\numprint{41280}q^{-2d-5}\\
 &+\numprint{49062}q^{-4d}+\numprint{51552}q^{-4d-1}+\numprint{9480}q^{-4d-2}+\numprint{5088}q^{-4d-3}+288q^{-6d}+66q^{-6d-2}) {d \choose 3}\\
 &+(\numprint{2017563224}+\numprint{3087285504}q^{-1}+\numprint{2252598408}q^{-2}+\numprint{1084658496}q^{-3}+\numprint{299992512}q^{-4}+\numprint{54960192}q^{-5}\\
 &+\numprint{9445632}q^{-6}+\numprint{839040}q^{-7}+\numprint{57763584}q^{-2d}+\numprint{66404448}q^{-2d-1}+\numprint{35316960}q^{-2d-2}+\numprint{7383456}q^{-2d-3}\\
 &+\numprint{1325088}q^{-2d-4}+\numprint{165120}q^{-2d-5}+\numprint{510768}q^{-4d}+\numprint{346816}q^{-4d-1}+\numprint{37920}q^{-4d-2}+\numprint{13568}q^{-4d-3}+768q^{-6d}) {d \choose 4}\\
 &+(\numprint{30937530481}+\numprint{35996081120}q^{-1}+\numprint{19450068320}q^{-2}+\numprint{6580542880}q^{-3}+\numprint{1145908480}q^{-4}+\numprint{134990208}q^{-5}\\
 &+\numprint{13701120}q^{-6}
  +\numprint{433509696}q^{-2d}+\numprint{348072000}q^{-2d-1}+\numprint{116214720}q^{-2d-2}+\numprint{11787040}q^{-2d-3}+\numprint{1305600}q^{-2d-4}\\
 &+\numprint{1454160}q^{-4d}+\numprint{469120}q^{-4d-1}) {d \choose 5}\\
 & +(\numprint{194498568156}+\numprint{167245290240}q^{-1}+\numprint{63682584960}q^{-2}+\numprint{13587383040}q^{-3}+\numprint{1122984960}q^{-4}\\
  &+\numprint{59473920}q^{-5}+\numprint{1285228800}q^{-2d}+\numprint{650035200}q^{-2d-1}+\numprint{102950400}q^{-2d-2}+\numprint{1166400}q^{-4d}) {d \choose 6}\\
 &+(\numprint{593322510704}+\numprint{357772800000}q^{-1}+\numprint{86356596480}q^{-2}+\numprint{8787116800}q^{-3}+\numprint{1612800000}q^{-2d}+\numprint{389760000}q^{-2d-1}) {d \choose 7}\\
 &+(\numprint{930918351616} +\numprint{353460445184}q^{-1} +\numprint{40929208320}q^{-2} +\numprint{716800000}q^{-2d}) {d \choose 8}\\
 & +\left.(\numprint{722456748288} +\numprint{130613649408}q^{-1}) {d \choose 9}
 +\numprint{219503494144} {d \choose 10}
 \right]\\
D_{11}(q) &= q^{24d - 22} \left[ {d \choose 1}
+(\numprint{268350}+\numprint{662424}q^{-1}+\numprint{912378}q^{-2}+\numprint{755936}q^{-3}+\numprint{435330}q^{-4}+\numprint{111112}q^{-5}+\numprint{4830}q^{-6}\right.\\
&+\numprint{42012}q^{-2d}+\numprint{98596}q^{-2d-1}+\numprint{102660}q^{-2d-2}+\numprint{56496}q^{-2d-3}+\numprint{6032}q^{-2d-4}\\
&+\numprint{3008}q^{-4d}+\numprint{4920}q^{-4d-1}+\numprint{2000}q^{-4d-2}+72q^{-6d}+12q^{-6d-1})
  {d \choose 2} \\
&+(\numprint{257091447}+\numprint{568629480}q^{-1}+\numprint{624866154}q^{-2}+\numprint{467409000}q^{-3}+\numprint{238715907}q^{-4}+\numprint{80060424}q^{-5}\\
&+\numprint{19650432}q^{-6}+\numprint{4958880}q^{-7}+\numprint{94500}q^{-9}+\numprint{17740860}q^{-2d}+\numprint{32773380}q^{-2d-1}+\numprint{30367248}q^{-2d-2}\\
&+\numprint{15154164}q^{-2d-3}+\numprint{4208376}q^{-2d-4}+\numprint{1151064}q^{-2d-5}+\numprint{148008}q^{-2d-6}+\numprint{479952}q^{-4d}+\numprint{698808}q^{-4d-1}\\
&+\numprint{315738}q^{-4d-2}+\numprint{90744}q^{-4d-3}+\numprint{25608}q^{-4d-4}+\numprint{5328}q^{-6d}+\numprint{1872}q^{-6d-1}+\numprint{1584}q^{-6d-2}+12q^{-8d-1})
  {d \choose 3}\\
&+(\numprint{22494953744}+\numprint{39420410688}q^{-1}+\numprint{33669058848}q^{-2}+\numprint{19535663616}q^{-3}+\numprint{7457884848}q^{-4}+\numprint{1932787968}q^{-5}\\
&+\numprint{394406080}q^{-6}+\numprint{71993664}q^{-7}+\numprint{3443136}q^{-8}+\numprint{756000}q^{-9}
  +\numprint{755437872}q^{-2d}+\numprint{1060971144}q^{-2d-1}\\
&+\numprint{735484416}q^{-2d-2}+\numprint{261507552}q^{-2d-3}+\numprint{57870624}q^{-2d-4}+\numprint{11783232}q^{-2d-5}+\numprint{1184064}q^{-2d-6}\\
&+\numprint{9062784}q^{-4d}+\numprint{9351648}q^{-4d-1}+\numprint{2819616}q^{-4d-2}+\numprint{678464}q^{-4d-3}+\numprint{102432}q^{-4d-4}\\
&+\numprint{35968}q^{-6d}+\numprint{7488}q^{-6d-1}+\numprint{4224}q^{-6d-2}) {d \choose 4}\\
&+(\numprint{507201540240}+\numprint{691805061120}q^{-1}+\numprint{452798848800}q^{-2}+\numprint{195805808000}q^{-3}+\numprint{52407897360}q^{-4}\\
&+\numprint{9325311360}q^{-5}+\numprint{1352772160}q^{-6}+\numprint{135333120}q^{-7}+\numprint{8844012400}q^{-2d}+\numprint{9187994080}q^{-2d-1}\\
&+\numprint{4468344240}q^{-2d-2}+\numprint{999695680}q^{-2d-3}+\numprint{142520640}q^{-2d-4}+\numprint{17157120}q^{-2d-5}+\numprint{47677760}q^{-4d}\\
&+\numprint{30820800}q^{-4d-1}+\numprint{4505760}q^{-4d-2}+\numprint{678400}q^{-4d-3}+\numprint{48640}q^{-6d})
  {d \choose 5} \\
&+(\numprint{4548861718272}+\numprint{4758841658880}q^{-1}+\numprint{2326927299840}q^{-2}+\numprint{710854571520}q^{-3}+\numprint{119621445120}q^{-4}\\
&+\numprint{12373857792}q^{-5}+\numprint{1021870080}q^{-6}+\numprint{41293532640}q^{-2d}+\numprint{30190824960}q^{-2d-1}\\
&+\numprint{9246453120}q^{-2d-2}+\numprint{981461760}q^{-2d-3}+\numprint{63191040}q^{-2d-4}+\numprint{89376000}q^{-4d}+\numprint{27340800}q^{-4d-1}) {d \choose 6} \\
&+(\numprint{19903875199488}+\numprint{15525985886208}q^{-1}+\numprint{5383330219008}q^{-2}+\numprint{1042574561280}q^{-3}+\numprint{83366115840}q^{-4}\\
&+\numprint{2972712960}q^{-5}+\numprint{88752861120}q^{-2d}+\numprint{41074268928}q^{-2d-1}+\numprint{5992694400}q^{-2d-2}+\numprint{53760000}q^{-4d}) {d \choose 7} \\
&+(\numprint{46672464052224}+\numprint{25712480747520}q^{-1}+\numprint{5672976777216}q^{-2}+\numprint{526170193920}q^{-3}+\numprint{88050271232}q^{-2d}\\
&+\numprint{19520083968}q^{-2d-1})
  {d \choose 8} \\
&+(\numprint{59894730326016}+\numprint{20886790864896}q^{-1}+\numprint{2215265697792}q^{-2} +\numprint{32653412352}q^{-2d} ) {d \choose 9}\\
&+\left.(\numprint{39627113103360} +\numprint{6604518850560}q^{-1}) {d \choose 10} +\numprint{10567230160896} {d \choose 11}\right]\\
D_{12}(1) &= {d \choose 1}
+\numprint{16576872} {d \choose 2} 
+\numprint{22086892828} {d \choose 3} 
+\numprint{1825033692350} {d \choose 4}
+\numprint{38707129124945} {d \choose 5} \\
&+\numprint{341855549212957} {d \choose 6} 
+\numprint{1554749521671500} {d \choose 7} 
+\numprint{4030548636699744} {d \choose 8} 
+\numprint{6199599459637248} {d \choose 9} \\
&+\numprint{5601509502521600} {d \choose 10} 
+\numprint{2747328861561856} {d \choose 11}
+\numprint{564668382613504} {d \choose 12}
\end{align*}
}

\section{The Polynomials $b_r$ for Bond Percolation}
\label{sec:br}

Here we give $b_r$, the coefficients of $p^r$ in the series expansion for $dpS$ defined in~\eqref{eq:br-bond}. The leading term $\sigma^r$ corresponds to the mean-field behavior.
{\small
\begin{align*}
b_1 &= \frac{\sigma}{2}+\frac{1}{2} \\
b_2 &= \sigma^2+\sigma \\
b_3 &= \sigma^3+\sigma^2 \\
b_4 &= \sigma^4+\sigma^3-\frac{3 \,\sigma^2}{2}+\frac{3}{2} \\
b_5 &= \sigma^5+\sigma^4-4 \,\sigma^3+3 \,\sigma^2+4 \,\sigma-3 \\
b_6 &= \sigma^6+\sigma^5-\frac{13 \,\sigma^4}{2}-3 \,\sigma^3+\frac{55 \,\sigma^2}{2}+3 \,\sigma-21 \\
b_7 &= \sigma^7+\sigma^6-9 \,\sigma^5-13 \,\sigma^4+98 \,\sigma^3-68 \,\sigma^2-89 \,\sigma +81 \\
b_8 &= \sigma^8+\sigma^7-\frac{23 \,\sigma^6}{2}-23 \,\sigma^5+89 \,\sigma^4+\frac{733 \,\sigma^3}{2}-\frac{1859 \,\sigma^2}{2}-\frac{687 \,\sigma}{2}+852 \\
b_9 &= \sigma^9+\sigma^8-14 \,\sigma^7-33 \,\sigma^6+\frac{119 \,\sigma^5}{2}+\frac{2747 \,\sigma^4}{2}-4383 \,\sigma^3+\frac{4455 \,\sigma^2}{2}+\frac{8675 \,\sigma}{2}-3568 \\
b_{10} &= \sigma^{10}+\sigma^9-\frac{33 \,\sigma^8}{2}-43 \,\sigma^7+\frac{145 \,\sigma^6}{4}+1389 \,\sigma^5+\frac{\numprint{12097} \,\sigma^4}{4}-\frac{\numprint{146053} \,\sigma^3}{4}+\numprint{51923} \,\sigma^2 \\
&+\frac{\numprint{140669} \,\sigma}{4}-\numprint{54967} \\
b_{11} &= \sigma^{11}+\sigma^{10}-19 \,\sigma^9-53 \,\sigma^8+\frac{77 \,\sigma^7}{4}+\frac{14429 \,\sigma^6}{12}+\frac{72193 \,\sigma^5}{4}-\frac{\numprint{543935} \,\sigma^4}{4}+\frac{\numprint{3342607} \,\sigma^3}{12} \\&-\frac{789173 \,\sigma^2}{12}-\frac{\numprint{3559189} \,\sigma}{12}+\frac{\numprint{802395}}{4} \\
b_{12} &= \sigma^{12}+\sigma^{11}-\frac{43 \,\sigma^{10}}{2}-63 \,\sigma^9+\frac{17 \,\sigma^8}{2}+\frac{12715 \,\sigma^7}{12}+\frac{\numprint{111581} \,\sigma^6}{6}+\frac{7897 \,\sigma^5}{12}-\frac{\numprint{6133171} \,\sigma^4}{6} \\
&+\frac{\numprint{49710709} \,\sigma^3}{12}-\frac{\numprint{11934461} \,\sigma^2}{3}-\frac{\numprint{16576855} \,\sigma}{4}+\numprint{4981765}
\end{align*}
}

\section{Perimeter Polynomials $D_v(q)$ for Site Percolation}
\label{sec:site-polynomials}

The perimeter polynomials $D_2, \ldots, D_7$ for site percolation appeared in~\cite{gaunt:sykes:ruskin:76}. We computed $D_8, \ldots, D_{12}$ and $D_{13}(1)$ from the enumerations in~\cite{luther:mertens:11a} and
the analytical results in~\cite{luther:mertens:17}, including the identity~\eqref{eq:identity-site}.
 
{\small
\begin{align*}
  D_2(q) &= q^{4d-2} {d \choose 1} \\
  D_{3}(q) &= q^{6d-4}\left[{d\choose 1}+4q^{-1} {d \choose 2}
              \right]\\
  D_{4}(q) &= q^{8d-6}\left[{d\choose 1}+(9q^{-2}+8q^{-1}) {d \choose
              2}+ (8q^{-3}+24q^{-2}) {d \choose 3} \right]\\
D_{5}(q) &= q^{10d-8}\left[{d\choose
            1}+(q^{-4}+20q^{-3}+28q^{-2}+12q^{-1}) {d \choose 2}+
            (12q^{-5}+96q^{-4}+168q^{-3}+72q^{-2}) {d \choose 3}\right.\\
          &\left.+ (16q^{-6}+192q^{-4}+192q^{-3}) {d \choose 4} \right]\\
D_{6}(q) &= q^{12d-10}\left[{d\choose
            1}+(4q^{-5}+54q^{-4}+80q^{-3}+60q^{-2}+16q^{-1}) {d
            \choose 2}+
            (6q^{-8}+280q^{-6}+720q^{-5}\right.\\
            &+966q^{-4}+720q^{-3}+144q^{-2})
            {d \choose 3}+
            (32q^{-9}+288q^{-7}+\numprint{1504}q^{-6}+\numprint{2784}q^{-5}+\numprint{3264}q^{-4}+768q^{-3})
            {d \choose 4}\\
           &+
            (32q^{-10}+640q^{-7}+480q^{-6}+\numprint{3840}q^{-5}+\left.\numprint{1920}q^{-4})
            {d \choose 5} \right]
\end{align*}
\begin{align*}
D_{7}(q) &= q^{14d-12}\left[{d\choose
            1}+(22q^{-6}+136q^{-5}+252q^{-4}+228q^{-3}+100q^{-2}+20q^{-1})
            {d \choose 2}\right.\\
           &+
             (q^{-12}+72q^{-9}+662q^{-8}+\numprint{2496}q^{-7}+\numprint{4924}q^{-6}+\numprint{6024}q^{-5}+\numprint{4926}q^{-4}+\numprint{1880}q^{-3}+240q^{-2})
             {d \choose 3}\\
           &+
             (24q^{-13}+672q^{-10}+\numprint{2288}q^{-9}+\numprint{10320}q^{-8}+\numprint{25440}q^{-7}+\numprint{36256}q^{-6}+\numprint{36624}q^{-5}+\numprint{15744}q^{-4}+\numprint{1920}q^{-3})
             {d \choose 4}\\
           &+
             (80q^{-14}+960q^{-11}+\numprint{1760}q^{-10}+\numprint{7680}q^{-9}+\numprint{23680}q^{-8}+\numprint{57920}q^{-7}+\numprint{89760}q^{-6}+\numprint{63360}q^{-5}+\numprint{9600}q^{-4})
             {d \choose 5}\\
&+
  (64q^{-15}+\numprint{1920}q^{-11}+\numprint{3840}q^{-9}+\numprint{19200}q^{-8}+\numprint{28800}q^{-7}+\left.\numprint{76800}q^{-6}+\numprint{23040}q^{-5})
  {d \choose 6} \right]
\end{align*}
\begin{align*}
D_{8}(q) &= q^{16d-14}\left[{d\choose
            1}+(4q^{-8}+80q^{-7}+388q^{-6}+777q^{-5}+818q^{-4}+480q^{-3}+152q^{-2}+24q^{-1})
            {d \choose 2}\right.\\
&+
            (12q^{-13}+6q^{-12}+288q^{-11}+\numprint{2089}q^{-10}+\numprint{8340}q^{-9}+\numprint{20304}q^{-8}+\numprint{33072}q^{-7}+\numprint{38148}q^{-6}+\numprint{32304}q^{-5}\\
&+\numprint{16002}q^{-4}+\numprint{3816}q^{-3}+360q^{-2})
            {d \choose 3}\\
&+
            (8q^{-18}+552q^{-14}+192q^{-13}+\numprint{5216}q^{-12}+\numprint{22096}q^{-11}+\numprint{78024}q^{-10}+\numprint{190016}q^{-9}+\numprint{331584}q^{-8}\\
&+\numprint{420000}q^{-7}+\numprint{385568}q^{-6}+\numprint{203712}q^{-5}+\numprint{47616}q^{-4}+\numprint{3840}q^{-3})
            {d \choose 4}\\
&+
            (80q^{-19}+\numprint{2720}q^{-15}+\numprint{1920}q^{-14}+\numprint{8640}q^{-13}+\numprint{43040}q^{-12}+\numprint{128880}q^{-11}+\numprint{372800}q^{-10}+\numprint{854720}q^{-9}\\
&+\numprint{1444960}q^{-8}+\numprint{1712160}q^{-7}+\numprint{1220160}q^{-6}+\numprint{341760}q^{-5}+\numprint{28800}q^{-4})
            {d \choose 5}\\
&+
            (192q^{-20}+\numprint{2880}q^{-16}+\numprint{6144}q^{-15}+\numprint{30720}q^{-13}+\numprint{90240}q^{-12}+\numprint{180480}q^{-11}+\numprint{547200}q^{-10}\\
&+\numprint{1324800}q^{-9}+\numprint{2304000}q^{-8}+\numprint{2833920}q^{-7}+\numprint{1290240}q^{-6}+\numprint{138240}q^{-5})
            {d \choose 6}
+ (128q^{-21}+\numprint{5376}q^{-16}\\
&+\numprint{13440}q^{-13}+\numprint{89600}q^{-12}+\numprint{322560}q^{-10}+\numprint{618240}q^{-9}+\left.\numprint{1209600}q^{-8}+\numprint{1612800}q^{-7}+\numprint{322560}q^{-6})
            {d \choose 7} \right]
\end{align*}
\begin{align*}
D_{9}(q) &= q^{18d-16}\left[{d\choose
            1}+(28q^{-9}+291q^{-8}+\numprint{1152}q^{-7}+\numprint{2444}q^{-6}+\numprint{2804}q^{-5}+\numprint{2089}q^{-4}+856q^{-3}+216q^{-2}+28q^{-1})
            {d \choose 2}\right.\\
&+
            (48q^{-15}+90q^{-14}+\numprint{1284}q^{-13}+\numprint{7415}q^{-12}+\numprint{30600}q^{-11}+\numprint{79512}q^{-10}+\numprint{154852}q^{-9}+\numprint{225636}q^{-8}+\numprint{247020}q^{-7}\\
&+\numprint{210372}q^{-6}+\numprint{120048}q^{-5}+\numprint{39018}q^{-4}+\numprint{6744}q^{-3}+504q^{-2})
            {d \choose 3}
            +(q^{-24}+192q^{-19}+56q^{-18}\\
&+\numprint{4128}q^{-16}+\numprint{9752}q^{-15}+\numprint{45720}q^{-14}+\numprint{205960}q^{-13}+\numprint{622008}q^{-12}+\numprint{1461712}q^{-11}+\numprint{2701512}q^{-10}\\
&+\numprint{4016208}q^{-9}+\numprint{4580136}q^{-8}+\numprint{3980496}q^{-7}+\numprint{2291360}q^{-6}+\numprint{724368}q^{-5}+\numprint{111744}q^{-4}+\numprint{6720}q^{-3})
            {d \choose 4}\\
&+
            (40q^{-25}+\numprint{2960}q^{-20}+\numprint{1120}q^{-19}+\numprint{26880}q^{-17}+\numprint{75120}q^{-16}+\numprint{172880}q^{-15}+\numprint{640800}q^{-14}+\numprint{2095840}q^{-13}\\
&+\numprint{5320080}q^{-12}+\numprint{11453440}q^{-11}+\numprint{20104160}q^{-10}+\numprint{27853120}q^{-9}+\numprint{28840320}q^{-8}+\numprint{19990880}q^{-7}+\numprint{7239360}q^{-6}\\
&+\numprint{1171200}q^{-5}+\numprint{67200}q^{-4})
            {d \choose 5}+
            (240q^{-26}+\numprint{9600}q^{-21}+\numprint{8064}q^{-20}+\numprint{34560}q^{-18}+\numprint{159360}q^{-17}\\
&+\numprint{330240}q^{-16}+\numprint{624384}q^{-15}+\numprint{2376960}q^{-14}+\numprint{5865600}q^{-13}+\numprint{14515200}q^{-12}+\numprint{31011840}q^{-11}+\numprint{56841600}q^{-10}\\
&+\numprint{78712320}q^{-9}+\numprint{76510080}q^{-8}+\numprint{38718720}q^{-7}+\numprint{7718400}q^{-6}+\numprint{483840}q^{-5})
            {d \choose 6}\\
&+
            (448q^{-27}+\numprint{8064}q^{-22}+\numprint{19712}q^{-21}+\numprint{107520}q^{-18}+\numprint{309120}q^{-17}+\numprint{454272}q^{-16}+\numprint{555520}q^{-15}+\numprint{3373440}q^{-14}\\
&+\numprint{6800640}q^{-13}+\numprint{13879040}q^{-12}+\numprint{35750400}q^{-11}+\numprint{66312960}q^{-10}+\numprint{97386240}q^{-9}+\numprint{88623360}q^{-8}\\
&+\numprint{28062720}q^{-7}+\numprint{2257920}q^{-6})
            {d \choose 7}\\
&+
            (256q^{-28}+\numprint{14336}q^{-22}+\numprint{43008}q^{-18}+\numprint{301056}q^{-17}+\numprint{71680}q^{-16}
+\numprint{1505280}q^{-14}+\numprint{4014080}q^{-13}\\
&+\numprint{1505280}q^{-12}+\numprint{18063360}q^{-11}+\numprint{24084480}q^{-10}+\numprint{45158400}q^{-9}+\numprint{36126720}q^{-8}+\left.\numprint{5160960}q^{-7})
            {d \choose 8} \right]
\end{align*}
\begin{align*}
D_{10}(q) &= q^{20d-18}\left[{d\choose
             1}+(4q^{-11}+154q^{-10}+986q^{-9}+\numprint{3676}q^{-8}+\numprint{7612}q^{-7}+\numprint{9750}q^{-6}+\numprint{8192}q^{-5}+\numprint{4330}q^{-4}+\numprint{1416}q^{-3}\right.\\
&+292q^{-2}+32q^{-1})
             {d \choose 2} +
             (212q^{-17}+753q^{-16}+\numprint{5224}q^{-15}+\numprint{32084}q^{-14}+\numprint{115836}q^{-13}+\numprint{323100}q^{-12}\\
&+\numprint{690016}q^{-11}
+\numprint{1163448}q^{-10}+\numprint{1547364}q^{-9}+\numprint{1638078}q^{-8}+\numprint{1383084}q^{-7}+\numprint{854808}q^{-6}+\numprint{339288}q^{-5}\\
&+\numprint{80556}q^{-4}+\numprint{10880}q^{-3}+672q^{-2})
             {d \choose 3}+
             (24q^{-25}+8q^{-24}+\numprint{1536}q^{-21}+\numprint{2112}q^{-20}+\numprint{4560}q^{-19}\\
&+\numprint{33616}q^{-18}+\numprint{134312}q^{-17}+\numprint{515824}q^{-16}+\numprint{1865120}q^{-15}+\numprint{5178896}q^{-14}+\numprint{11710048}q^{-13}+\numprint{22170392}q^{-12}\\
&+\numprint{35115760}q^{-11}+\numprint{46183200}q^{-10}+\numprint{48898512}q^{-9}+\numprint{40778400}q^{-8}+\numprint{24415968}q^{-7}+\numprint{9171024}q^{-6}\\
&+\numprint{1970304}q^{-5}+\numprint{224448}q^{-4}+\numprint{10752}q^{-3})
             {d \choose 4}\\
&+
             (10q^{-32}+\numprint{1560}q^{-26}+400q^{-25}+\numprint{31920}q^{-22}+\numprint{57200}q^{-21}+\numprint{77640}q^{-20}+\numprint{326880}q^{-19}+\numprint{1278600}q^{-18}\\
&+\numprint{3586000}q^{-17}+\numprint{10522734}q^{-16}+\numprint{30712960}q^{-15}+\numprint{73534720}q^{-14}+\numprint{150801360}q^{-13}+\numprint{266010272}q^{-12}\\
&+\numprint{393996960}q^{-11}
+\numprint{478780720}q^{-10}+\numprint{451891520}q^{-9}+\numprint{305368240}q^{-8}\\
&+\numprint{125711680}q^{-7}+\numprint{28044480}q^{-6}+\numprint{3102720}q^{-5}+\numprint{134400}q^{-4})
             {d \choose 5}\\
&+
             (160q^{-33}+\numprint{12960}q^{-27}+\numprint{5760}q^{-26}+\numprint{124800}q^{-23}+\numprint{340800}q^{-22}+\numprint{413760}q^{-21}\\
&+\numprint{822656}q^{-20}+\numprint{3081120}q^{-19}+\numprint{9619840}q^{-18}+\numprint{21003360}q^{-17}+\numprint{54009600}q^{-16}+\numprint{144823424}q^{-15}\\
&+\numprint{324407040}q^{-14}+\numprint{661888320}q^{-13}+\numprint{1171213440}q^{-12}+\numprint{1748548800}q^{-11}+\numprint{2065014720}q^{-10}\\
&+\numprint{1772025600}q^{-9}+\numprint{925670400}q^{-8}+\numprint{244684800}q^{-7}+\numprint{29529600}q^{-6}+\numprint{1290240}q^{-5})
             {d \choose 6}\\
&+
             (672q^{-34}+\numprint{30912}q^{-28}+\numprint{30464}q^{-27}+\numprint{120960}q^{-24}+\numprint{634368}q^{-23}+\numprint{1128960}q^{-22}\\
&+\numprint{924672}q^{-21}+\numprint{2587200}q^{-20}+\numprint{10257408}q^{-19}+\numprint{24990560}q^{-18}+\numprint{44491776}q^{-17}+\numprint{108450048}q^{-16}\\
&+\numprint{291737600}q^{-15}+\numprint{602112000}q^{-14}+\numprint{1250027520}q^{-13}+\numprint{2340732800}q^{-12}+\numprint{3597807360}q^{-11}\\
&+\numprint{4256985600}q^{-10}+\numprint{3287208960}q^{-9}+\numprint{1225244160}q^{-8}+\numprint{184504320}q^{-7}+\numprint{9031680}q^{-6})
             {d \choose 7}\\
&+
             (\numprint{1024}q^{-35}+\numprint{21504}q^{-29}+\numprint{59392}q^{-28}+\numprint{344064}q^{-24}+\numprint{1132544}q^{-23}+\numprint{1318912}q^{-22}+\numprint{537600}q^{-21}\\
&+\numprint{2594816}q^{-20}+\numprint{13683712}q^{-19}+\numprint{29159424}q^{-18}+\numprint{36528128}q^{-17}\\
&+\numprint{82216960}q^{-16}+\numprint{285788160}q^{-15}+\numprint{508067840}q^{-14}+\numprint{994416640}q^{-13}+\numprint{2243082240}q^{-12}\\
&+\numprint{3350323200}q^{-11}+\numprint{4204032000}q^{-10}+\numprint{2807562240}q^{-9}+\numprint{655441920}q^{-8}+\numprint{41287680}q^{-7})
             {d \choose 8}\\
&+
             (512q^{-36}+\numprint{36864}q^{-29}+\numprint{129024}q^{-24}+\numprint{1032192}q^{-23}+\numprint{258048}q^{-21}+\numprint{161280}q^{-20}\\
&+\numprint{6193152}q^{-19}+\numprint{14450688}q^{-18}+\numprint{10321920}q^{-17}+\numprint{7741440}q^{-16}+\numprint{118702080}q^{-15}+\numprint{180633600}q^{-14}\\
&+\numprint{216760320}q^{-13}+\numprint{894136320}q^{-12}+\numprint{1083801600}q^{-11}+\numprint{1625702400}q^{-10}+\left.\numprint{867041280}q^{-9}+\numprint{92897280}q^{-8})
             {d \choose 9} \right]
\end{align*}
\begin{align*}
D_{11}(q) &= q^{22d-20}\left[{d\choose
             1}+52q^{-12}+644q^{-11}+\numprint{3530}q^{-10}+\numprint{11772}q^{-9}+\numprint{24472}q^{-8}+\numprint{33336}q^{-7}+\numprint{31202}q^{-6}+\numprint{19532}q^{-5}\right.\\
&+\numprint{8130}q^{-4}+\numprint{2180}q^{-3}+380q^{-2}+36q^{-1})
             {d \choose 2}+
             (78q^{-20}+788q^{-19}+\numprint{4476}q^{-18}+\numprint{27564}q^{-17}+\numprint{134622}q^{-16}\\
&+\numprint{485724}q^{-15}+\numprint{1347336}q^{-14}+\numprint{3077772}q^{-13}
+\numprint{5692422}q^{-12}+\numprint{8616240}q^{-11}+\numprint{10764504}q^{-10}+\numprint{11033940}q^{-9}\\
&+\numprint{9242862}q^{-8}+\numprint{5983548}q^{-7}+\numprint{2711448}q^{-6}+\numprint{799716}q^{-5}+\numprint{148404}q^{-4}+\numprint{16440}q^{-3}+864q^{-2})
             {d \choose 3}\\
&+
             (192q^{-27}+276q^{-26}+192q^{-25}+\numprint{1316}q^{-24}+\numprint{11104}q^{-23}+\numprint{36792}q^{-22}+\numprint{77800}q^{-21}+\numprint{397656}q^{-20}\\
&+\numprint{1661216}q^{-19}+\numprint{5628072}q^{-18}+\numprint{17374664}q^{-17}+\numprint{44682900}q^{-16}+\numprint{98212824}q^{-15}+\numprint{185555192}q^{-14}\\
&+\numprint{304966528}q^{-13}+\numprint{430769888}q^{-12}+\numprint{518306896}q^{-11}+\numprint{517753104}q^{-10}+\numprint{417546016}q^{-9}+\numprint{254764512}q^{-8}\\
&+\numprint{106452816}q^{-7}+\numprint{28179984}q^{-6}+\numprint{4506144}q^{-5}+\numprint{405504}q^{-4}+\numprint{16128}q^{-3})
             {d \choose 4}\\
&+
             (q^{-40}+400q^{-33}+90q^{-32}+\numprint{18000}q^{-28}+\numprint{27560}q^{-27}+\numprint{15000}q^{-26}+\numprint{49840}q^{-25}+\numprint{347280}q^{-24}+\numprint{1166080}q^{-23}\\
&+\numprint{2076920}q^{-22}+\numprint{6344960}q^{-21}+\numprint{21645284}q^{-20}+\numprint{64089560}q^{-19}+\numprint{173957760}q^{-18}+\numprint{442604160}q^{-17}\\
&+\numprint{1003577610}q^{-16}+\numprint{1992112688}q^{-15}+\numprint{3481282200}q^{-14}+\numprint{5323227680}q^{-13}+\numprint{7006036400}q^{-12}\\
&+\numprint{7720824320}q^{-11}+\numprint{6805302112}q^{-10}+\numprint{4493044960}q^{-9}+\numprint{1986467760}q^{-8}+\numprint{536602080}q^{-7}\\
&+\numprint{83668800}q^{-6}+\numprint{6963840}q^{-5}+\numprint{241920}q^{-4})
             {d \choose 5}\\
&+
             (60q^{-41}+\numprint{9120}q^{-34}+\numprint{2520}q^{-33}+\numprint{181440}q^{-29}+\numprint{382800}q^{-28}+\numprint{267840}q^{-27}+\numprint{308640}q^{-26}+\numprint{1925184}q^{-25}\\
&+\numprint{7626720}q^{-24}+\numprint{14772960}q^{-23}+\numprint{27560480}q^{-22}+\numprint{76929360}q^{-21}+\numprint{222560928}q^{-20}+\numprint{553623840}q^{-19}\\
&+\numprint{1314060960}q^{-18}+\numprint{3091960920}q^{-17}+\numprint{6778878720}q^{-16}+\numprint{13156725984}q^{-15}+\numprint{22862414400}q^{-14}\\
&+\numprint{34825843200}q^{-13}+\numprint{45193896960}q^{-12}+\numprint{47500560000}q^{-11}+\numprint{37532963520}q^{-10}+\numprint{19772224320}q^{-9}\\
&+\numprint{6102172800}q^{-8}+\numprint{1026650880}q^{-7}+\numprint{86976000}q^{-6}+\numprint{2903040}q^{-5})
             {d \choose 6}\\
&+
             (560q^{-42}+\numprint{49728}q^{-35}+\numprint{26208}q^{-34}+\numprint{497280}q^{-30}+\numprint{1567104}q^{-29}+\numprint{1831872}q^{-28}+\numprint{822528}q^{-27}\\
&+\numprint{3203200}q^{-26}+\numprint{15201984}q^{-25}+\numprint{43525888}q^{-24}+\numprint{65583168}q^{-23}+\numprint{119482944}q^{-22}+\numprint{324410688}q^{-21}\\
&+\numprint{849534336}q^{-20}+\numprint{1896722688}q^{-19}+\numprint{4014401440}q^{-18}+\numprint{9200648128}q^{-17}+\numprint{19989185664}q^{-16}\\
&+\numprint{38556748160}q^{-15}+\numprint{68671733760}q^{-14}+\numprint{107864158080}q^{-13}+\numprint{141236659200}q^{-12}+\numprint{144177747840}q^{-11}\\
&+\numprint{101179249920}q^{-10}+\numprint{40593853440}q^{-9}+\numprint{8235763200}q^{-8}+\numprint{778982400}q^{-7}+\numprint{27095040}q^{-6})
             {d \choose 7}\\
&+
             (\numprint{1792}q^{-43}+\numprint{93184}q^{-36}+\numprint{105984}q^{-35}+\numprint{387072}q^{-31}+\numprint{2283008}q^{-30}+\numprint{4623360}q^{-29}+\numprint{2460672}q^{-28}\\
&+\numprint{1576960}q^{-27}+\numprint{9397248}q^{-26}+\numprint{46448640}q^{-25}+\numprint{98180096}q^{-24}+\numprint{129253376}q^{-23}+\numprint{194460672}q^{-22}\\
&+\numprint{608183296}q^{-21}+\numprint{1509408768}q^{-20}+\numprint{3134777856}q^{-19}+\numprint{5605476352}q^{-18}+\numprint{13360951296}q^{-17}\\
&+\numprint{29633372160}q^{-16}+\numprint{55336386560}q^{-15}+\numprint{102507991040}q^{-14}+\numprint{171497840640}q^{-13}+\numprint{225982740480}q^{-12}\\
&+\numprint{228707512320}q^{-11}+\numprint{138796277760}q^{-10}+\numprint{39592304640}q^{-9}+\numprint{4696473600}q^{-8}+\numprint{185794560}q^{-7})
             {d \choose 8}\\
&+
             (\numprint{2304}q^{-44}+\numprint{55296}q^{-37}+\numprint{170496}q^{-36}+\numprint{1032192}q^{-31}+\numprint{3833856}q^{-30}+\numprint{4976640}q^{-29}+\numprint{1935360}q^{-27}\\
&+\numprint{8128512}q^{-26}+\numprint{58963968}q^{-25}+\numprint{110380032}q^{-24}+\numprint{97542144}q^{-23}+\numprint{122572800}q^{-22}+\numprint{469905408}q^{-21}\\
&+\numprint{1359880704}q^{-20}+\numprint{2540353536}q^{-19}+\numprint{3491776512}q^{-18}+\numprint{8805888000}q^{-17}+\numprint{23029493760}q^{-16}\\
&+\numprint{37619527680}q^{-15}+\numprint{73086935040}q^{-14}+\numprint{136826081280}q^{-13}+\numprint{175935836160}q^{-12}+\numprint{182798622720}q^{-11}\\
&+\numprint{91457372160}q^{-10}+\numprint{16442818560}q^{-9}+\numprint{836075520}q^{-8})
             {d \choose 9}\\
&+
             (\numprint{1024}q^{-45}+\numprint{92160}q^{-37}+\numprint{368640}q^{-31}+\numprint{3317760}q^{-30}+\numprint{860160}q^{-27}+\numprint{24514560}q^{-25}+\numprint{61931520}q^{-24}\\
&+\numprint{46448640}q^{-22}+\numprint{63866880}q^{-21}+\numprint{557383680}q^{-20}+\numprint{766402560}q^{-19}+\numprint{954777600}q^{-18}+\numprint{1393459200}q^{-17}\\
&+\numprint{8360755200}q^{-16}+\numprint{9165864960}q^{-15}+\numprint{19508428800}q^{-14}+\numprint{43893964800}q^{-13}+\numprint{52022476800}q^{-12}\\
&+\numprint{58525286400}q^{-11}+\left.\numprint{22295347200}q^{-10}+\numprint{1857945600}q^{-9})
             {d \choose 10} \right]
\end{align*}
\begin{align*}
D_{12}(q) &= q^{24d-22}\left[{d\choose
            1}+(9q^{-14}+325q^{-13}+\numprint{2644}q^{-12}\right.+\numprint{12502}q^{-11}+\numprint{38694}q^{-10}+\numprint{79730}q^{-9}+\numprint{114342}q^{-8}+\numprint{115502}q^{-7}\\
&+\numprint{83183}q^{-6}+\numprint{41136}q^{-5}+\numprint{14064}q^{-4}+\numprint{3208}q^{-3}+480q^{-2}+40q^{-1})
            {d \choose 2}
+
  (9q^{-24}+432q^{-22}+\numprint{4668}q^{-21}+\numprint{25440}q^{-20}\\
&+\numprint{138904}q^{-19}+\numprint{620231}q^{-18}+\numprint{2097936}q^{-17}+\numprint{5926745}q^{-16}+\numprint{13865948}q^{-15}+\numprint{27402318}q^{-14}+\numprint{45459498}q^{-13}\\
&+\numprint{63704774}q^{-12}+\numprint{75606576}q^{-11}+\numprint{75472992}q^{-10}+\numprint{62723968}q^{-9}+\numprint{41798098}q^{-8}+\numprint{20741808}q^{-7}+\numprint{7158960}q^{-6}\\
&+\numprint{1661832}q^{-5}+\numprint{252468}q^{-4}+\numprint{23640}q^{-3}+\numprint{1080}q^{-2})
  {d \choose 3}+
  (160q^{-30}+\numprint{1424}q^{-29}+\numprint{3696}q^{-28}+\numprint{7056}q^{-27}+\numprint{25980}q^{-26}\\
&+\numprint{117640}q^{-25}+\numprint{429416}q^{-24}+\numprint{1373536}q^{-23}+\numprint{5256924}q^{-22}+\numprint{19307064}q^{-21}+\numprint{60129640}q^{-20}+\numprint{165679408}q^{-19}\\
&+\numprint{401149300}q^{-18}+\numprint{851344600}q^{-17}+\numprint{1602760536}q^{-16}+\numprint{2672842064}q^{-15}+\numprint{3949417120}q^{-14}+\numprint{5116018552}q^{-13}\\
&+\numprint{5745649000}q^{-12}+\numprint{5469111376}q^{-11}+\numprint{4290529800}q^{-10}+\numprint{2640845536}q^{-9}+\numprint{1183404144}q^{-8}+\numprint{361370880}q^{-7}\\
&+\numprint{72248256}q^{-6}+\numprint{9134016}q^{-5}+\numprint{678144}q^{-4}+\numprint{23040}q^{-3})
  {d \choose 4}
  (40q^{-41}+10q^{-40}+\numprint{4800}q^{-35}+\numprint{6840}q^{-34}+\numprint{3680}q^{-33}\\
&+450q^{-32}+\numprint{28640}q^{-31}+\numprint{183400}q^{-30}+\numprint{587760}q^{-29}+\numprint{651000}q^{-28}+\numprint{1454640}q^{-27}+\numprint{5843560}q^{-26}+\numprint{19181720}q^{-25}\\
&+\numprint{47530750}q^{-24}+\numprint{127896240}q^{-23}+\numprint{388723120}q^{-22}+\numprint{1084985280}q^{-21}+\numprint{2757696960}q^{-20}+\numprint{6446061120}q^{-19}\\
&+\numprint{13767375360}q^{-18}+\numprint{26497791360}q^{-17}+\numprint{45857136550}q^{-16}+\numprint{71071270720}q^{-15}+\numprint{97786149160}q^{-14}+\numprint{117432352720}q^{-13}\\
&+\numprint{119939898280}q^{-12}+\numprint{100154952320}q^{-11}+\numprint{64752425920}q^{-10}+\numprint{29873551040}q^{-9}+\numprint{9140193360}q^{-8}+\numprint{1772429760}q^{-7}\\
&+\numprint{209425920}q^{-6}+\numprint{13916160}q^{-5}+\numprint{403200}q^{-4}) {d \choose 5}+(12q^{-50}+\numprint{3540}q^{-42}+720q^{-41}+\numprint{135840}q^{-36}+\numprint{245280}q^{-35}\\
&+\numprint{129900}q^{-34}+\numprint{18000}q^{-33}+\numprint{399360}q^{-32}+\numprint{2483520}q^{-31}+\numprint{9367248}q^{-30}+\numprint{12518880}q^{-29}+\numprint{16035360}q^{-28}+\numprint{47697600}q^{-27}\\
&+\numprint{159410700}q^{-26}+\numprint{393989328}q^{-25}+\numprint{827877600}q^{-24}+\numprint{2007451200}q^{-23}+\numprint{5289092160}q^{-22}+\numprint{12968879520}q^{-21}\\
&+\numprint{29810834496}q^{-20}+\numprint{65139117600}q^{-19}+\numprint{134353485960}q^{-18}+\numprint{253569485520}q^{-17}+\numprint{432870902880}q^{-16}\\
&+\numprint{664628651328}q^{-15}+\numprint{901917951120}q^{-14}+\numprint{1052925072960}q^{-13}+\numprint{1011974394240}q^{-12}+\numprint{752224121280}q^{-11}\\
&+\numprint{395862472320}q^{-10}+\numprint{134197906560}q^{-9}+\numprint{27675947520}q^{-8}+\numprint{3329579520}q^{-7}+\numprint{215424000}q^{-6}+\numprint{5806080}q^{-5})
  {d \choose 6}\\
&+
  (280q^{-51}+\numprint{43680}q^{-43}+\numprint{14000}q^{-42}+\numprint{853440}q^{-37}+\numprint{2101344}q^{-36}+\numprint{1674792}q^{-35}+\numprint{307440}q^{-34}+\numprint{1379840}q^{-33}\\
&+\numprint{9625280}q^{-32}+\numprint{40944960}q^{-31}+\numprint{83523776}q^{-30}+\numprint{88159680}q^{-29}+\numprint{138491584}q^{-28}+\numprint{435093680}q^{-27}+\numprint{1277565520}q^{-26}\\
&+\numprint{2659260352}q^{-25}+\numprint{4927060096}q^{-24}+\numprint{11261474784}q^{-23}+\numprint{27295013760}q^{-22}+\numprint{62106803584}q^{-21}+\numprint{133111109600}q^{-20}\\
&+\numprint{278163144112}q^{-19}+\numprint{570015161744}q^{-18}+\numprint{1080480157120}q^{-17}+\numprint{1863030926400}q^{-16}+\numprint{2912637247520}q^{-15}\\
&+\numprint{4005783848640}q^{-14}+\numprint{4636680236160}q^{-13}+\numprint{4246497279360}q^{-12}+\numprint{2790938075520}q^{-11}+\numprint{1162070716800}q^{-10}\\
&+\numprint{280070945280}q^{-9}+\numprint{37219553280}q^{-8}+\numprint{2522419200}q^{-7}+\numprint{67737600}q^{-6})
  {d \choose 7}+
  (\numprint{1792}q^{-52}+\numprint{173824}q^{-44}+\numprint{107520}q^{-43}\\
&+\numprint{1784832}q^{-38}+\numprint{6408192}q^{-37}+\numprint{8567552}q^{-36}+\numprint{2626560}q^{-35}+\numprint{1290240}q^{-34}+\numprint{13275136}q^{-33}+\numprint{63114240}q^{-32}\\
&+\numprint{206603264}q^{-31}+\numprint{288296960}q^{-30}+\numprint{272971776}q^{-29}+\numprint{464708864}q^{-28}+\numprint{1649027072}q^{-27}+\numprint{4405051392}q^{-26}\\
&+\numprint{7979697152}q^{-25}+\numprint{13162963968}q^{-24}+\numprint{28913826816}q^{-23}+\numprint{68086100992}q^{-22}+\numprint{148218871808}q^{-21}+\numprint{298732675584}q^{-20}\\
&+\numprint{594246754304}q^{-19}+\numprint{1242503485440}q^{-18}+\numprint{2377444643840}q^{-17}+\numprint{4150496215040}q^{-16}+\numprint{6719884738560}q^{-15}\\
&+\numprint{9523134351360}q^{-14}+\numprint{11000978257920}q^{-13}+\numprint{9680191180800}q^{-12}+\numprint{5522523955200}q^{-11}+\numprint{1746670141440}q^{-10}\\
&+\numprint{282402570240}q^{-9}+\numprint{21676032000}q^{-8}+\numprint{619315200}q^{-7})
  {d \choose 8}\\
&+
  (\numprint{4608}q^{-53}+\numprint{267264}q^{-45}+\numprint{345600}q^{-44}+\numprint{1161216}q^{-39}+\numprint{7621632}q^{-38}+\numprint{17339904}q^{-37}+\numprint{10391040}q^{-36}\\
&+\numprint{5677056}q^{-34}+\numprint{34191360}q^{-33}+\numprint{176394240}q^{-32}+\numprint{431253504}q^{-31}+\numprint{500023296}q^{-30}+\numprint{320606208}q^{-29}\\
&+\numprint{687278592}q^{-28}+\numprint{2871687168}q^{-27}+\numprint{7684411392}q^{-26}+\numprint{12036003840}q^{-25}+\numprint{17087035392}q^{-24}+\numprint{35552047104}q^{-23}\\
&+\numprint{89409696768}q^{-22}+\numprint{190096155648}q^{-21}+\numprint{359983378944}q^{-20}+\numprint{656605587456}q^{-19}+\numprint{1461124546560}q^{-18}\\
&+\numprint{2856624906240}q^{-17}+\numprint{4913585510400}q^{-16}+\numprint{8413271838720}q^{-15}+\numprint{12412241694720}q^{-14}+\numprint{14290489221120}q^{-13}\\
&+\numprint{12206985154560}q^{-12}+\numprint{5862793789440}q^{-11}+\numprint{1322578575360}q^{-10}+\numprint{127455068160}q^{-9}+\numprint{4180377600}q^{-8})
  {d \choose 9}\\
&+
  (\numprint{5120}q^{-54}+\numprint{138240}q^{-46}+\numprint{471040}q^{-45}+\numprint{2949120}q^{-39}+\numprint{12211200}q^{-38}+\numprint{17510400}q^{-37}+\numprint{6451200}q^{-34}\\
&+\numprint{27402240}q^{-33}+\numprint{209756160}q^{-32}+\numprint{460800000}q^{-31}+\numprint{357565440}q^{-30}+\numprint{96768000}q^{-29}+\numprint{472227840}q^{-28}\\
&+\numprint{1977507840}q^{-27}+\numprint{7057612800}q^{-26}+\numprint{9483264000}q^{-25}+\numprint{9323274240}q^{-24}+\numprint{20261928960}q^{-23}+\numprint{57411164160}q^{-22}\\
&+\numprint{130223493120}q^{-21}+\numprint{223962439680}q^{-20}+\numprint{354261196800}q^{-19}+\numprint{854235648000}q^{-18}+\numprint{1861919539200}q^{-17}\\
&+\numprint{2852694835200}q^{-16}+\numprint{5438206771200}q^{-15}+\numprint{8347594752000}q^{-14}+\numprint{9538537881600}q^{-13}+\numprint{8018893209600}q^{-12}\\
&+\numprint{3089763532800}q^{-11}+\numprint{442191052800}q^{-10}+\numprint{18579456000}q^{-9})
  {d \choose 10}\\
&+
  (\numprint{2048}q^{-55}+\numprint{225280}q^{-46}+\numprint{1013760}q^{-39}+\numprint{10137600}q^{-38}+\numprint{2703360}q^{-34}+\numprint{81100800}q^{-32}+\numprint{248033280}q^{-31}\\
&+\numprint{2838528}q^{-30}+\numprint{189235200}q^{-28}+\numprint{141926400}q^{-27}+\numprint{2838528000}q^{-26}+\numprint{3406233600}q^{-25}+\numprint{567705600}q^{-24}\\
&+\numprint{5464166400}q^{-23}+\numprint{11330457600}q^{-22}+\numprint{40874803200}q^{-21}+\numprint{54159114240}q^{-20}+\numprint{80188416000}q^{-19}\\
&+\numprint{166053888000}q^{-18}+\numprint{562028544000}q^{-17}+\numprint{585872179200}q^{-16}+\numprint{1445946163200}q^{-15}+\numprint{2253223526400}q^{-14}\\
&+\numprint{2554675200000}q^{-13}+\numprint{2145927168000}q^{-12}+\numprint{613122048000}q^{-11}+\left.\numprint{40874803200}q^{-10}) {d \choose 11} \right]
\end{align*}
\begin{align*}
D_{13}(1) 
&= {d \choose 1} 
  + \numprint{1903888} {d \choose 2}
  + \numprint{3317057654} {d \choose 3}
  + \numprint{408413968572} {d \choose 4}
  + \numprint{12036577000605} {d \choose 5}
  + \numprint{138284024035957} {d \choose 6} \\
  &+ \numprint{779444454950084} {d \choose 7}
  + \numprint{2417940914461280} {d \choose 8}
  + \numprint{4336107249936384} {d \choose 9}
  + \numprint{4477975127425280} {d \choose 10} \\
  &+ \numprint{2471677136321536} {d \choose 11}
  + \numprint{564668382613504} {d \choose 12}
\end{align*} 
}

\section{The Polynomials $b_r$ for Site Percolation}
\label{sec:site-br}

Here we give $b_r$, the coefficients of $p^r$ in the series expansion for the expected cluster size $S$ as defined in~\eqref{eq:br-site}. The leading term $\sigma^r$ corresponds to the mean-field behavior.
{\small
\begin{align*}
b_0 &= 1 \\
b_1 &= \sigma+1 \\
b_2 &= \sigma^2 + \sigma \\
b_3 &= \sigma^3 - \frac{\sigma^2}{2} + \frac{3}{2} \\
b_4 &= \sigma^4 - 2 \sigma^3 + 2 \sigma^2 + 3 \sigma - 2 \\
b_5 &= \sigma^5 - \frac{7 \sigma^4}{2} + \frac{7 \sigma^3}{2} + \frac{15 \sigma^2}{2} - \frac{7 \sigma}{2} - 3 \\
b_6 &= \sigma^6 - 5 \sigma^5 + \frac{29 \,\sigma^4}{4} + 18 \,\sigma^3 - 40 \,\sigma^2 - 12 \,\sigma + \frac{131}{4} \\
b_7 &= \sigma^7-\frac{13 \,\sigma^6}{2}+\frac{53 \,\sigma^5}{4}+\frac{7 \,\sigma^4}{4}+\frac{104 \,\sigma^3}{3}-\frac{299 \,\sigma^2}{2}-\frac{575 \,\sigma}{12}+\frac{621}{4} \\
b_8 &= \sigma^8-8 \,\sigma^7+\frac{43 \,\sigma^6}{2}-\frac{85 \,\sigma^5}{4}+\frac{3391 \,\sigma^4}{12}-\frac{4715 \,\sigma^3}{3}+\frac{\numprint{12235} \,\sigma^2}{6}+\frac{\numprint{19223} \,\sigma}{12}-\frac{9373}{4} \\
b_9 &= \sigma^9-\frac{19 \,\sigma^8}{2}+32 \,\sigma^7-\frac{435 \,\sigma^6}{8}+\frac{1217 \,\sigma^5}{4}-\frac{\numprint{21257} \,\sigma^4}{24}-\frac{\numprint{22645} \,\sigma^3}{12}+\frac{\numprint{170963} \,\sigma^2}{24}+\frac{9305 \,\sigma}{6}-\frac{\numprint{49383}}{8} \\
b_{10} &= \sigma^{10}-11 \,\sigma^9+\frac{179 \,\sigma^8}{4}-101 \,\sigma^7+\frac{2275 \,\sigma^6}{6}+\frac{\numprint{24401} \,\sigma^5}{12}-\frac{\numprint{109642} \,\sigma^4}{3}+\frac{\numprint{849277} \,\sigma^3}{6}-\frac{\numprint{431164} \,\sigma^2}{3} \\
& -\frac{\numprint{1721599} \,\sigma}{12}+\frac{\numprint{719379}}{4} \\
b_{11} &= \sigma^{11}-\frac{25 \,\sigma^{10}}{2}+\frac{239 \,\sigma^9}{4}-\frac{329 \,\sigma^8}{2}+\frac{\numprint{12581} \,\sigma^7}{24}+\frac{26317 \,\sigma^6}{24}-\frac{\numprint{407843} \,\sigma^5}{24}-\frac{\numprint{407117} \,\sigma^4}{12}+\frac{\numprint{11823737} \,\sigma^3}{24} \\
& -\frac{\numprint{6505119} \,\sigma^2}{8}-\frac{\numprint{11429909} \,\sigma}{24}+\frac{\numprint{3384591}}{4} \\
b_{12} &= \sigma^{12}-14 \,\sigma^{11}+77 \,\sigma^{10}-\frac{993 \,\sigma^9}{4}+\frac{12181 \,\sigma^8}{16}-\frac{103 \,\sigma^7}{3}+\frac{\numprint{1015339} \,\sigma^6}{24}-\frac{\numprint{11977951} \,\sigma^5}{12}+\frac{\numprint{20826482} \,\sigma^4}{3} \\ 
&-\frac{\numprint{117862115} \,\sigma^3}{6}+\frac{\numprint{366648613} \,\sigma^2}{24}+\numprint{20642146} \,\sigma -\frac{\numprint{356197285}}{16}
\end{align*}
}

\ack{%
  We are indebted to Sebastian Luther for modifying the enumeration
  code to include the vertex count in bond animals, to Remco van der
  Hofstad for helpful pointers to the mathematical literature, and to
  Bob Ziff for stimulating discussions. SM thanks the Santa Fe
  Institute for their hospitality and financial support.  }

\section*{References}

\bibliographystyle{unsrt} 
\bibliography{animals,math,percolation,mertens}

\end{document}